%% file: neurips_2025.tex
\documentclass{article}

    \usepackage[preprint]{neurips_2025}

\input{aux/math_commands}

\usepackage[compact]{titlesec}
\allowdisplaybreaks
\titlespacing*{\section}
  {0pt}    
  {0pt}    
  {0pt}    
\usepackage{diagbox}  
\usepackage{tikz}
\newlength{\BoxH}
\usepackage{subcaption}  
\usepackage{amsmath}
\usepackage{enumitem}
\usepackage{bm}
\usepackage{amssymb}
\usepackage{amsthm}
\usepackage{tabularx}
\usepackage{float}
\usepackage[most]{tcolorbox}
\usepackage[utf8]{inputenc} 
\usepackage[T1]{fontenc}    
\usepackage{hyperref}       
\usepackage{url}            
\usepackage{booktabs}       
\usepackage{amsfonts}       
\usepackage{nicefrac}       
\usepackage{microtype}      
\usepackage{xcolor}         
\newtheorem{theorem}{Theorem}
\newtheorem{definition}[theorem]{Definition}

\newtheorem{proposition}[theorem]{Proposition}
\newtheorem{lemma}[theorem]{Lemma}
\newtheorem{example}[theorem]{Example}
\newtheorem{corollary}[theorem]{Corollary}

\usepackage{hyperref}
\usepackage[titlenumbered,ruled,linesnumbered]{algorithm2e}
\title{Optimal Algorithms for Bandit Learning in \\Matching Markets}

\author{%
  Tejas Pagare\thanks{Correspondence to Tejas Pagare: tejaspagare2002@gmail.com.} \\
  Carnegie Mellon University \\
  Pittsburgh \\
  \And
  Agniv Bandyopadhyay \\
  TIFR \\
  Mumbai \\
  \And
  Sandeep Juneja \\
  Ashoka University \\
  Sonipat \\
}

\begin{document}

\maketitle

\begin{abstract}
We study the problem of pure exploration in matching markets under uncertain preferences, where the goal is to identify a stable matching with confidence parameter $\delta$ and minimal sample complexity. Agents learn preferences via stochastic rewards, with expected values indicating preferences. This finds use in labor market platforms like Upwork, where firms and freelancers must be matched quickly despite noisy observations and no prior knowledge, in a stable manner that prevents dissatisfaction. We consider markets with unique stable matching and establish information-theoretic lower bounds on sample complexity for (1) one-sided learning, where one side of the market knows its true preferences, and (2) two-sided learning, where both sides are uncertain. We propose a computationally efficient algorithm and prove that it asymptotically ($\delta\to 0$) matches the lower bound to a constant for one-sided learning. Using the insights from the lower bound, we extend our algorithm to the two-sided learning setting and provide experimental results showing that it closely matches the lower bound on sample complexity. Finally, using a system of ODEs, we characterize the idealized fluid path that our algorithm chases.
\end{abstract}

\input{paper/introduction}
\input{paper/preliminaries}

\input{paper/one-sided}
\input{paper/two-sided}

\input{paper/experiments}

\input{paper/conclusion}


\clearpage
\bibliography{aux/ref}
\bibliographystyle{abbrv}

\input{paper/appendix/appendix-final}

\end{document}

%% file: aux/math_commands.tex
\usepackage{amsmath,amsfonts,bm}
\usepackage{longtable}

\def\eqref#1{equation~\ref{#1}}

\def\1{\bm{1}}

\DeclareMathAlphabet{\mathsfit}{\encodingdefault}{\sfdefault}{m}{sl}
\SetMathAlphabet{\mathsfit}{bold}{\encodingdefault}{\sfdefault}{bx}{n}

\newcommand{\UMA}{\mathcal{A}^{\texttt{UM}}}
\newcommand{\stable}{\texttt{stable}}
\newcommand{\E}{\mathbb{E}}

\newcommand{\R}{\mathbb{R}}

\DeclareMathOperator*{\argmax}{arg\,max}

%% file: paper/introduction.tex
\section{Introduction}
With the advent of data in online marketplaces, the design of efficient matching algorithms incorporating statistical uncertainty has become increasingly important. We study matching in a two-sided market problem where preferences are uncertain and need to be inferred through interaction. Uncertainty in preferences is inherent in online labor markets, where the central decision maker must match workers and employers despite incomplete information. For workers, this uncertainty arises from varying incentives, personal familiarity with tasks, and differences in job descriptions, and employers face uncertainty about the actual skill levels and suitability of workers. Abstracting this uncertainty using a multi-armed bandit framework has been a recent study of interest \cite{liu2020competing,soda,shah}. Therein, agents on either side of the market compete with agents on the same side to match with the other side, receiving stochastic reward signaling the preference. Each agent prefers an agent on the other side, which, in expectation, provides larger utility. The objective is to achieve a notion of equilibria in the matching market, a stable matching,  where no two agents from different sides prefer each other over their respective match. 

Recent studies have studied this problem in a minimization of regret \cite{liu2020competing}, where the objective is to minimize the general regret incurred by each agent by not matching to the stable match. However, algorithms devised in the regret minimization setting may lack practical applicability, as these algorithms constantly alternate between matching, which is often impractical and can lead to instability. We thus study the problem of pure exploration in this setting, where the goal is to identify a stable matching as quickly as possible. The pure exploration problem, more formally, aims to minimize the interaction rounds, so that the estimated stable match is incorrect with, at most $\delta$, a pre-specified confidence parameter. 

In the pure-exploration literature, algorithms which are asymptotically optimal, i.e., as $\delta$ tends to 0, the algorithm samples exactly as specified by the information-theoretic lower bound on the sample complexity, have been of wide interest, especially in the specialized best-arm-identification (\texttt{BAI}) setting. Here, asymptotically optimal algorithms also tend to perform well in practice for practically reasonable $\delta$ values. However, existing popular algorithms \cite{garivier2016optimal} involve repeated solving an optimization problem and tend to be computationally expensive. Thus, we ask the question \textit{ can we design a computationally efficient and asymptotically optimal algorithm for stable matching identification with pure exploration?}

\textit{Pure exploration problem}:  This problem has been studied in various forms in the bandit literature. \texttt{BAI} \cite{garivier2016optimal} and top best-arm identification  \cite{you2023information} involves identifying the arm(s) with largest expected reward. The partition-identification problem studied in \cite{juneja2019sample} generalizes this to identifying if the uncertain parameter lies in some region. A substantial amount of interest in the learning theory community is in designing algorithms that match the lower-bound sample complexity. As indicated earlier, \cite{garivier2016optimal} first introduced the \textit{Track-and-Stop} algorithm that iteratively obtains the optimal proportion by solving the lower bound optimization problem, and aims to track it. Although it achieves the desired sample complexity for small $\delta$, it is computationally inefficient. Algorithms referred to as $\beta$ top-two algorithms are computationally efficient \cite{russo2020simple,jourdan2022top}, which at any time randomizes between the empirically best arm and the challenger arm, with the pre-specified probability $\beta$. Although computationally faster, these algorithms are typically only order-optimal in that they do not match the constant in the lower bound. \cite{bandyopadhyay2024optimal} recently proposed a deterministic scheme, which decides on selecting the best or the challenger arm based on the sign of an anchor function. The algorithm is shown to match even the constant term in the lower bound as $\delta \rightarrow 0$.

\textit{Matching Bandit problem}: \cite{liu2020competing} aimed at developing a regret minimization scheme to achieve player-optimal matching. In regret minimization, several extensions have been studied, decentralized algorithms in which agents make decisions only based on local information \cite{liu2021decentralized,maheshwari2022decentralized,pmlr-v139-basu21a}, two-sided learning in which both agents are unaware of preferences \cite{pagaretwosided}, and time-varying preferences \cite{min2022learn,ghosh}. Although logarithmic regret performance similar to the lower bound has been obtained, these algorithms tend to have large sample complexity. Recently, \cite{hosseini2024putting} designed an algorithm that proceeds as in the Deferred Acceptance algorithm, which is heavily based on the knowledge of preferences of one side, and their complexity does not match the lower bound and involves additional logarithmic factors. 

\textit{One-Sided and Two-Sided learning}: As in the existing literature, we study two problem settings: one-sided learning in which one side knows the preferences, and two-sided learning where both sides are unaware of their preferences. Based on practical applicability, either a one-sided or a two-sided learning model is valid. For example, in school admissions with specialized programs, universities have historical data on student performance, but students may not know their preferences for different universities or majors before enrollment. On the other hand, in online dating, both individuals are initially unaware of their compatibility unless they start interacting. 

\textit{Pure exploration vs Regret minimization}: Online labor markets like UpWork facilitate the matching of freelancers and employers, typically for short-term contracts. However, employers can extend contracts if the employer gains confidence in a particular worker. In the pure exploration setting, the employer aims to efficiently identify the best worker as quickly as possible, allowing for more extended or more frequent contracts once confidence is established. In contrast, the regret minimization setting involves a trade-off between exploiting the best-known option and continuing exploration to gather more information, leading to frequent match changes.

Our key {\bf contributions} are as follows
\begin{itemize}[leftmargin=*, nosep,noitemsep, topsep=0pt]
\item \textbf{Theoretical Foundations:} We establish the first information-theoretic lower bounds on the sample complexity of stable matching market identification problem. The bound characterizes the difficulty of markets with a unique stable matching, and we discuss the additional challenges posed by multiple stable matchings in the appendix~\ref{app:multistable}.

\item \textbf{Novel Algorithms:} We design two practical Top-Two-style algorithms: (i) the \texttt{Anchored Top-Two (ATT)} algorithm, which leverages the Karush-Kuhn-Tucker (\texttt{KKT}) conditions of the optimal allocation, and (ii) the $\beta$-Top-Two algorithm, applicable to both one-sided and two-sided learning settings.

\item \textbf{Algorithmic Insights:} We provide a rigorous analysis of our methods, including a fluid-dynamic interpretation that offers intuition into their behavior. We also prove that the \texttt{ATT} algorithm is asymptotically optimal for one-sided learning.

\item \textbf{Empirical Validation:} We conduct extensive experiments across various market instances with unique stable matchings. Our results validate theoretical guarantees and highlight performance gains of \texttt{ATT} over $\beta$-Top-Two algorithms.
\end{itemize}

\textit{Relevance of fluid dynamics}: Analyzing the proposed algorithms requires dealing with the noise in the samples and also the granularity of the samples at the same time. Convergence to the optimal allocation has to be shown considering all these difficulties. To get a clearer intuition of how the allocation evolves for the algorithm, we consider the fluid analysis where:~1)~ignore the noise present in the samples by assuming that all the means are known, and~2)~allocating samples across different player-arm pair is treated as filling up sand in different buckets. In this paper, we do not prove convergence of the non-fluid dynamics of the algorithm to the fluid dynamics, which might be technically challenging and is also beyond the scope of our paper. Instead, we use the fluid dynamics to derive intuitions which later help us in proving important results concerning the dynamics of the algorithm in the non-fluid regime. For instance, in the one-sided setting, we use intuitions obtained from the fluid dynamics to prove convergence of minimum indexes of all the players. In the two-sided setting, the fluid dynamics provide a heuristic argument for asymptotic optimality of the proposed algorithm. The essence of the argument is that if we ignore the noise in the empirical means and granularity in the sample allocations, then the algorithm does converge to the optimal allocation in a well-defined manner following the solution of a system of ODEs pasted together.

\textit{Roadmap}: In Section 2, we provide preliminaries to the pure exploration problem and details on the market instances with unique stable matching. We cover one-sided learning in Section 3, where we first describe the lower bound, our Top-Two algorithm, fluid dynamics, and a proof overview of asymptotic optimality. Section 4 covers Two-Sided learning, describing the lower bound, our algorithm, and a prelude to the fluid dynamics. In Section 5, we provide empirical experiments, and in Section 6, we conclude with a summary and open problems.

\textit{Limitations}:~The algorithms proposed in this paper are designed specifically for instances having unique stable matching, which is generally not true. Furthermore, our analysis is asymptotic. As is standard in the literature, we assume the rewards between different pairs of players' arms to be independent. Another limitation is our assumption that reward distributions are stationary, which is standard when long-horizons are involved. 

%% file: paper/preliminaries.tex
\section{Preliminaries}\label{sec:preliminaries}

We consider a two-sided matching market in which on one side of the market we have players $\mathcal{P}=\{p_1,\ldots,p_M\}$ and on the other side arms $\mathcal{A}=\{a_1,\ldots,a_K\}$.
A preference profile, denoted by $\zeta$, describes the preference of each player over arms and each arm over players. A stable matching under $\zeta$ is a bijective function $m:\mathcal{P}\to\mathcal{A}$ such that there exists no pair $(p_i,a_k), a_k\neq a_{m(i)}$ such that $a_k\succ_{p_i}a_{m(i)}$ and $p_i\succ_{a_k}p_{m^{-1}(k)}$, i.e. $p_i$ and $a_k$ prefer each other over their respective match, called a blocking pair.

We consider the setting when $\zeta$ is unknown and must be learned, for which we propose two learning mechanisms (1) one-sided learning where the arms are aware of their preference profile and (2) two-sided learning, where both the players and the arms are unaware of their preference profile. Each player learns the preference of each arm through a specific reward distribution, which signals the arms' preference for this player. In particular, let $X^{(i)}_k,X^{(i)}_l$ denote the random reward variable corresponding to $(p_i,a_k)$ and $(p_i,a_{l})$, respectively, we say $a_k\succ_{p_{i}}a_l\iff \E[X^{(i)}_k]>\E[X^{(i)}_l]$. Similarly, for arms, let $Y^{(i)}_k,Y^{(j)}_k$ be reward a random variable corresponding to $(p_i,a_k)$ and $(p_j,a_{k})$, respectively, we say $p_i\succ_{a_k}p_j\iff \E[Y^{(i)}_k]>\E[Y^{(j)}_k]$. We assume for simplicity that $X^{(i)}_k,Y^{(i)}_k,X^{(i)}_{l},Y^{(j)}_{k}\ \forall i\neq j, k\neq l$ are independent leading to a private learning mechanism for each player and arm. More sophisticated learning rules with dependence of learning across players and arms are left for future discussion. 

We assume reward distributions belong to a Single-Parameter Exponential Family (SPEF), with density $f_{\theta}(x)=\exp(\theta x-b(\theta))$ for a parameter $\theta \in \Theta \subseteq \mathbb{R}$, where $b(\theta)$ is the log-partition function. A key property of this family is the unique mapping between the parameter $\theta$ and the mean reward $\mu = \dot{b}(\theta)$. We thus define the distribution with mean $\mu$ by $\texttt{SPEF}(\mu)$. For any two means $\mu_1, \mu_2$, the Kullback-Leibler (KL) divergence, which quantifies the difficulty of distinguishing between the two underlying distributions, is given by $D(\mu_1,\mu_2)=(\theta_{\mu_1}-\theta_{\mu_2})\mu_1- b(\theta_{\mu_1})+b(\theta_{\mu_2})$. While we focus on this family, our algorithms can be extended to any distribution with bounded support, as detailed in \cite[Appendix I]{bandyopadhyay2024optimal}.

We now describe the interaction protocol for the central decision maker, which at each instant makes the decision of which player $p_i\in \mathcal{P}$ and arm $a_k\in\mathcal{A}$ to match based on history up to that time denoted by the $\sigma$-algebra $\mathcal{F}_{t-1}:=\sigma\{m_{1},\ldots,m_{t-1},X^{(i)}_{k,1},Y^{(i)}_{k,1},\ldots,X^{(i)}_{k,t-1},Y^{(i)}_{k,t-1}\}$, which contains reward feedback and match taken. After matching $(p_i,a_k)$, the decision maker receives $X^{(i)}_k\sim\texttt{SPEF}(\mu^{(i)}_k)$ and if two-sided learning, it also receives $Y^{(i)}_k\sim\texttt{SPEF}(\eta^{(i)}_k)$. This constitutes a one-round and the goal is to design a $\delta$-correct algorithm, with minimum number of rounds. 
\begin{definition}
    An algorithm is said to be $\delta$-correct for the matching bandit problem, if for every preference profile $\zeta$, for any specified $\delta\in(0,1)$, it restricts the probability of announcing the unstable match to at most $\delta$, i.e. $\mathbb{P}_{\zeta}(m\not\in\mathcal{M}_{\zeta})\leq \delta$ where $m$ is the announced match, and $\mathcal{M}_{\zeta}$ is the set of stable matching under the instance $\zeta$. 
\end{definition}
We restrict the scope of our current work to preference profiles $\zeta$ with unique stable matching. Finding a necessary and sufficient condition for an instance that has unique stable matching is an open problem and has been provided for a specific case in \cite{KARPOV201963}, so we do not impose a property on $\zeta$. The well-known conditions are (a) Serial Dictatorship: if agents on some side have the same preference ordering for all agents on the other side, (b)  Sequential preference condition: there exists a ordering of players and arms s.t. $p_i$ prefers $a_i$ over $a_{i+1},\ldots,a_K$ and $a_j$ prefers $p_j$ over $p_{j+1},\ldots,p_N$ and (c) $\alpha$-condition: there exists a ordering of players and arms s.t. $p_i$ prefers its match over $a_{i+1},\ldots,a_K$ and possibly different ordering of player and arms such that $a_j$ prefers its match over $p_{j+1},\ldots,p_N$. \cite{KARPOV201963} presents an algorithm to verify the $\alpha-$condition.
We emphasize that the decision maker is only aware that $\zeta$ has a unique stable match and does not know if it has any of the above mentioned properties. For markets with multiple stable matching, one could consider a different objective that aims at finding a specific stable matching, e.g. the player/arm-optimal or fair matching. We comment more on the lower bound in this case in the Appendix J.

Note that we assume unique stable matching as compared to several works which assume stronger properties such as serial dictatorship \cite{sankararaman2021dominate}, $\alpha$-condition \cite{pmlr-v139-basu21a}, $\alpha$-reducibility \cite{maheshwari2022decentralized}. Such assumption is merely required to simplify the lower bound, as in multiple stable matching the lower bound involves outer maximization over all stable matchings, which is a combinatorial hard discrete problem and additionally depends on the objective of finding stable matching or a specific matching. This does not arise in the regret minimization setting where the regret is benchmarked only against the unique stable matching e.g. pessimal/optimal stable matching. There are tradeoffs, of course, assuming unique stable matching allowed us to exactly match the constant. Morever, such assumption is also motivated by the result of \cite{ashlagi2017unbalanced} where unbalanced markets, i.e. market with unequal number of players and arms yield unique stable matching with high probability under random preferences. Our algorithms works even for unbalanced setting.

%% file: paper/one-sided.tex
\section{One-sided Learning}
Recall, in one-sided learning, each arm $a_k\in \mathcal{A}$ knows the preference over the players and denote $p_i\succ_{a_k}p_j$ if $a_k$ prefers $p_i$ over $p_j$. Let $\mathcal{S}_1$ denote the set of instances $\mu := (\mu^{(i)}_k)_{(p_i,a_k) \in \mathcal{P} \times \mathcal{A}}$ such that the preference profile induced by $(\mu, \succ)$ admits a unique stable matching, denoted by $m$. We first derive the lower bound of any $\delta$-correct algorithm  in one-sided learning.  Using the data-processing inequality \cite{kaufmann2020contributions}, we have
\begin{align}
\sum_{p_i\in\mathcal{P}}\sum_{a_k\in\mathcal{A}}\E_{\mu}[N^{(i)}_k] D\left(\mu^{(i)}_k,\lambda^{(i)}_k\right)\geq D(\delta,1-\delta)\geq \log(1/(2.4\delta))
\end{align}
where $N^{(i)}_k$ denotes the number of times $p_i$ and $a_k$ are matched and $\lambda\in \texttt{Alt}_{\mu}:=\{\lambda:m\not\in\mathcal{M}_{\lambda}\}$, the alternate set. For $t^{(i)}_k = \E_{\mu}[N^{(i)}_k]/\log(1/(2.4\delta)) $, the lower bound sample complexity problem can be equivalently modeled by the following problem
\begin{align*}
  \texttt{\textbf{LO1}:} \quad \min_{\mathbf{t}=(t^{(i)}_k)_{(p_i,a_k)\in\mathcal{P}\times\mathcal{A}}} \sum_{p_i\in\mathcal{P}}\sum_{a_k\in\mathcal{A}} t^{(i)}_k \quad\text{s.t.}\ t^{(i)}_k\geq 0\ \text{and}\ 
\inf_{\lambda\in\texttt{Alt}_{\mu}}\sum_{p_i\in\mathcal{P}}\sum_{a_k\in\mathcal{A}}t^{(i)}_kD\left(\mu^{(i)}_k,\lambda^{(i)}_k\right)\geq 1
\end{align*}

To characterize the alternate instance, we refer to the definition of a blocking pair under matching stable $m$ from the preliminaries where the alternate instance is characterized as an instance under which $m$ creates a blocking pair. Thus, the alternate set simplifies to $\neg m:=\cup_{p_i}\cup_{a_k\in\mathcal{D}^{(i)}_m\cup \UMA_m}\left\{\lambda: \lambda^{(i)}_k>\lambda^{(i)}_{m(i)}\right\}$ where for each player $p_i\in\mathcal{P}$, $\mathcal{D}^{(i)}_m:=\{a_k:p_i\succ_{a_k}m^{-1}(a_k)\}$ is the set of arms that prefer $p_i$ over their current match and $\UMA_m = \{a_k:m^{-1}(a_k)=\emptyset\}$ is the set of unmatched arms under stable matching $m$. We further call $\mathcal{E}_{m}^{(i)}:=\mathcal{D}^{(i)}_m\cup\UMA_m$ the set of probable blocking arms for the player $p_i$ w.r.t. the stable matching $m$. Note that $\lambda\in \neg m$ need not have unique stable matching. Using \texttt{LO1} and defining $w^{(i)}_k:=t^{(i)}_k/\sum_{i,k}t^{(i)}_k$, we derive the following lower bound for one-sided learning.
\begin{theorem}
    Any $\delta$-correct algorithm for unique stable matching market  instance $\mu\in \mathcal{S}_1$  under the one-sided learning model satisfies
    \begin{align*}
        \lim\inf_{\delta\to 0}\dfrac{\E_{\mu}[\tau_{\delta}]}{\log(1/\delta)}&\geq T^{\star}(\mu):=D(\mu)^{-1} \ \text{where}\\
        D(\mu) =\max_{w\in\Delta_{|\mathcal{P}|\times|\mathcal{A}}|}\min_{p_i\in\mathcal{P}}\min_{a_k\in\mathcal{E}_{m}^{(i)}}&\left\{w^{(i)}_{k}D\left(\mu^{(i)}_{k},x^{(i)}_{m(i),k}\right)+w^{(i)}_{m(i)}D\left(\mu^{(i)}_{m(i)},x^{(i)}_{m(i),k}\right)\right\}
    \end{align*}
   and $\Delta_{\mathcal{P}\times\mathcal{A}}:=\{(w^{(i)}_k)_{(p_i,a_k)\in\mathcal{P}\times\mathcal{A}}: \sum_{p_i\in\mathcal{P}}\sum_{a_k\in\mathcal{A}}w^{(i)}_k=1\}$, $w^{(i)}_k$ represent the proportion of samples given to player, arm pair $(p_i,a_k)$ and $x^{(i)}_{m(i),k}:=(w^{(i)}_{m(i)}\mu^{(i)}_{m(i)}+w^{(i)}_{k}\mu^{(i)}_{k})/(w^{(i)}_{m(i)}+w^{(i)}_{k})$. 
   Moreover, the maximizing optimal proportion denoted by $w^{\star}=(w^{\star(i)}_k)_{(p_i,a_k)\in\mathcal{P}\times\mathcal{A}}$ is unique. 
   \label{thm:optimalonesided}
\end{theorem}



Based on \texttt{LO1}, we have a set of necessary and sufficient first-order conditions obtained using Karush–Kuhn–Tucker (\texttt{KKT}) conditions \cite{Boyd_Vandenberghe_2004}, which uniquely characterizes the optimal allocation. 
\begin{lemma}
\label{lemma:firstone_sided}
    The solution $\bm{t}$ to the \texttt{\textbf{LO1}} problem satisfies the following first order conditions $\forall p_i\in\mathcal{P}$
    \begin{align*}
    C^{(i)}_{m(i),k}(\mu,\bm{t})&:=t^{(i)}_{m(i)}D\left(\mu^{(i)}_{m(i)},x^{(i)}_{m(i),k}\right) +t^{(i)}_{k}D\left(\mu^{(i)}_{k},x^{(i)}_{m(i),k}\right)= \log(1/2.4\delta)\ \forall a_k\in\mathcal{E}^{(i)}_m \ \text{and}\\
    g^{(i)}_{m(i)}(\mu,\mathcal{E}_{m}^{(i)},\bm{t})&:=\sum_{a_k\in\mathcal{E}_{m}^{(i)}}\dfrac{D\left(\mu^{(i)}_{m(i)},x^{(i)}_{m(i),k}\right)}{D\left(\mu^{(i)}_{k},x^{(i)}_{m(i),k}\right)} -1=0
\end{align*}
where $C^{(i)}_{m(i),k}$ is the index corresponding to $(p_i,a_k)$ defined when the player $p_i$ and $a_k$ can form a blocking pair if $p_i$ mistakes the preference order of arms $a_k$ and $a_{m(i)}$, $g^{(i)}_{m(i)}$ as the anchor function corresponding to the pair $(p_i,m(i))$ and $x^{(i)}_{m(i),k}:=(t^{(i)}_{m(i)}\mu^{(i)}_{m(i)}+t^{(i)}_{k}\mu^{(i)}_{k})/(t^{(i)}_{m(i)}+t^{(i)}_{k})$.
\end{lemma}

Motivated by the first-order conditions, we devise an iterative algorithm that performs finite exploration and exploits the pair $(p_i,a_k)$ so as to \textit{chase} the these conditions using a Top-Two rule. 
\paragraph{Anchored-Top-Two Algorithm} Denote the empirical index and anchor function as $\Tilde{C}^{(i)}_{\hat{m}(i),k}$, and $\Tilde{g}^{(i)}_{\hat{m}(i)}$ evaluated at the empirical estimates of $\mu$ denoted by $\hat{\mu}$ and $\hat{m}$ is a stable match with $(\hat{\mu},\succ)$. Denote the subroutines $\texttt{DA}_{\texttt{Arm}},\texttt{DA}_{\texttt{Player}}$ which outputs arm-optimal and player-optimal stable matching using Gale-Shapley algorithm described in the appendix, and can be performed off-line. For instances with $|\mathcal{P}|<|\mathcal{A}|$, given an exploration parameter $\gamma\in(0,1)$, and a stopping rule, the algorithm at iteration $t$ is hierarchically divided as player and arm choosing rule, which proceeds till the stopping rule is satisfied. Denote $N^{(i)}:=\sum_{k\in\mathcal{A}}N^{(i)}_k$
\begin{enumerate}[leftmargin=*, nosep,noitemsep, topsep=0pt]
    \item $\hat{m}\leftarrow \texttt{DA}_{\texttt{Arm}}(\hat{\mu},\succ)$ and
    construct the sets $\mathcal{E}_{\hat{m}}^{(i)}:=\mathcal{D}_{\hat{m}}^{(i)}\cup \UMA_{\hat{m}}\ \forall p_i\in\mathcal{P}$
    \item Choose $i_t$ from the \textcolor{blue}{player choosing rule}
\begin{itemize}[leftmargin=*, nosep,noitemsep, topsep=0pt,label=\textcolor{blue}{\rule{0.6ex}{0.6ex}}]
\item If $\min_{i\in\mathcal{P}} N^{(i)}<N^{\gamma}: i_t\leftarrow \arg\min_{i\in\mathcal{P}} N^{(i)}$, Else: $i_t\leftarrow \arg\min_{i\in\mathcal{P}} \min_{k\in\mathcal{E}_{m}^{(i)}} \tilde{C}^{(i)}_{m(i),k}$
\end{itemize}
\item Choose $k_t$ from the \textcolor{red}{arm choosing rule} for player $p_{i_t}$: 
\begin{itemize}[leftmargin=*, nosep,noitemsep, topsep=0pt,label=\textcolor{red}{\rule{0.6ex}{0.6ex}}]
\item If $\min_{k\in\mathcal{A}} N_k^{(i_t)}<\left(N^{(i_t)}\right)^{\gamma}: k_t\leftarrow \arg\min_{k\in\mathcal{A}} N^{(i_t)}_k$
\item Else: if $g^{(i_t)}_{\hat{m}_{i_t}}>0$ $k_t\leftarrow  \hat{m}(i_t) $ else $k_t\leftarrow\arg\min_{k\in\mathcal{E}_{\hat{m}}^{(i_t)}} \tilde{C}^{(i_t)}_{\hat{m}(i_t),k}$
\end{itemize}
    \item Match $p_{i_t}$ with $a_{k_t}$ and observe $X^{(i)}_k\sim\texttt{SPEF}(\mu^{(i)}_k)$, update $N^{(i)}_k$ and $\hat{\mu}^{(i)}_k$
\end{enumerate}

\textbf{Stopping Rule}  Our algorithm has a two-fold stopping rule, (1) if the estimated instance has unique matching (2) the minimum index is at least the threshold $\beta(t,\delta) = \log((|\mathcal{M}|-1)/\delta)+3|\mathcal{P}||\mathcal{A}|\log(1+\log(t))$ where $\mathcal{M}$ is the set of all matching possible with $|\mathcal{P}|$ players and $|\mathcal{A}|$ arms. We show in the appendix that any matching rule with this stopping rule satisfies the $\delta$-correctness property. The stopping rule (2) is derived through the traditionally used Generalized Likelihood Ratio test to rule any alternate hypothesis corresponding to different stable matching, see Chap. 3 \cite{kaufmann2020contributions} for more details.

\textbf{Remarks} (1)  Trivial case: When $\succ$ is such that all arms prefer distinct player, the matching is already known as we assumed that the instance has unique stable matching and running Arm-Proposing Deferred Acceptance with $\succ$ irrespective of any preference profile of players, will always lead to a same matching. (2) The estimated preference profile $\hat{\mu}$ can have multiple stable matchings of the estimated instance $(\hat{\mu},\succ)$. For our algorithm, we use $\texttt{DA}_{\texttt{Arm}}$ instead of $\texttt{DA}_{\texttt{Player}}$ as we can observe that the set $\mathcal{D}_{\hat{m}}^{(i)}$ is decreasing as we move down the distributive lattice formed by stable matchings \cite{Echenique_Immorlica_Vazirani_2023} and saves computation in the later stages of the algorithm. From Rural-Hospital Theorem \cite{Echenique_Immorlica_Vazirani_2023}, the set of unmatched arms $\UMA_{\hat{m}}$ remain unchanged for any stable matching  under the estimated instance.

\textbf{Fluid Dynamics}: The fluid dynamics closely approximates the algorithms' behavior with high probability.  To derive the dynamics, we consider the behavior of the algorithm in an idealized situation where the exact means are known and the samples are allocated as infinitesimal continuous objects instead of a discrete multiple of integers. In addition, we do not consider exploration since the means are assumed to be known. Further, \texttt{ATT1} stops exploration after a finite time. We denote by $t$ the continuous quantity interpreted as the budget that must be allocated between different pairs, that is, $\sum_{i}\sum_{k}t^{(i)}_k(t)=t$ where $t^{(i)}_k$ is the allocation to pair $(p_i,a_k)$. We analyze the dynamics on the global time scale, defined using $t$, and the local time scale for each player $p_i$ defined using $t^{(i)}=\sum_{k}t^{(i)}_k$. Throughout this section, the indexes and anchor functions are evaluated with $\mu$ and only change with the allocations.
Let $C_{\min}(t):=\min_{p_i,a_k}C^{(i)}_{m(i),k}(t)$. Define the set of players with minimum index as $\mathcal{P}_{\min}(t)$
and the set of arms for each minimum indexed player $p_i\in \mathcal{P}_{\min}(t)$ as  $\mathcal{A}_{\min}^{(i)}(t)$
    . We drop $t$, when it is clear. Based on the sign of the anchor function, we partition the set $\mathcal{P}_{\min}$ into $\mathcal{P}_0\cup\mathcal{P}_{-}\cup\mathcal{P}_{+}$, defined as the sets of players with anchor function being $0$, $-$ ve, and $+$ ve, respectively.
    
We first describe the fluid dynamics in words.   For players $p_i\in  \mathcal{P}_{+}$, $t^{(i)}_{m(i)}$ increases, as a result $g^{(i)}_{m(i)}$ decreases. For players $p_i\in \mathcal{P}_{-}$, samples are given to challenger arms with minimum index, i.e. $a_k\in\mathcal{A}^{(i)}_{\min}$ thus $t^{(i)}_{k}$ increases, as a result $g^{(i)}_{m(i)}$ increases. For players $p_i\in  \mathcal{P}_{0}$, the samples are given to both the stable matched arm and the challenger arm with a minimum index such that $g^{(i)}_{m(i)}$ remains zero. Furthermore, the index for $(p_i,a_k):a_k\in\mathcal{A}^{(i)}_{\min}(t)$ increases sublinearly for $p_i\in\mathcal{P}_{-}\cup \mathcal{P}_{+}$, linearly for $p_i\in\mathcal{P}_{0}$. In addition, these indexes stay together for $p_i\in \mathcal{P}_{-}\cup\mathcal{P}_{0}$. Thus, $p_i\in\mathcal{P}_{-}\cup \mathcal{P}_{+}$  are prioritized by algorithm until the anchor function is zero.  If the minimum index for two pairs meet, they remain equal i.e. if a player is once given a sample, it continues to receive samples. For $p_i\not\in \mathcal{P}_{\min}(t)$, $t^{(i)}_{k}(t)$ is constant with $t$ for all $k$. In the following, we provide a result that allows us to formulate the fluid dynamics, mainly the existence and uniqueness of the solution of the fluid path.

\begin{proposition}
There exists $N_{\min}$ such that $\E_{\mu}[N_{\min}]<\infty$ and for every positive $N\geq N_{\min}$, for a subset of players $p_i\in \mathcal{P}_{1}\subseteq\mathcal{P}$, and a set of arms $\mathcal{B}^{(i)}\subseteq\mathcal{E}_{m}^{(i)}$ there is a unique set of variables $\bm{N}^{(i)}_{\mathcal{B}} = \left(N^{(i)}_k: a_k\in\mathcal{B}^{(i)}\cup\{a_{m(i)}\}\right)\ \forall p_i\in \mathcal{P}_{1}$ and $I(N)$ satisfying the conditions $g^{(i)}_{m(i)} = 0$ for all  $p_i\in\mathcal{P}_{1}$, $\sum_{p_i\in\mathcal{P}}\sum_{a_k\in\mathcal{A}}N^{(i)}_k = N$ and $N^{(i)}_{k}D\left(\mu^{(i)}_{k},x^{(i)}_{m(i),k}\right)+N^{(i)}_{m(i)}D\left(\mu^{(i)}_{m(i)},x^{(i)}_{m(i),k}\right)=I(N)$ $\forall ~ p_i\in\mathcal{P}_{1}, ~a_k\in\mathcal{B}^{(i)}$.
\end{proposition}
Based on the arguments described above, we derive the dynamics using Ordinary Differential Equations (ODEs). To ease the presentation, we use the following notations. 
For each player $p_i$, define $f^{(i)}_k := -\frac{\partial}{\partial x} D(\mu^{(i)}_{m(i)},x)/D(\mu^{(i)}_{k},x)\big|_{x=x^{(i)}_{m(i),k}}$.  Let $\Delta^{(i)}_{k}:=\mu^{(i)}_{m(i)}-\mu^{(i)}_{k}$ and $h^{(i)}_{k} := f^{(i)}_k (t^{(i)}_{m(i)})^2\Delta^{(i)}_{k}/{(t^{(i)}_{m(i)}+t^{(i)}_k)^2}$. Further denote $D^{(i)}_{m(i),k}:=D(\mu^{(i)}_{m(i)},x^{(i)}_{m(i),k})$ and $D^{(i)}_{k,k}:=D(\mu^{(i)}_{k},x^{(i)}_{m(i),k})$. Let $h^{(i)}_{\min}:= \sum_{k\in\mathcal{A}^{(i)}_{\min}}h^{(i)}_{k}/D^{(i)}_{k,k}$, $h^{(i)}_{\texttt{w}} = \sum_{k\not\in\mathcal{A}^{(i)}_{\min}\cup\{m(i)\}}h^{(i)}_{k}t^{(i)}_k$, $D^{(i)}_{\min} = (\sum_{k\in\mathcal{A}^{(i)}_{\min}}1/D^{(i)}_{k,k})^{-1}$, and $\cdot$ on top denotes derivative w.r.t $t$.

\begin{theorem}[One-Sided Fluid]
\label{thm:onesidedgeneral}
    If at some time $t_0>0$, suppose the minimum index player set $\mathcal{P}_{\min}(t_0)$ partitions into $\mathcal{P}_{+}(t_0)\cup\mathcal{P}_{-}(t_0)\cup \mathcal{P}_{0}(t_0)$. We consider the ODEs till $t_1$ defined as the smallest time after $t_0$ such that either minimum index at $t_1$ becomes equal to the index of player, arm pair for which the index at $t_0$ was strictly larger
     or anchor function of some player becomes equal to zero at $t_1$. This gives  the following for all $t\in [t_0,t_1)$.

\begin{enumerate}[leftmargin=*, nosep,noitemsep, topsep=0pt]
\item The minimum index evolves using the following ODE \vspace{-0.1in}\begin{align*}
    \left(\dot{C}_{\min}(t)\right)^{-1}&=\sum_{p_i\in\mathcal{P}_{+}}\dfrac{1}{D^{(i)}_{m(i),k(t)}}+\sum_{\scriptstyle \substack{p_i \in \mathcal{P}_{-} \\ a_k \in \mathcal{A}^{(i)}_{\min}}}
\dfrac{1}{D^{(i)}_{k,k}}+\sum_{p_i\in\mathcal{P}_{0}}\dfrac{h^{(i)}_{\min}\sum_{k\in\bar{\mathcal{A}}^{(i)}_{\min}}t^{(i)}_k+h^{(i)}_{\texttt{w}}/D^{(i)}_{\min}}{C_{\min}h^{(i)}_{\min}+h^{(i)}_{\texttt{w}}}
\end{align*}
where $\forall p_i\in \mathcal{P}_{+}$, $k(t) := \arg\min_{a_k\in\mathcal{A}^{(i)}_{\min}(t)}C^{(i)}_{m(i),k}(t)$.
    \item for $p_i\in \mathcal{P}_{-}, a_k\in\mathcal{A}^{(i)}_{\min}$, we have $\dot{t}^{(i)}_k = (D^{(i)}_{k,k})^{-1}\dot{C}_{\min}(t)$ and for $p_i\in \mathcal{P}_{+}$, we have $\dot{t}^{(i)}_{m(i)} = (D^{(i)}_{m(i),k(t)})^{-1}\dot{C}_{\min}(t)$. Further, for  $p_i\in \mathcal{P}_{0}, a_k\in\mathcal{A}^{(i)}_{\min}$, we have
    \begin{align*}
     \dot{t}^{(i)}_{m(i)}(t) &= \dot{C}_{\min}\dfrac{t^{(i)}_{m(i)}h^{(i)}_{\min}}{C_{\min }h^{(i)}_{\min}+h^{(i)}_{\texttt{w}}}\ \ \text{and}\ \ 
      \dot{t}^{(i)}_{k}(t) = \dot{C}_{\min}\dfrac{h^{(i)}_{\texttt{w}}/D^{(i)}_{k,k}+t^{(i)}_k h^{(i)}_{\min}}{{C_{\min} h^{(i)}_{\min}}+h^{(i)}_{\texttt{w}}}
\end{align*}
 
        \item for $p_{i}\in\mathcal{P}\backslash\mathcal{P}_{\min}(t)$ and $a_{k}\in\mathcal{E}_{m}^{(i)}$, we have $\dot{C}_{i,m(i),k}(t)=0$

\item When $\mathcal{P}_{0}\neq \emptyset$, there exists a constant $\beta>0$ independent of $t$ such that $\dot{C}_{\min}>\beta$ and in addition, other larger indexes for players $p_i\in\mathcal{P}_{\min}$ are bounded from above. 
When $\mathcal{P}_{0}= \emptyset$, indexes for $p_i\in\mathcal{P}_{\min}$ increase with $N$. Indexes for player $p_i\not\in\mathcal{P}_{\min}$ are constant. Thus, eventually all indexes catch up. 
\end{enumerate}
\end{theorem}

By an argument similar to \cite{bandyopadhyay2024optimal}, for $t^{\star}$ defined as the smallest time after $t_{0}$ such that the first order condition of Lemma 3 is satisfied, we have $t^{\star}\leq \bigl(\min_{p_i\in\mathcal{P},a_k\in\mathcal{E}_{m}^{(i)}}w^{\star(i)}_k\bigr)^{-1}t_0$.
Now, we analyze the relationship between the fluid dynamics for each player $p_i\in\mathcal{P}_{\min}$ evolving on its local time scale $t^{(i)}(t)$ and global time scale $t$. Note that locally, each player is solving a best arm identification \texttt{BAI} problem. Thus, locally, the fluid dynamics will follow as in Theorem 4.1~\cite{bandyopadhyay2024optimal}. 
\begin{theorem}[Relating the local time scale with the global time scale]
\label{thm:onesidedlocalglobal}
    If $p_i$ correspond to the minimum index player at time $t_0$, then its minimum index defined as $C^{(i)}_{\min}(t_0)$ follows the following ODE 
    \begin{align*}
 \dot{C}^{(i)}_{\min}(t_0)  = \dot{C}_{\min}(t_0)= \left({\sum_{p_j\in\mathcal{P}_{\min}(t_0)}\left(\dfrac{d}{dt^{(j)}}C^{(j)}_{\min}(t_0^{(j)})\right)^{-1}}\right)^{-1}\ \ \forall p_i\in\mathcal{P}_{\min}(t_0)
\end{align*}
where the player $p_i$'s allocation defined as $t^{(i)}(t_0):=\sum_{a_k\in\mathcal{E}^{(i)}_m}t^{(i)}_k(t_0)+t^{(i)}_{m(i)}(t_0)$ evolves as follows
\begin{align*}
    \dot{t}^{(i)}(t_0) &= 
    \dot{C}_{\min}(t_0)\left(\dfrac{d}{dt^{(i)}}C^{(i)}_{\min}(t_0^{(i)})\right)^{-1}.
\end{align*}

\end{theorem}

\paragraph{Theoretical Guarantees} We first show that $\texttt{ATT1}$ follows first-order conditions after a finite time, at a rate $N^{-\xi}$ for $\xi=3\gamma/8$. Choosing $T_{\texttt{stable}}$ such that $\forall N$ after we have $|\hat{\mu}^{(i)}_k - \mu^{(i)}_k| \leq N^{-\xi}$, and $N^{(i)}_k=\Theta(N),\forall p_i\in \mathcal{P},a_k\in \mathcal{E}_{\hat{m}}^{(i)}$ implies convergence of the anchor function. In addition, the convergence of the indexes follows as after $T_{\texttt{stable}}$, each pair is pulled at least once in the $N^{1-\xi}$ rounds.
\begin{proposition}[Convergence of first order conditions]
There exists a random time $T_{\texttt{stable}}$ such that $\E[T_{\texttt{stable}}]<\infty$ then for all $N>T_{\texttt{stable}}$ we have, for instance-dependent constant $D>0$ that $ \max_{p_i} |\tilde{g}^{(i)}_{m(i)}|<DN^{-\xi}$ and $
        \max_{(p_i,a_k),(p_j,a_l)}|\Tilde{C}^{(i)}_{m(i),k}-\Tilde{C}^{(j)}_{m(j),l}|<DN^{1-\xi}$.
\end{proposition}
Due to the uniqueness of the optimal allocation, \texttt{ATT1} therefore follows the optimal proportion up to a perturbation of $N^{-\xi}$. The following theorem follows as at stopping time $\tau_{\delta}$, $\min C^{(i)}_{m(i),k}$ crosses $\beta(\tau_{\delta},\delta)=\log(1/\delta)+o(\log(1/\delta))$, thus we have $\Tilde{C}^{(i)}_{m(i),k}\approx \tau_{\delta}T^{\star{-1}}(\mu)\underset{\delta\to 0}{=}\log(1/\delta)$ a.s. in $\mathbb{P}_{\mu}$. 
\begin{theorem}[Asymptotic optimality]\label{thm:one_sided_assymp_opt}
    \texttt{ATT1} is $\delta$-correct and is asymptotically optimal for one-sided learning over market instances $\mu\in \mathcal{S}_{1}$ i.e. the corresponding stopping time satisfy 
\begin{align*}
    \limsup_{\delta\to 0}\frac{\E_{\mu}[\tau_{\delta}]}{\log(1/\delta)}\leq T^{\star}(\mu) \ \ \text{and}\ \ \limsup_{\delta\to 0}\frac{\tau_{\delta}}{\log(1/\delta)}\leq T^{\star}(\mu)\ \text{a.s. in}\ \mathbb{P}_\mu.
\end{align*}

\end{theorem}

%% file: paper/two-sided.tex
\section{Two-sided Learning}\label{sec:two_sided}
Recall, in this setting, both players and arms are unaware of their preferences. Let $\mu:=(\mu^{(i)}_k),\eta:=(\eta^{(i)}_k)$ be such that $(\mu,\eta)\in \mathcal{S}_2$ the set of instances such that the preference profile induced by $(\mu,\eta)$ has a unique stable match, denote as $m$. Let $\neg m=\{(\lambda,\nu):m\not\in {\mathcal{M}}_{\lambda,\nu}\}$ be the alternate set where $\mathcal{M}_{\lambda,\nu}$ is the set of stable matching induced by the preference profiles $(\lambda,\nu)$. 
\begin{theorem}
    Any $\delta$-correct algorithm for finding the stable matching of any market instance $(\mu,\eta)\in\mathcal{S}_2$ for the two-sided learning model satisfies
    \begin{align*}
    \lim\inf_{\delta\to 0}\frac{\E_{\mu,\eta}[\tau_{\delta}]}{\log(1/\delta)}&\geq T^{\star}(\mu,\eta):=D(\mu,\eta)^{-1}\quad \text{where}\\
        D(\mu,\eta) = \max_{w\in\Delta_{|\mathcal{P}|\times|\mathcal{A}}|}&\inf_{(\lambda,\nu)\in\neg m}\sum_{p_i\in\mathcal{P}}\sum_{a_k\in\mathcal{A}} w^{(i)}_k \left[d(\mu^{(i)}_k,\lambda^{(i)}_k)+d(\eta^{(i)}_k,\nu^{(i)}_k)\right].
    \end{align*}
\end{theorem}
Note that in the definition of $\neg m$, the instance $(\lambda,\nu)$ may have multiple stable matchings. For a pair $(p_i,a_{m(j)})$ to form a blocking pair under matching $m$, the conditions must be $\lambda^{(i)}_{m(i)} < \lambda^{(i)}_{m(j)}$ and $\eta^{(j)}_{m(j)} < \eta^{(i)}_{m(j)}$. In the true instance $(\mu,\eta)$, any pair $(p_i,a_{m(j)})$ can be categorized by the four logical possibilities shown below:
\begin{align*}
    \text{(1)} \quad \mu^{(i)}_{m(i)} > \mu^{(i)}_{m(j)} \ &\text{and} \  \eta^{(j)}_{m(j)} < \eta^{(i)}_{m(j)} \quad &&\text{(2)} \quad \mu^{(i)}_{m(i)} < \mu^{(i)}_{m(j)} \ \text{and} \  \eta^{(j)}_{m(j)} > \eta^{(i)}_{m(j)} \\
    \text{(3)} \quad \mu^{(i)}_{m(i)} > \mu^{(i)}_{m(j)} \ &\text{and} \  \eta^{(j)}_{m(j)} > \eta^{(i)}_{m(j)} \quad &&\text{(4)} \quad \mu^{(i)}_{m(i)} < \mu^{(i)}_{m(j)} \ \text{and} \  \eta^{(j)}_{m(j)} < \eta^{(i)}_{m(j)}
\end{align*}
The fourth case is impossible, as it implies that $(p_i,a_{m(j)})$ is a blocking pair for the matching $m$, contradicting our assumption that $m$ is stable. We can thus partition these pairs into three sets, $\mathcal{B}_1(m)$, $\mathcal{B}_2(m)$, and $\mathcal{B}_3(m)$, corresponding to the conditions (1), (2), and (3) respectively. The alternate instance is then defined as the union of all potential blocking pairs:
\[
 \neg m = \bigcup_{p_i}\bigcup_{m(j)\neq m(i)}\{\lambda^{(i)}_{m(i)} < \lambda^{(i)}_{m(j)}\ \text{and}\ \eta^{(j)}_{m(j)} < \eta^{(i)}_{m(j)}\}.
\]
For notational simplicity, we define the player and arm indices. The \textit{player index} for $p_i$, when comparing arms $a_{m(i)}$ and $a_{m(j)}$, is:
\begin{align*}
    C^{(i)}_{m(i),m(j)} := N^{(i)}_{m(j)} d\left(\mu^{(i)}_{m(j)}, x^{(i)}_{m(i),m(j)}\right) + N^{(i)}_{m(i)} d\left(\mu^{(i)}_{m(i)}, x^{(i)}_{m(i),m(j)}\right)
\end{align*}
The \textit{arm index} for $a_{m(j)}$, when comparing players $p_j$ and $p_i$, is:
\begin{align*}
    C^{(j,i)}_{m(j)} := N^{(i)}_{m(j)} d\left(\eta^{(i)}_{m(j)}, y^{(j,i)}_{m(j)}\right) + N^{(j)}_{m(j)} d\left(\eta^{(j)}_{m(j)}, y^{(j,i)}_{m(j)}\right)
\end{align*}
where the terms $x^{(i)}_{m(i),m(j)}$ and $y^{(i,j)}_{m(j)}$ are defined as:
\begin{align*}
    x^{(i)}_{m(i),m(j)} = \frac{N^{(i)}_{m(i)}\mu^{(i)}_{m(i)} + N^{(i)}_{m(j)}\mu^{(i)}_{m(j)}}{N^{(i)}_{m(i)} + N^{(i)}_{m(j)}}, \quad 
    y^{(i,j)}_{m(j)} = \frac{N^{(j)}_{m(j)}\eta^{(j)}_{m(j)} + N^{(i)}_{m(j)}\eta^{(i)}_{m(j)}}{N^{(j)}_{m(j)} + N^{(i)}_{m(j)}}.
\end{align*}
As before the lower bound can be modeled as the following convex program
\begin{align*}
 \texttt{\textbf{LO2}:}\quad\quad \quad\quad \quad\quad \quad\quad \quad\quad  \min_{N}  \sum_{i=1}^N\sum_{k=1}^K N_k^{(i)} \quad  \text{s.t.}\ N^{(i)}_k&\geq 0\\
  \min\left\{\min_{(p_i,a_k)\in \mathcal{B}_1}C^{(i)}_{m(i),k},\min_{(p_i,a_k)\in \mathcal{B}_2}C^{(m^{-1}(k),i)}_{k},\min_{(p_i,a_k)\in \mathcal{B}_3}C^{(i)}_{m(i),k}+C^{(m^{-1}(k),i)}_{k},\right\}&\geq C_{\delta}
\end{align*}
\begin{corollary}
    There may exist a player $p_i$, such that $w^{\star (i)}_{m(i)}=0$. This player is such that $\not\exists\ a_k: p_i\succ_{a_k}m^{-1}(k)$ and $\not\exists\ p_j:a_{ m(i)}\succ_{p_j}a_{m(j)}$, in words, all arms prefer their match over $p_i$ and all players prefer their match over match of $p_i$. For other players $p_j$, we have $w^{\star (j)}_{m(j)}>0$, and for all players $p_i$, we have $w^{\star (i)}_{\ell}>0\ \forall a_{\ell}\neq a_{m(i)}$.
\end{corollary}

\begin{example} {\emph{
    Consider a toy example of two players and two arms. The preference profile for players is $a_1\succ_{p_1}a_2$, $a_2\succ_{p_2}a_1$ and for arms is $p_1\succ_{a_1}p_2$, $p_2\succ_{a_2}p_1$. Thus, the unique stable matching is $p_1-a_1$ and $p_2-a_2$. In Fig. \ref{fig:interesting}, we show the optimal proportion of leaders when the two-sided instance is $\mu=$\(
\left( \begin{smallmatrix} \mu^{(1)}_1 & 1\\1 & 2\end{smallmatrix} \right)
\) and $\eta=$\(
\left( \begin{smallmatrix} \eta^{(1)}_1 & 1\\1 & 2\end{smallmatrix} \right)
\).  We vary $\mu^{(1)}_1,\eta^{(1)}_1\in (1,4]$ and observe instances in which optimal leader proportion is zero. Observe a transition from $w^{\star(1)}_1=0,w^{\star(2)}_2\neq 0$ to $w^{\star(1)}_1\neq 0,w^{\star(2)}_2=0$.}} 
\label{ex:toy}
\end{example} 

\begin{figure}[H]
    \centering
    \captionsetup{font=small} %
    \begin{minipage}{0.32\textwidth}
        \centering
        \includegraphics[width=\linewidth]{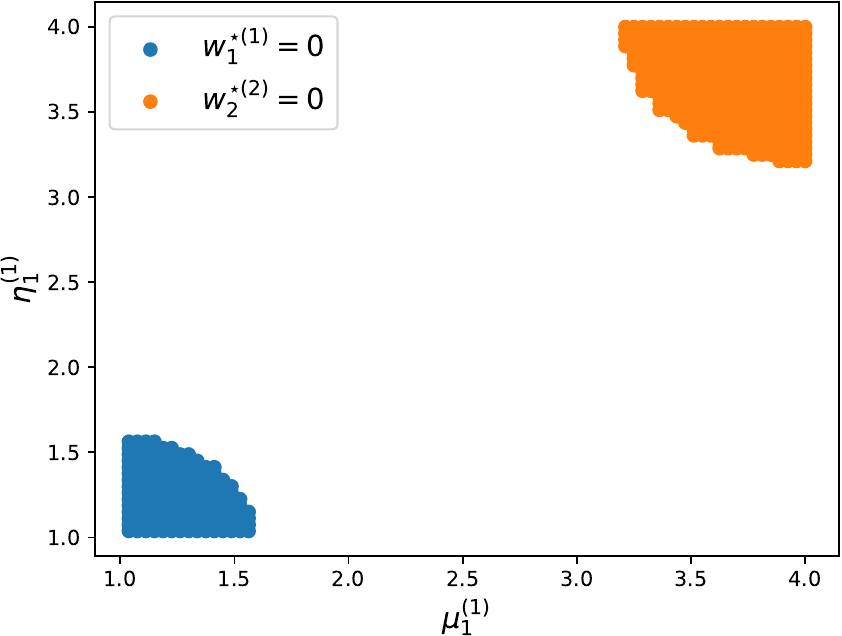}
    \end{minipage} \hfill
    \begin{minipage}{0.32\textwidth}
        \centering
        \includegraphics[width=\linewidth]{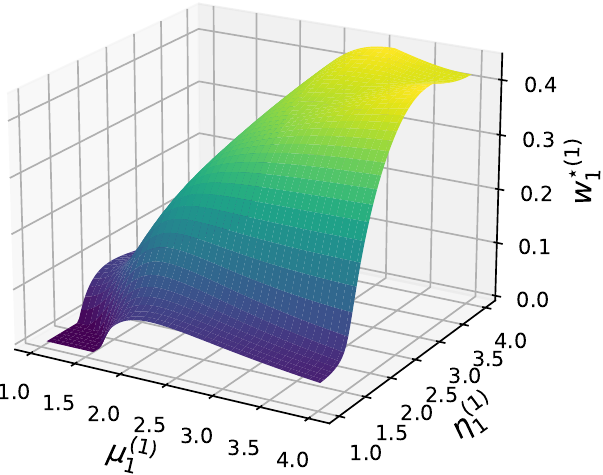}
    \end{minipage} \hfill
    \begin{minipage}{0.32\textwidth}
        \centering
        \includegraphics[width=\linewidth]{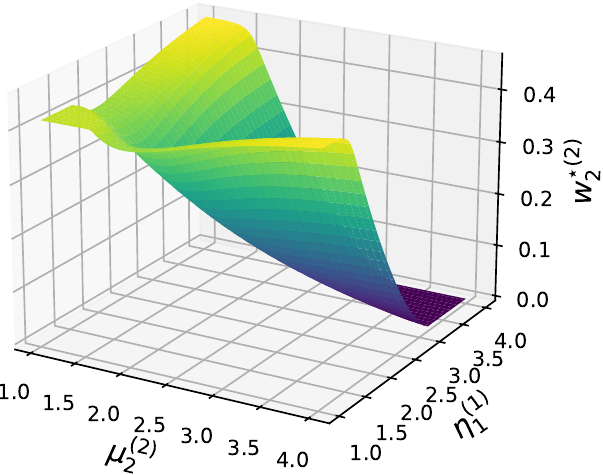}
    \end{minipage}
    \caption{Optimal Proportion of leaders in Example~\ref{ex:toy}}
    \label{fig:interesting}
\end{figure}

\begin{proposition}[First order conditions]
 The solution to the \texttt{\textbf{LO2}} problem satisfies the following first order conditions $\forall(p_i,a_k\neq a_{m(i)}),$
\[
C^{(i)}_{m(i),k}\,\mathbf{1}_{(p_i,a_k)\in\mathcal B_1}
+ C^{(i,m^{-1}(k))}_{k}\,\mathbf{1}_{(p_i,a_k)\in\mathcal B_2}
+ \bigl(C^{(i)}_{m(i),k}+C^{(i,m^{-1}(k))}_{k}\bigr)\,
  \mathbf{1}_{(p_i,a_k)\in\mathcal B_3}
= C_\delta
\]
i.e. all the index constraints are tight. For player $p_i:\mathcal{B}_1^{(i)}\cup \{p_j:a_{m(i)}\in \mathcal{B}_2^{(j)}\} \neq \emptyset$, we have $g^{(i)}_{m(i)} = 0$ and for $(p_i,a_k)\in\mathcal{B}_3$, we have $\max\left\{g^{(i)}_{m(i)},g^{(m^{-1}(k))}_{k}\right\} =0$ and $N^{(i)}_{m(i)}g^{(i)}_{m(i)}=0$ where
\begin{align*}
    &g^{(i)}_{m(i)}:=-1+\sum_{a_k\in\mathcal{B}_1^{(i)}}\dfrac{d(\mu^{(i)}_{m(i)} ,x^{(i)}_{m(i),k} )}{d(\mu^{(i)}_{k},x^{(i)}_{m(i),k} )}+\sum_{p_j:m(i)\in \mathcal{B}_2^{(j)}}\dfrac{d(\eta^{(i)}_{m(i)},y^{(j,i)}_{m(i)})}{d(\eta^{(j)}_{m(i)},y^{(j,i)}_{m(i)})}\\
   &+\sum_{k\in\mathcal{B}_3^{(i)}}\dfrac{d(\mu^{(i)}_{m(i)} ,x^{(i)}_{m(i),k} )}{d(\mu^{(i)}_{k},x^{(i)}_{m(i),k} )+d(\eta^{(i)}_{k},y^{(i,m^{-1}(k))}_{k})}+\sum_{p_j:m(i)\in \mathcal{B}_3^{(j)}}\dfrac{d(\eta^{(i)}_{m(i)},y^{(j,i)}_{m(i)})}{d(\mu^{(j)}_{m(i)},x^{(j)}_{m(j),m(i)} )+d(\eta^{(j)}_{m(i)},y^{(j,i)}_{m(i)})}.
\end{align*}
\end{proposition}
\paragraph{Anchored-Top-Two Algorithm} Based on the first-order conditions, we design our algorithm. For ease of exposition, we present the algorithm for $|\mathcal{P}|=|\mathcal{A}|$, and present a general algorithm in the appendix. Given an exploration parameter $\gamma\in(0,1)$, and a stopping rule, the algorithm at iteration $N$ proceeds as follows till the stopping rule is satisfied.\\ Denote $\tilde{C}^{(i)}_{\texttt{TS},\min}:=\min\Big\{\min_{a_k:(p_i,a_k)\in \mathcal{B}_1}C^{(i)}_{m(i),k},\min_{a_k:(p_i,a_k)\in \mathcal{B}_2}C^{(i,m^{-1}(k))}_k,\\
\min_{a_k:(p_i,a_k)\in \mathcal{B}_3}C^{(i)}_{m(i),k}+C^{(i,m^{-1}(k))}_k\Big\}$ and $N^{(i)}=\sum_{k}N^{(i)}_k$
 \setlength{\BoxH}{2.4cm}  
 \vspace{-10pt}
\begin{figure}[H]
    \centering
    \begin{minipage}[t]{0.6\textwidth}
        \centering
       {\setcounter{algocf}{1}
\SetAlgorithmName{Subroutine}{subroutine}{List of Subroutines}
\begin{minipage}[t]{\textwidth}
            \begin{minipage}[t]{0.43\textwidth}
                \begin{algorithm}[H]
                \small
                \uIf{$ \Tilde{g}^{(i)}_{\hat{m}(i)}>0$}{
                \textcolor{red}{Match $p_{i}$ with $\hat{m}(i)$}
                }
                \Else{
                \textcolor{orange}{Match $p_{i}$ with $a_{k}$}
                }
                \caption*{}
                \end{algorithm}
            \end{minipage}
            \hfill
            \begin{minipage}[t]{0.53\textwidth}
               {\setcounter{algocf}{2}
\SetAlgorithmName{Subroutine}{subroutine}{List of Subroutines}
\begin{algorithm}[H]
                \small
                \uIf{$ \Tilde{g}^{(\hat{m}^{-1}(k))}_{k}>0$}{
                \textcolor{blue}{Match $\hat{m}^{-1}(k)$ with $a_k$}
                }
                \Else{
                \textcolor{orange}{Match $p_{i}$ with $a_{k}$}
                }
                \caption*{}
                \end{algorithm}}
            \end{minipage}
        \end{minipage}}

        \vspace{0.5em}
        \begin{minipage}[t]{\textwidth}
            \centering
            \resizebox{0.5\linewidth}{!}{%
                \input{figures/algo_fig}
            }
        \end{minipage}
    \end{minipage}
    \hfill
    \begin{minipage}[t]{0.38\textwidth}
     {\setcounter{algocf}{3}
\SetAlgorithmName{Subroutine}{subroutine}{List of Subroutines}
        \begin{algorithm}[H]
        \small
        \uIf{$ \max\{\Tilde{g}^{(i)}_{\hat{m}(i)},\Tilde{g}^{(\hat{m}^{-1}(k))}_{k}\}<0$}{
        \textcolor{orange}{Match $p_{i}$ with $a_{k}$}
        }
        \Else{
            \uIf{$\Tilde{g}^{(i)}_{\hat{m}(i)}>\Tilde{g}^{(\hat{m}^{-1}(k))}_{k}$}{
            \textcolor{red}{Match $p_{i}$ with $\hat{m}(i)$}
            }
            \Else{
            \textcolor{blue}{Match $\hat{m}^{-1}(k)$ with $a_k$}
            }
        }
        \caption*{}
        \end{algorithm}}
    \end{minipage}
\end{figure}

\begin{enumerate}[leftmargin=*, nosep,noitemsep, topsep=0pt]
    \item $\hat{m}\leftarrow \texttt{DA}_{\texttt{Arm}}(\hat{\mu},\hat{\eta})$ and
    construct the sets $\mathcal{E}_{\hat{m}}^{(i)}:=\mathcal{D}_{\hat{m}}^{(i)}\cup \UMA_{\hat{m}}\ \forall p_i\in\mathcal{P}$
    \item Choose $i_t$ from the \textcolor{blue}{player choosing rule}
\begin{itemize}[leftmargin=*, nosep,noitemsep, topsep=0pt,label=\textcolor{blue}{\rule{0.6ex}{0.6ex}}]
\item If $\min_{i\in\mathcal{P}} N^{(i)}<N^{\gamma}: i_t\leftarrow \arg\min_{i\in\mathcal{P}} N^{(i)}$, Else: $i_t\leftarrow \arg\min_{i\in\mathcal{P}} \tilde{C}^{(i)}_{\texttt{TS},\min}$
\end{itemize}
\item Choose $k_t$ from the \textcolor{red}{arm choosing rule} for player $p_{i_t}$: 
\begin{itemize}[leftmargin=*, nosep,noitemsep, topsep=0pt,label=\textcolor{red}{\rule{0.6ex}{0.6ex}}]
\item If $\min_{k\in\mathcal{A}} N_k^{(i_t)}<\left(N^{(i_t)}\right)^{\gamma}: k_t\leftarrow \arg\min_{k\in\mathcal{A}} N^{(i_t)}_k$
\item Else: Let the minimum index correspond to $(p_{i_t},a_{k_t})$. If $(p_{i_t},a_{k_t})\in\mathcal{B}_{\texttt{I}}$ follow Subroutine \texttt{I}
\end{itemize}
    \item Match $p_{i_t}$ with $a_{k_t}$ and observe $X^{(i)}_k\sim\texttt{SPEF}(\mu^{(i)}_k)$, update $N^{(i)}_k$ and $\hat{\mu}^{(i)}_k$
\end{enumerate}

Sampling the leader corresponding to the largest anchor function is motivated by complementary slackness condition, e.g., if $g^{(i)}_{m(i)}>g^{(j)}_{m(j)}$, the leader pair $(p_j,a_{m(j)})$ receives no samples and thus $t^{(j)}_{m(j)}\approx 0$ further $g^{(i)}_{m(i)}$ is pushed towards 0 if it is positive. Let $h^{(i)}_{m(j)}:=\max\small{\{g^{(i)}_{m(i)},g_{m(j)}^{(j)}\}}$.

\textbf{Fluid Dynamics}: We now briefly discuss the fluid dynamics, which is an extension of one-sided dynamics. Suppose that the minimum index pair $(p_i,a_{m(j)})\in\mathcal{B}_3$. If $h^{(i)}_{m(j)}<0$, $t^{(i)}_{m(j)}$ increases, till $h^{(i)}_{m(j)}=0$. If $h^{(i)}_{m(j)}>0$, then the leader corresponding to the larger anchor function increases, till $h^{(i)}_{m(j)}=0$. Further, if the anchor functions are equal, depending on the allocation, they maintain together, else, they separate out. When $h^{(i)}_{m(j)}=0$, it remains 0, by increasing both the leader(s) and the challenger. The explicit ODEs can be found in the Appendix.

%% file: figures/algo_fig.tex
    
    
    
    
\centering
\begin{tikzpicture}
    \definecolor{redline}{RGB}{220, 50, 50}
    \definecolor{blueline}{RGB}{50, 100, 220}
    \node (p1) at (0,1) {\( p_i \)};
    \node (p2) at (-0.5,0) {\( m^{-1}(a_k) \)};
    \node (a1) at (3,1) {\( m(p_i) \)};
    \node (a2) at (3,0) {\( a_k \)};
    
    \draw[redline, thick] (p1) -- (a1);
    \draw[orange, dashed] (p1) -- (a2);
    \draw[blueline, thick] (p2) -- (a2);
\end{tikzpicture}

%% file: paper/experiments.tex
\section{Experimental Results}
In Fig.~\ref{fig:allthree}, we show the algoarithms' dynamics which shows the convergence to the first-order conditions. We consider a 3x3 market with a distinct preference setting. In \ref{fig:onesidedlfuid}, we present the dynamics of the anchor function and the idealized fluid setting, in addition to the dynamics of the normalized index. In \ref{fig:twosidedfluid}, we show the dynamics of the anchor function, it can be observed that \texttt{ATT2} converges to the optimal negative value of the anchor function for $p_1$, albeit slowly, but identifies it as negative quite early. 
\vspace{-10pt}
\begin{figure}[H]
    \centering

    \begin{subcaptionbox}{\texttt{One-sided learning}\label{fig:onesidedlfuid}}[0.64\linewidth]
        {
        \centering
        \begin{subfigure}[b]{0.49\linewidth}
            \includegraphics[width=\linewidth]{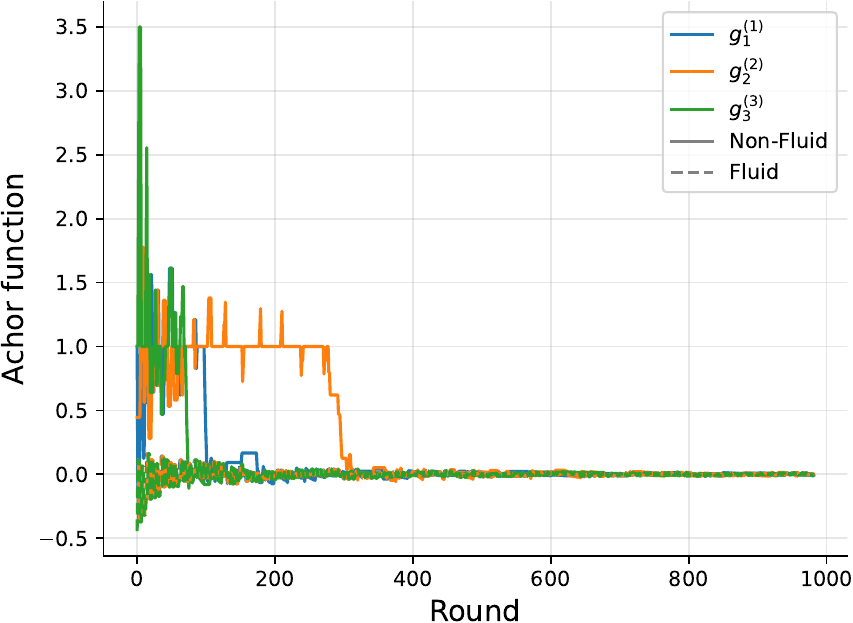}
        \end{subfigure}
        \hfill
        \begin{subfigure}[b]{0.49\linewidth}
            \includegraphics[width=\linewidth]{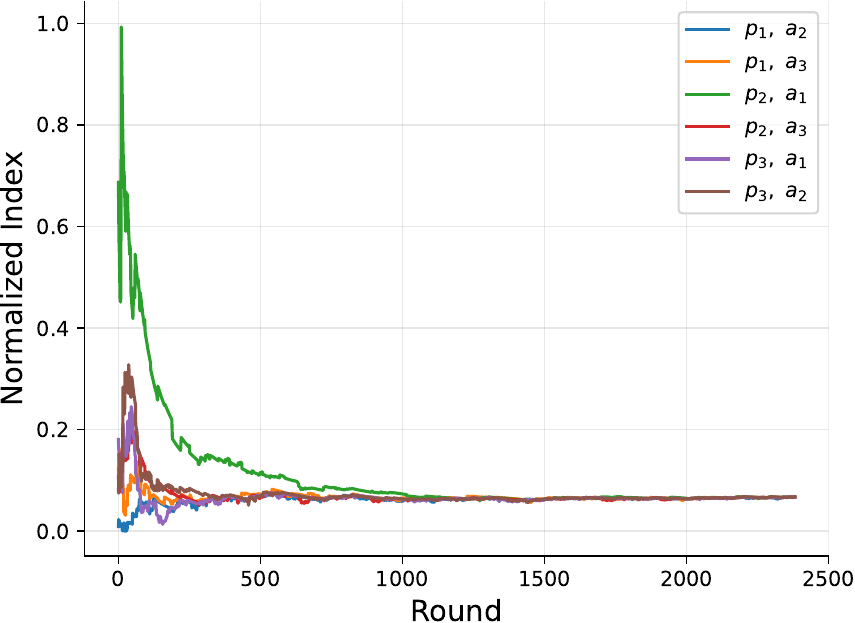}
        \end{subfigure}
        }
    \end{subcaptionbox}
    \hfill
    \begin{subcaptionbox}{\texttt{Two-sided learning}\label{fig:twosidedfluid}}[0.32\linewidth]
        {
        \includegraphics[width=\linewidth]{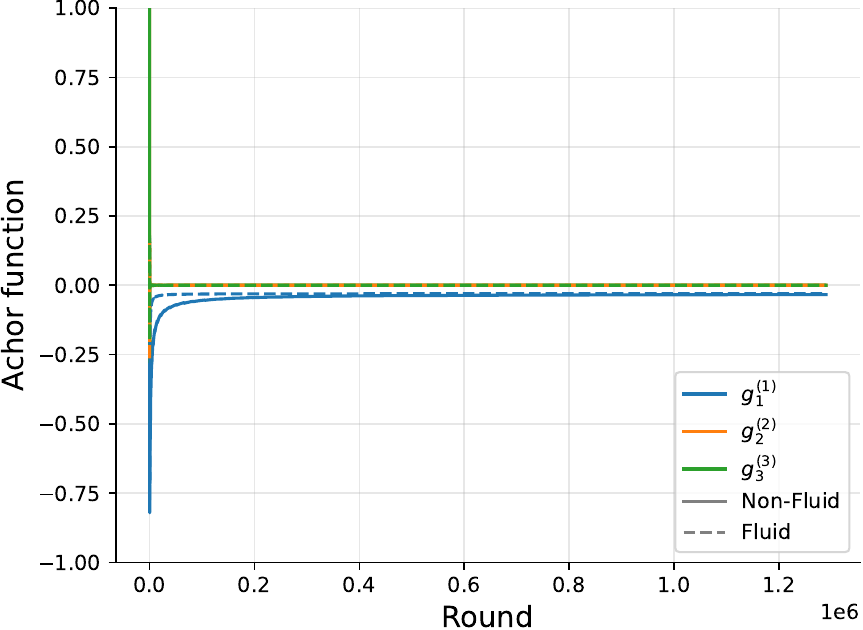}
        }
    \end{subcaptionbox}
    \caption{\texttt{ATT1} and \texttt{ATT2} dynamics}
    \label{fig:allthree}
\end{figure}
\begin{table}[H]
    \centering
  {\small\setlength{\tabcolsep}{8pt}\renewcommand{\arraystretch}{1.25}   \begin{tabular}{|c|c|c|c|c|}
    \hline
         {\small Learning Model  $\rightarrow$}&  \multicolumn{2}{c|}{One-sided}  & \multicolumn{2}{c|}{Two-Sided} \\\hline
        \diagbox[width=9em]{\scriptsize Instance $\downarrow$}{ \scriptsize
        Algorithm $\rightarrow$}   &  \texttt{ATT1} & $\bar{\beta}$-\texttt{EB-TC} &  \texttt{ATT2} & $(\bar{\alpha},\bar{\beta},\bar{\delta})$-\texttt{EB-TC} \\\hline
       \texttt{Distinct}  & $1008.31\pm 0.89$& $1029.36\pm0.89$ & $337.20\pm 0.33$ & $345.59\pm 0.32$ \\\hline
       \texttt{Serial}  & $1459.37\pm1.09$& $1518.94\pm1.16$&$1213.9\pm 1.96$& $1298.89\pm 1.26$\\\hline
        \texttt{SPC} & $1917.97\pm 1.30$& $2015.82\pm1.37$& $1433.01\pm 1.29$ & $1472.29\pm 1.29$\\\hline
    \end{tabular}}
    \caption{Stopping time for 5x5}
    \label{tab:one-sided}
\end{table}
To understand the variation in sample complexity or the stopping time, we run the proposed algorithms for 5000 runs and report the sample average and standard deviation in Table~\ref{tab:one-sided}. We set the exploration parameter $\gamma=0.25$ and the confidence parameter $\delta=0.001$. We consider a market of size 5x5 with market instances satisfying the unique stable matching properties, i.e. distinct first preferences, serial dictatorship, and sequential preference condition (SPC). We consider Gaussian instance with variance as $1$ and the means for each player and arm $\in\{2,2.5,3.5,5,7\}$ ordered according to the preferences. We compare the performance with extension of $\overline{\beta}$-\texttt{EB}-\texttt{TC} \cite{jourdan2022top}, which involves $\beta$ for each leader for one-sided learning, and for three types of leaders in two-sided learning, which we set as $0.5$. We demonstrate additional experiments in the Appendix that demonstrate superior performance of \texttt{ATT} and also show that $\beta$-algorithms maintain 2-competitiveness for $\beta=1/2$.

%% file: paper/conclusion.tex
\section{Conclusion}
We introduced a stable matching identification problem and designed computationally efficient Top-Two algorithms. We considered two-learning setups, one-sided learning in which we presented fluid dynamics that the algorithm tracks and prove that it is also asymptotically optimal. For two-sided learning, we proposed an extension of the algorithm which we show can have surprising property and very involved dynamics, and in special cases, we also derive the fluid ODEs. We demonstrate superior performance and also experimentally show that the algorithm is asymptotically optimal in both learning setups. The current research focuses mainly on market instances with unique stable matching, we comment on our trails and open problems in the appendix for the general setting.

%% file: paper/appendix/appendix-final.tex
\clearpage
\appendix
\begin{center}
    {\Large{\textbf{Appendix}}}
\end{center}

\section{Outline}
\begin{itemize}
    \item In \ref{app:pseudocode} we present pseudocodes of the \texttt{ATT1} and \texttt{ATT2} algorithm 
    \item In \ref{app:firstorder} we prove the first-order conditions of one-sided and two-sided learning corresponding to the convex program for the lower bound problem
    \item In \ref{app:stoppingrule} we prove the $\delta$-correctness of the \texttt{GLLR} stopping rule which follows majorly from \cite{kaufmann2020contributions}
    \item In \ref{app:fluid-one-sided} we prove the existence fluid dynamics and ODEs for allocations for one-sided learning
    \item In \ref{app:nonfluid-one-sided} we prove asymptotic optimality of \texttt{ATT1} algorithm for one-sided learning. The key observation here is to use the convergence from \cite{bandyopadhyay2024optimal} for the local time scale i.e. each player $p_i$ which is essentially a Best-Arm identification problem
    \item In \ref{app:fluid-two-sided} we provide the fluid dynamics ODEs for two-sided learning for 2x2 example with distinct preferences and a general case with two-sided serial dictatorship preferences
    \item In \ref{app:beta-algos} we provide the $\beta$-Top-Two algorithms and provide the competitive ratio following \cite{russo2020simple}
    \item In \ref{app:add-exps} we provide experimental details for experiments in the main paper along with additional experiments 
    \item In \ref{app:multistable} we provide insights about the problem under multiple stable matching
\end{itemize}

\textit{Comparison to recent work}: \cite{hosseini2024putting} considers a similar problem, but it distinguishes from our work as it considers the problem in a fixed-budget setting where the objective is to minimize the probability of mistake so that we stop before the given budget $T$ whereas our problem considers minimizing the expected sample complexity such that probability of mistake is less than $\delta$.  In particular, \cite{hosseini2024putting} (a) lacks a stopping rule dependent on $\delta$ but instead uses \texttt{LCB} and \texttt{UCB} estimates (b) \cite[Section 4.2]{hosseini2024putting} shows that their guarantee on $\delta-$ (or $\alpha-$) correctness depend on the choice of $\beta$ the parameter used in computing confidence bounds (c) their algorithms does not achieve optimal rate for small $\delta$ which they also list as one of their future directions, further, from Theorem 5 the ratio of sample complexity by $\log(1/\delta)$ is $\Tilde{\mathcal{O}}(\texttt{ES}(\underline{m})/\Delta^2)$ which can be large as compared to $T^{\star}$ in our paper (d) they do not consider two-sided learning as their algorithm necessarily requires one side to know the preference ranking. We would like to highlight that in bandit literature the obtaining tight lower bound and an algorithm matching it is an open problem in the fixed-budget setting \cite{pmlr-v178-open-problem-qin22a}.\\
 Regret minimizing algorithms, which rely on explore-commit, will suffer in fixed-confidence setting as the exploration is uniform and allocation to the suboptimal match can deviate a lot from the optimal allocation. Hence, we do not directly compare our derived algorithms with these works. 

\textit{Motivation for Centralized Learning}:
Economic literature on Market Design tells us that centralized mechanisms are often solutions to real-world market failures. Specifically, they can overcome problems of congestion, where the sheer volume of offers becomes unmanageable, and market thinning, which occurs when transactions unravel to be too early  \cite{roth2015gets}. A classic example is the National Resident Matching Program (NRMP), which was established to remedy the inefficiencies of the preceding decentralized market for medical residents. In that earlier system, the market unraveled as hospitals made offers progressively earlier each year. This practice forced applicants into rushed decisions on ``exploding offers'' with incomplete information, frequently leading to unstable and suboptimal matches.
     
\textit{Extension to Decentralized Learning}: Our one-sided learning algorithm can be easily extended in a decentralized mechanism as follows: Suppose we assume the existence of a preference query model for each arm, which imitates the matching process, i.e., upon multiple proposals. In that case, it gives 1 feedback to the most preferred proposal and 0 to others. As with previous literature, synchronization of all players is also assumed. Note that the Deferred Acceptance algorithm can be done in a decentralized fashion using this query model. \\
    The sets $\mathcal{D}^{(i)}_m$ can also be found in a decentralized manner using the query model as follows, under some pre-agreed ordering of players and arms. 
    \begin{itemize}[leftmargin=*, nosep,noitemsep, topsep=0pt]
        \item For $p_i\in\{p_1,\ldots,p_N\}$ and $k\in \{1,\ldots,N\}\backslash\{
    \hat{m}(i)\}$: $p_j\neq p_i$ queries $a_{\hat{m}(i)}$ and $p_i$ queries $a_k$
    \item For all players $p_{\ell}$ if it receives $1$, he adds the queried arm in $\mathcal{D}^{(\ell)}_{\hat{m}}$ if the arm is different from $a_{\hat{m}(\ell)}$
    \end{itemize}
Each player then communicates its minimum index value, and the minimum indexed player is then chosen. Based solely on its own information, player similarly decides the arm as in our paper. Note that this is not a unique extension; there may be multiple ways to design a decentralized mechanism based on various assumptions made in previous works, e.g., \cite{liu2021decentralized} assumes that after each round all players know the entire matching and also assumes that all players know, for each arm, which players are ranked higher than themselves.
\input{paper/appendix/notation}
\input{paper/appendix/algorithms}
\input{paper/appendix/first-order}

\input{paper/appendix/stopping-rule}

\input{paper/appendix/fluid-one-sided}
\input{paper/appendix/nonfluid-one-sided}

\input{paper/appendix/fluid-two-sided}
\input{paper/appendix/beta-algos}
\input{paper/appendix/add-exps}
\input{paper/appendix/multistable}

%% file: paper/appendix/notation.tex
$$
\begin{array}{l|l}
\hline
\multicolumn{2}{c}{\textbf{General Notation}} \\
\hline
\textbf{Symbol} & \textbf{Description} \\
\hline
\mathcal{P}, \mathcal{A} & \text{Sets of } M \text{ players and } K \text{ arms, respectively with $M\leq K$} \\
\succ_{p_i}, \succ_{a_k} & \text{A preference order for player}\ p_i\ \text{and arm}\ a_k\\  
\zeta & \text{A preference profile for players and arms which includes}\ \succ_{p_i}, \succ_{a_k}\ \forall i,k \\
m(\cdot) & \parbox[t]{12cm}{A bijective matching function from players $\mathcal{P}$ to arms $\mathcal{A}$ with inverse $m^{-1}(\cdot)$ defining a map from arms $\mathcal{A}$ to players $\mathcal{P}$}  \\
\delta & \text{The pre-specified confidence parameter} \\
\mathcal{M}_{\zeta} & \text{The set of stable matchings under preference profile } \zeta \\
\neg m & \text{set of instances (one- or two-sided) in which matching $m$ is unstable}\\
D(\mu_1,\mu_2) & \text{KL-divergence between two distributions} \\
N_k^{(i)} & \text{Number of times the pair } (p_i, a_k) \text{ is matched} \\
w_k^{(i)} & \text{Proportion of samples allocated to the pair } (p_i, a_k) \\
\tau_{\delta} & \text{The stopping time of the algorithm with confidence parameter $\delta$} \\
\hline
\multicolumn{2}{c}{\textbf{One-Sided Learning}} \\
\hline
X_k^{(i)} & \text{Random reward variables for the pair $(p_i, a_k)$ indicating $p_i$'s preference}  \\
\mu^{(i)}_k & \text{Expected reward for the pair $(p_i, a_k)$ indicating $p_i$'s preference}\ \E[X_k^{(i)}] \\
\hat{\mu}^{(i)}_k & \text{Empirical estimate of $\mu^{(i)}_k$}\\
\hat{m} & \text{The estimated stable matching based on empirical means} \\
T^{\star}(\mu) & \text{The asymptotic lower bound on sample complexity} \\
\mathcal{E}_m^{(i)} &  \parbox[t]{12cm}{Set of probable blocking arms for player $p_i$ under stable matching $m$ and is equal to $\mathcal{D}^{(i)}_m\cup \mathcal{A}^{\texttt{UM}}_m$ where $\mathcal{D}^{(i)}_m:=\{a_k:p_i\succ_{a_k}m^{-1}(a_k)\}$ and $\mathcal{A}^{\texttt{UM}}_m$ defined as the set of unmatched arms}  \\
C_{m(i),k}^{(i)} & \text{The index function for a potential blocking pair $(p_i, a_k$) defined when $a_k\in \mathcal{E}_m^{(i)}$}  \\
g_{m(i)}^{(i)} & \text{The anchor function for player } p_i \\
\Bar{\beta} & :=\parbox[t]{12cm}{$(\beta^{(1)},\ldots, \beta^{(M)})$ where $\beta^{(i)}$ is a parameter for $p_i$ in the $\Bar{\beta}$-Top-Two algorithm associated with $\frac{w^{(i)}_{m(i)}}{\sum_kw^{(i)}_k}$} \\
\hline
\multicolumn{2}{c}{\textbf{Two-Sided Learning}} \\
\hline
Y_k^{(i)} & \text{Random reward variables for the pair $(p_i, a_k)$ indicating $a_k$'s preference}  \\
\eta^{(i)}_k & \text{Expected reward for the pair $(p_i, a_k)$ indicating $a_i$'s preference}\ \E[Y_k^{(i)}] \\
\hat{\eta}^{(i)}_k & \text{Empirical estimate of $\eta^{(i)}_k$}\\
T^{\star}(\mu,\eta) & \text{The asymptotic lower bound on sample complexity} \\
\mathcal{B}_1(m) & \text{A set of elements $(p_i,a_k)$ s.t. $\mu^{(i)}_{m(i)}>\mu^{(i)}_{k}$ and $\eta^{(i)}_k>\eta^{m^{-1}(k)}_k$}\\
\mathcal{B}_2(m) & \text{A set of elements $(p_i,a_k)$ s.t. $\mu^{(i)}_{m(i)}<\mu^{(i)}_{k}$ and $\eta^{(i)}_k>\eta^{m^{-1}(k)}_k$}\\
\mathcal{B}_3(m) & \text{A set of elements $(p_i,a_k)$ s.t. $\mu^{(i)}_{m(i)}<\mu^{(i)}_{k}$ and $\eta^{(i)}_k<\eta^{m^{-1}(k)}_k$}\\
C_{m(i),k}^{(i)} & \text{The index function for a pair $(p_i, a_k$) defined when $\mu^{(i)}_{m(i)}>\mu^{(i)}_{k}$}  \\
C^{(i,m^{-1}(k))}_{m(i)} & \text{The index function for a pair $(p_i, a_k$) defined when $\eta^{(m^{-1}(k))}_{k}>\eta^{(i)}_{k}$}\\
\mathcal{B}_\ell^{(i)} & \text{A set of arms $a_k$ such that $(p_i,a_k)\in\mathcal{B}_\ell$  for $\ell\in\{1,2,3\}$}\\
\texttt{DA}_{\texttt{Arm}}(\cdot, \cdot\cdot) & \parbox[t]{12cm}{a matching output of Arm-proposing Deferred acceptance algorithm where $\cdot$ is an estimated mean matrix or a preference profile $\succ$ for players and $\cdot\cdot$ is similarly for arms} \\
h^{(i)}_{k} & := \max\{g^{(i)}_{m(i)},g^{(m^{-1}(k)}_k)\}\ \text{defined for a pair $(p_i,a_k)\in \mathcal{B}_{\ell}$ for $\ell\in\{1,2,3\}$} \\
\hline
\end{array}
$$

%% file: paper/appendix/algorithms.tex
\section{Algorithms}
\label{app:pseudocode}
\subsection{One-Sided Learning}

\begin{algorithm}[!h]
\caption{\texttt{ATT1} with $M\leq K$}
Input: preferences of arms over players $\succ$\\
    \For{$t=1,2,\ldots$}{
    $\mathcal{E}_p=\{p_i:\sum_kN^{(i)}_k\leq N^{\gamma}\}$ and $\mathcal{E}^{(i)}_a=\{a_k:N^{(i)}_k\leq \left(\sum_kN^{(i)}_k\right)^{\gamma}\}\ \forall p_i\in \mathcal{E}_p$\\
    Let $\mathcal{E}:=\{(p_i,a_k):p_i\in \mathcal{E}_p\ \text{and}\ a_k\in\mathcal{E}_a^{(i)}\}$\\
    $\hat{m}\leftarrow \texttt{DA}_{\texttt{Arm}}(\hat{\mu},\succ)$\\
    Construct $\mathcal{D}_{\hat{m}}^{(i)}=\{k\neq \hat{m}(i): p_i\succ_{a_k}\hat{m}^{-1}(k)\} \ \forall i$ and $\UMA_{\hat{m}}=\{k: \hat{m}^{-1}(k)=\emptyset\}$\\
    \uIf{$\mathcal{E}\neq \emptyset$ }{
       Match $p_{i_t}$ with $a_{k_t}$ from $\mathcal{E}$
    }
    \Else{
    \lIf{$\mathcal{E}_p\neq \emptyset$}{Select $p_{i_t}$ from $\mathcal{E}_p$} \lElse{ Select $p_{i_t}$ with the minimum index i.e.  $\Tilde{C}^{(i_t)}_{\hat{m}(i_t),\ell}=\min_{i\in[N],k\in \mathcal{D}^{(i)}_{\hat{m}}}\Tilde{C}^{(i)}_{\hat{m}(i),k}(\hat{\mu},N)$}
    \lIf{$\mathcal{E}_a^{(i_t)}\neq \emptyset$}{Select $a_{k_t}$ from $\mathcal{E}_a^{(i_t)}$}\lElse{Select $a_{k_t}$ with the minimum index i.e. $\Tilde{C}^{(i_t)}_{\hat{m}(i_t),k_t}=\min_{k\in \mathcal{D}^{(i_t)}_{\hat{m}}}\Tilde{C}^{(i_t)}_{\hat{m}(i_t),k}(\hat{\mu},N)$}
    
 \uIf{$\Tilde{g}^{(i_t)}_{\hat{m}}>0$}{
    Match player $p_{i_{t}}$ with arm $a_{{\hat{m}}(i_t)}$
    }
    \Else{
     Match player $p_{i_{t}}$ with arm $a_{k_t}$
    }
    }

    \lIf{$\texttt{DA}_{\texttt{Arm}}(\hat{\mu},\succ)=\texttt{DA}_{\texttt{Player}}(\hat{\mu},\succ)$ and $\Tilde{C}>\beta(t,\delta)$ }{Recommend $\texttt{DA}_{\texttt{Arm}}(\hat{\mu},\succ)$}}
    \label{algo:mainATT1}
\end{algorithm}
Another possible exploration sets are $\mathcal{E}_p=\{p_i:\sum_kN^{(i)}_k\leq N^{\gamma}\}$ and $\mathcal{E}^{(i)}_a=\{a_k:N^{(i)}_k\leq \left(N\right)^{\gamma^2}\}\ \forall p_i\in \mathcal{E}_p$.

\subsection{Two-Sided Learning}
\begin{algorithm}[!h]
\caption{\texttt{ATT2} (Two-Sided Learning model with $M\leq K)$}
Input: preferences of arms over players $\pi$\\
    \For{$t=1,2,\ldots$}{
    $\mathcal{E}_p=\{p_i:\sum_kN^{(i)}_k\leq N^{\gamma}\}$ and $\mathcal{E}^{(i)}_a=\{a_k:N^{(i)}_k\leq \left(\sum_kN^{(i)}_k\right)^{\gamma}\}\ \forall p_i\in \mathcal{E}_p$\\
    Let $\mathcal{E}:=\{(p_i,a_k):p_i\in \mathcal{E}_p\ \text{and}\ a_k\in\mathcal{E}_a^{(i)}\}$\\
    $\hat{m}\leftarrow \texttt{DA}_{\texttt{Arm}}(\hat{\mu},\hat{\eta})$\\
    Construct $\mathcal{B}_1^{(i)},\mathcal{B}_2^{(i)},\mathcal{B}_3^{(i)}$ for all players $p_i$ and $\UMA_{\hat{m}}=\{k: \hat{m}^{-1}(k)=\emptyset\}$\\
    Let $\Tilde{C}^{(i)}_{\hat{m}(i),k}=\min_{i\in[M],k\in \mathcal{B}_1^{(i)}\cup \UMA_{\hat{m}}}\Tilde{C}^{(i)}_{\hat{m}(i),k}(\hat{\mu})$,\\
    \phantom{Let }$\Tilde{C}^{(i,\hat{m}^{-1}(k))}_{k}=\min_{i\in[M],k\in \mathcal{B}_2^{(i)}}\Tilde{C}^{(i,\hat{m}^{-1}(k))}_{k}(\hat{\eta})$, \\
    \phantom{Let }$ \Tilde{C}^{(i,\hat{m}^{-1}(k))}_{m(i),k}=\min_{i\in[M],m(j)\in \mathcal{B}_3^{(i)}}  \Tilde{C}^{(i,\hat{m}^{-1}(k))}_{m(i),k}(\hat{\mu},\hat{\eta})$\\

    Let $\Tilde{C}=\min\left\{\Tilde{C}^{(i)}_{\hat{m}(i),k},\Tilde{C}^{(i,\hat{m}^{-1}(k))}_{k},\Tilde{C}^{(i,\hat{m}^{-1}(k))}_{m(i),k}\right\}$\\
    \uIf{$\mathcal{E}\neq \emptyset$ }{
       Match $p_{i_t}$ with $a_{k_t}$ from $\mathcal{E}$
    }
    \Else{
    \lIf{$\mathcal{E}_p\neq \emptyset$}{Select $p_{i_t}$ from $\mathcal{E}_p$} \lElse{ Select $p_{i_t}$ with the minimum index i.e.  $\Tilde{C}^{(i_t)}_{\hat{m}(i_t),\ell}=\min_{i\in[N],k\in \mathcal{D}^{(i)}_{\hat{m}}}\Tilde{C}^{(i)}_{\hat{m}(i),k}(\hat{\mu},N)$}
    \lIf{$\mathcal{E}_a^{(i_t)}\neq \emptyset$}{Select $a_{k_t}$ from $\mathcal{E}_a^{(i_t)}$}\lElse{Select $a_{k_t}$ with the minimum index i.e. $\Tilde{C}^{(i_t)}_{\hat{m}(i_t),k_t}=\min_{k\in \mathcal{D}^{(i_t)}_{\hat{m}}}\Tilde{C}^{(i_t)}_{\hat{m}(i_t),k}(\hat{\mu},N)$}
    Call subroutine $\texttt{I}$ if the minimum index corresponding to $(p_{i_t},a_{k_t})$ belongs to $\mathcal{B}_{\texttt{I}}$ else if $a_{k_t}\in\mathcal{A}^{\texttt{UM}}_{\hat{m}}$ call subroutine 1
    }
    \uIf{$\texttt{DA}_{\texttt{Arm}}(\hat{\mu},\hat{\eta})=\texttt{DA}_{\texttt{Player}}(\hat{\mu},\hat{\eta})$ and $\Tilde{C}>\beta(t,\delta)$ }{Recommend the matching $\texttt{DA}_{\texttt{Arm}}(\hat{\mu},\hat{\eta})$}}
\end{algorithm}

 \setlength{\BoxH}{2.4cm}  
 \vspace{-10pt}
\begin{figure}[H]
    \centering
    \begin{minipage}[t]{0.6\textwidth}
        \centering
       {\setcounter{algocf}{1}
\SetAlgorithmName{Subroutine}{subroutine}{List of Subroutines}
\begin{minipage}[t]{\textwidth}
            \begin{minipage}[t]{0.43\textwidth}
                \begin{algorithm}[H]
                \small
                \uIf{$ \Tilde{g}^{(i)}_{\hat{m}(i)}>0$}{
                \textcolor{red}{Match $p_{i}$ with $\hat{m}(i)$}
                }
                \Else{
                \textcolor{orange}{Match $p_{i}$ with $a_{k}$}
                }
                \caption*{}
                \end{algorithm}
            \end{minipage}
            \hfill
            \begin{minipage}[t]{0.53\textwidth}
               {\setcounter{algocf}{2}
\SetAlgorithmName{Subroutine}{subroutine}{List of Subroutines}
\begin{algorithm}[H]
                \small
                \uIf{$ \Tilde{g}^{(\hat{m}^{-1}(k))}_{k}>0$}{
                \textcolor{blue}{Match $\hat{m}^{-1}(k)$ with $a_k$}
                }
                \Else{
                \textcolor{orange}{Match $p_{i}$ with $a_{k}$}
                }
                \caption*{}
                \end{algorithm}}
            \end{minipage}
        \end{minipage}}

        \vspace{0.5em}
        \begin{minipage}[t]{\textwidth}
            \centering
            \resizebox{0.5\linewidth}{!}{%
                \input{figures/algo_fig}
            }
        \end{minipage}
    \end{minipage}
    \hfill
    \begin{minipage}[t]{0.38\textwidth}
     {\setcounter{algocf}{3}
\SetAlgorithmName{Subroutine}{subroutine}{List of Subroutines}
        \begin{algorithm}[H]
        \small
        \uIf{$ \max\{\Tilde{g}^{(i)}_{\hat{m}(i)},\Tilde{g}^{(\hat{m}^{-1}(k))}_{k}\}<0$}{
        \textcolor{orange}{Match $p_{i}$ with $a_{k}$}
        }
        \Else{
            \uIf{$\Tilde{g}^{(i)}_{\hat{m}(i)}>\Tilde{g}^{(\hat{m}^{-1}(k))}_{k}$}{
            \textcolor{red}{Match $p_{i}$ with $\hat{m}(i)$}
            }
            \Else{
            \textcolor{blue}{Match $\hat{m}^{-1}(k)$ with $a_k$}
            }
        }
        \caption*{}
        \end{algorithm}}
    \end{minipage}
\end{figure}

%% file: paper/appendix/first-order.tex
\section{Lower Bound}
\label{app:firstorder}
In this section, we prove the first order conditions using the lower bound optimization problems introduced in the paper. Recall, the first order conditions for a convex program from \cite{bandyopadhyay2024optimal}. The uniqueness of the solution is straightforward due to Slater's condition.
\subsection{One-Sided Learning}
Recall the lower bound optimization problem 
   \begin{align*}
  \texttt{\textbf{LO1}:}\quad\quad \quad\quad \quad\quad \quad\quad \quad\quad \quad\quad \quad\quad \min_{t} \sum_{i=1}^t\sum_{k=1}^K t_k^{(i)}\\
    \text{s.t.}\ t^{(i)}_k&\geq 0\\
t^{(i)}_{k}D\left(\mu^{(i)}_{k},x^{(i)}_{m(i),k}\right)+t^{(i)}_{m(i)}D\left(\mu^{(i)}_{m(i)},x^{(i)}_{m(i),k}\right)&\geq 1 \ \forall i\in[M],\ k\in \mathcal{D}_m^{(i)}\\
t^{(i)}_{k}D\left(\mu^{(i)}_{k},x^{(i)}_{m(i),k}\right)+t^{(i)}_{m(i)}D\left(\mu^{(i)}_{m(i)},x^{(i)}_{m(i),k}\right)&\geq 1 \ \forall i\in [M],\ k\in\UMA_m
\end{align*}

Let us introduce KKT variables: dual variables $\{\gamma^{(i)}_k\}_{k=1}^{K}\phantom{}_{i=1}^M$ and $\alpha_{i,k}$ for the second constraint and $\beta_{i,k}$ for the third constraint. The Lagrangian is as follows
\begin{align*}
    \mathcal{L}(t,\alpha,\gamma) 
    &=  \sum_{i=1}^t t_{m(i)}^{(i)}(1-\gamma^{(i)}_{m(i)})+ \sum_{i=1}^t\sum_{k\neq m(i): k\not\in\mathcal{A}^{\texttt{UM}}_m} t_{k}^{(i)}(1-\gamma^{(i)}_{k})+ \sum_{i=1}^t\sum_{k\in\mathcal{A}^{\texttt{UM}}_m} t_{k}^{(i)}(1-\gamma^{(i)}_{k})\\
    &-\sum_{i=1}^t\sum_{k\in\mathcal{D}^{(i)}_{m}}\alpha_{i,k}\left( t^{(i)}_{k}D\left(\mu^{(i)}_{k},x^{(i)}_{m(i),k}\right)+t^{(i)}_{m(i)}D\left(\mu^{(i)}_{m(i)},x^{(i)}_{m(i),k}\right)-1\right)\\
    &-\sum_{i=1}^t\sum_{k\in\mathcal{A}^{\texttt{UM}}_m}\beta_{i,k}\left( t^{(i)}_{k}D\left(\mu^{(i)}_{k},x^{(i)}_{m(i),k}\right)+t^{(i)}_{m(i)}D\left(\mu^{(i)}_{m(i)},x^{(i)}_{m(i),k}\right)-1\right)\\
    &= \sum_{i=1}^t t_{m(i)}^{(i)}(1-\gamma^{(i)}_{m(i)})+ \sum_{i=1}^t\sum_{k\in \mathcal{D}_m^{(i)}} t_{k}^{(i)}(1-\gamma^{(i)}_{k})+ \sum_{i=1}^t\sum_{k\neq m(i): k\not\in\{\mathcal{A}^{\texttt{UM}}_m\cup \mathcal{D}_m^{(i)}\}} t_{k}^{(i)}(1-\gamma^{(i)}_{k})\\
    &+ \sum_{i=1}^t\sum_{k\in\mathcal{A}^{\texttt{UM}}_m} t_{k}^{(i)}(1-\gamma^{(i)}_{k})\\
     &-\sum_{i=1}^t\sum_{j\in\mathcal{D}^{(i)}_{m}}\alpha_{i,k}\left( t^{(i)}_{k}D\left(\mu^{(i)}_{k},x^{(i)}_{m(i),k}\right)+t^{(i)}_{m(i)}D\left(\mu^{(i)}_{m(i)},x^{(i)}_{m(i),k}\right)-1\right)\\
    &-\sum_{i=1}^t\sum_{k\in\mathcal{A}^{\texttt{UM}}_m}\beta_{i,k}\left( t^{(i)}_{k}D\left(\mu^{(i)}_{k},x^{(i)}_{m(i),k}\right)+t^{(i)}_{m(i)}D\left(\mu^{(i)}_{m(i)},x^{(i)}_{m(i),k}\right)-1\right)
\end{align*}
For the stationarity condition, we have the following. 
\begin{align*}
    1-\gamma^{(i)}_{m(i)}-\sum_{j\in\mathcal{D}^{(i)}_{m}}\alpha_{i,k}D\left(\mu^{(i)}_{m(i)},x^{(i)}_{m(i),k}\right)-\sum_{k\in\mathcal{A}^{\texttt{UM}}_m}\beta_{i,k}D\left(\mu^{(i)}_{m(i)},x^{(i)}_{m(i),k}\right)&=0 \tag{$t^{(i)}_{m(i)}$}\\
     1-\gamma^{(i)}_{k}-\alpha_{i,k}D\left(\mu^{(i)}_{k},x^{(i)}_{m(i),k}\right)&=0 \tag{$t^{(i)}_{k}\ k\in \mathcal{D}^{(i)}_{m}$}\\
     1-\gamma^{(i)}_{k}-\beta_{i,k}D\left(\mu^{(i)}_{k},x^{(i)}_{m(i),k}\right)&=0 \tag{$t^{(i)}_{k}\ k\in\mathcal{A}^{\texttt{UM}}_m$}\\
1-\gamma^{(i)}_{k}&=0 \tag{$t^{(i)}_{k}\ k\neq m(i): k\not\in\{\mathcal{A}^{\texttt{UM}}_m\cup \mathcal{D}_m^{(i)}\}$}
\end{align*}

For $p_i,a_k$, using complementary slackness we have $t^{(i)}_k\gamma^{(i)}_k$. From stationary conditions, we have $1-\gamma^{(i)}_{k}=0$  $\forall k\neq m(i): k\not\in\{\mathcal{A}^{\texttt{UM}}_m\cup \mathcal{D}_m^{(i)}\}$, thus to minimize the objective we have $t^{(i)}_k=0$ for each of these $k$. \\
Further, for $k\in \{m(i)\}\cup \mathcal{E}_{m}$ where recall $\mathcal{E}_m:=\mathcal{D}_m^{(i)}\cup\mathcal{A}^{\texttt{UM}}_m$, we have from the second and third constraint that $t^{(i)}_k\neq 0$. This along with complementary slackness condition implies that $\gamma^{(i)}_k=0 \ \forall i,k\in \{m(i)\}\cup \mathcal{E}_{m}$. Substituting this in the above equation gives the following first order conditions
\begin{align}
\sum_{j\in\mathcal{D}^{(i)}_{m}}\dfrac{D\left(\mu^{(i)}_{m(i)},x^{(i)}_{m(i),k}\right)}{D\left(\mu^{(i)}_{k},x^{(i)}_{m(i),k}\right)}+\sum_{k\in\mathcal{A}^{\texttt{UM}}_m}\dfrac{D\left(\mu^{(i)}_{m(i)},x^{(i)}_{m(i),k}\right)}{D\left(\mu^{(i)}_{k},x^{(i)}_{m(i),k}\right)}&=1 \ \forall i\in[M]\\
     t^{(i)}_{k}D\left(\mu^{(i)}_{k},x^{(i)}_{m(i),k}\right)+t^{(i)}_{m(i)}D\left(\mu^{(i)}_{m(i)},x^{(i)}_{m(i),k}\right)&=1 \ \forall i\in[M]\ \text{and}\ k\in\mathcal{D}^{(i)}_{\mu}\\
t^{(i)}_{k}D\left(\mu^{(i)}_{k},x^{(i)}_{m(i),k}\right)+t^{(i)}_{m(i)}D\left(\mu^{(i)}_{m(i)},x^{(i)}_{m(i),k}\right)&= 1 \ \forall i\in[M]\ \text{and}\ k\in\mathcal{A}^{\texttt{UM}}_m
\end{align}

\subsection{Two-Sided Learning}

Recall that for $(\lambda,\nu)$ to create a blocking pair $(p_i,a_k)$ under the market instance $(\mu,\eta)$ we need $\lambda^{(i)}_{m(i)}<\lambda^{(i)}_{k}$ and $\eta^{(m^{-1}(k))}_k<\eta^{(i)}_k$. For any $(p_i,a_k)$ pair following three cases can occur under the instance $\mu,\eta$, 
\begin{enumerate}
    \item $\mu^{(i)}_{m(i)}>\mu^{(i)}_{k}$ and $\eta^{(m^{-1}(k))}_{k}<\eta^{(i)}_{k}$
    \item $\mu^{(i)}_{m(i)}<\mu^{(i)}_{k}$ and $\eta^{(m^{-1}(k))}_{k}>\eta^{(i)}_{k}$
    \item $\mu^{(i)}_{m(i)}>\mu^{(i)}_{k}$ and $\eta^{(m^{-1}(k))}_{k}>\eta^{(i)}_{k}$
\end{enumerate}
Consider $\mathcal{B}_1,\mathcal{B}_2,\mathcal{B}_3$ as the set of such $(p_i,a_k)$ pair which satisfy above condition 1, 2 and 3 resp. Further for each player $p_i$ denote by $\mathcal{B}^{(i)}_1,\mathcal{B}^{(i)}_2,\mathcal{B}^{(i)}_3$ as the set of arms satisfying condition 1, 2 and 3 resp.\\
The alternate instance is now defined as 
\begin{align*}
    \neg m = \cup_{p_i}\cup_{a_k\neq a_i}\{\lambda^{(i)}_{i}<\lambda^{(i)}_{k}\ \text{and}\ \nu^{(k)}_{k}<\nu^{(i)}_{k}\}
\end{align*}

   \begin{align*}
        D(\mu,\eta) &= \max_{w\in\Delta_{t\times K}}\inf_{(\lambda,\nu)\in\neg m}\sum_{i=1}^t\sum_{k=1}^K w^{(i)}_k \left[d(\mu^{(i)}_k,\lambda^{(i)}_k)+d(\eta^{(i)}_i,\nu^{(i)}_k)\right]\\
      &= \max_{w\in\Delta_{t\times K}}\min\Bigg\{\min_{(p_i,a_k)\in \mathcal{B}_1} w^{(i)}_k d(\mu^{(i)}_k,x^{(i)}_{m(i),k} )+w^{(i)}_{m(i)} d(\mu^{(i)}_{m(i)} ,x^{(i)}_{m(i),k} ),\\
      &\min_{(p_i,a_k)\in \mathcal{B}_2} w^{(i)}_k d(\eta^{(i)}_k,y^{(i,m^{-1}(k))}_{k})+w^{(m^{-1}(k))}_{k}d(\eta^{(m^{-1}(k))}_{k},y^{(i,m^{-1}(k))}_{k}),\\
      &\min_{(p_i,a_k)\in \mathcal{B}_3} w^{(i)}_k d(\mu^{(i)}_k,x^{(i)}_{m(i),k} )+w^{(i)}_{i} d(\mu^{(i)}_{m(i)} ,x^{(i)}_{m(i),k} )+\\
      &w^{(i)}_k d(\eta^{(i)}_k,y^{(i,m^{-1}(k))}_{k})+w^{(m^{-1}(k))}_{k}d(\eta^{(m^{-1}(k))}_{k},y^{(i,m^{-1}(k))}_{k})\Bigg\}
    \end{align*}

where $x^{(i)}_{m(i),k}= \dfrac{\mu^{(i)}_{m(i)}w^{(i)}_{i}+\mu^{(i)}_{k}w^{(i)}_{k}}{w^{(i)}_{i}+w^{(i)}_{k}}$ and  $y^{(i,m^{-1}(k))}_{k}= \dfrac{\eta^{(k)}_{k}w^{(k)}_{k}+\eta^{(i)}_kw^{(i)}_k}{w^{(k)}_{k}+w^{(i)}_k}$.

Recall the 
  Lower bound optimization problem: 
    \begin{align*}
 \texttt{\textbf{LO2}:}\quad\quad \quad\quad \quad\quad \quad\quad \quad\quad  \min_{t}  \sum_{i=1}^t\sum_{k=1}^K t_k^{(i)}\\
    \text{s.t.}\ t^{(i)}_k&\geq 0\\
  C^{(i)}_{m(i),k}&\geq 1 \ \forall (p_i,a_k)\in \mathcal{B}_1\\
C^{(m^{-1}(k))}_{k,i}&\geq 1 \ \forall (p_i,a_k)\in \mathcal{B}_2\\
C^{(i,m^{-1}(k))}_{m(i),k}&\geq 1 \ \forall (p_i,a_k)\in \mathcal{B}_3
\end{align*}
where the player index 
$C^{(i)}_{m(i),k}:=t^{(i)}_k d(\mu^{(i)}_k,x^{(i)}_{m(i),k} )+t^{(i)}_{m(i)} d(\mu^{(i)}_{m(i)} ,x^{(i)}_{m(i),k})$ defined when arm $a_k$ likes player $p_i$ over its current match
and arm index as 
$C^{(m^{-1}(k))}_{k,i}:=t^{(i)}_k d(\eta^{(i)}_k,y^{(i,m^{-1}(k))}_{k})+t^{(m^{-1}(k))}_{k}d(\eta^{(m^{-1}(k))}_{k},y^{(i,m^{-1}(k))}_{k})$ defined when player $p_i$ likes arm $a_k$ over its current match.

KKT conditions: dual variables $\{\gamma^{(i)}_k\}_{k=1}^{K}\phantom{}_{i=1}^t$ and $\alpha_{i,k},\beta_{i,k},\nu_{i,k}$ for the constraint of $\mathcal{B}_1,\mathcal{B}_2$ and $\mathcal{B}_3$ resp.

\begin{align*}
    \mathcal{L}(t,\alpha,\gamma) &= \sum_{i=1}^t\sum_{j=1}^t t_k^{(i)}(1-\gamma^{(i)}_k)-\sum_{(p_i,a_k)\in\mathcal{B}_1}\alpha_{i,k}\left( C^{(i)}_{m(i),k}-1\right)-\sum_{(p_i,a_k)\in\mathcal{B}_2}\beta_{i,k}\left( C^{(m^{-1}(k))}_{k,i}-1\right)\\
    &-\sum_{(p_i,a_k)\in\mathcal{B}_3}\nu_{i,k}\Big( C^{(i,m^{-1}(k))}_{m(i),k}-1\Big)
\end{align*}

For each $i$, denote by $\mathcal{B}_1^{(i)}=\{a_k:(p_i,a_k)\in\mathcal{B}_1\}$,  $\mathcal{B}_2^{(i)}=\{a_k:(p_i,a_k)\in\mathcal{B}_2\}$ and  $\mathcal{B}_3^{(i)}=\{a_k:(p_i,a_k)\in\mathcal{B}_3\}$\\
tote that $a_i\not\in \mathcal{B}_1^{(i)},\mathcal{B}_2^{(i)},\mathcal{B}_3^{(i)}$.

Stationary condition
\begin{align}
    1-\gamma^{(i)}_{m(i)}-\sum_{a_k\in\mathcal{B}_1^{(i)}}\alpha_{i,k}d(\mu^{(i)}_{m(i)} ,x^{(i)}_{m(i),k} )-\sum_{p_j:m(i)\in \mathcal{B}_2^{(j)}}\beta_{j,m(i)}d(\eta^{(i)}_{m(i)},y^{(j,i)}_{m(i)})&\nonumber\\
    -\sum_{a_k\in\mathcal{B}_3^{(i)}}\nu_{i,k} d(\mu^{(i)}_{m(i)} ,x^{(i)}_{m(i),k} )-\sum_{p_j:m(i)\in \mathcal{B}_3^{(j)}}\nu_{j,m(i)} d(\eta^{(i)}_{m(i)},y^{(j,i)}_{m(i)})&=0 \label{eq:statleader}\\
     1-\gamma^{(i)}_{k}-\alpha_{i,k}d(\mu^{(i)}_k,x^{(i)}_{m(i),k} )&=0 \quad \forall k\in \mathcal{B}_1^{(i)} \label{eq:statB1}\\
     1-\gamma^{(i)}_{k}-\beta_{i,k}d(\eta^{(i)}_k,y^{(i,m^{-1}(k))}_{k})&=0  \forall k\in \mathcal{B}_2^{(i)}\label{eq:statB2}\\
     1-\gamma^{(i)}_{k}-\nu_{i,k}\left[d(\mu^{(i)}_k,x^{(i)}_{m(i),k} )+d(\eta^{(i)}_k,y^{(i,m^{-1}(k))}_{k})\right]&=0  \forall k\in \mathcal{B}_3^{(i)}\label{eq:statB3}
\end{align}

Complementary slackness 
\begin{align}
 \gamma^{(i)}_{k}t^{(i)}_k&= 0\ \ \forall\ i,k \label{eq:complslack}\\
 \alpha_{i,k}\left(C^{(i)}_{m(i),k}-1\right)&=0  \ \forall (p_i,a_k)\in \mathcal{B}_1\\
 \beta_{i,k}\left(C^{(m^{-1}(k))}_{k,i}-1\right)&=0\ \forall (p_i,a_k)\in \mathcal{B}_2\\
\nu_{i,k}\Big(C^{(i)}_{m(i),k}
+C^{(m^{-1}(k))}_{k,i}-1\Big)&=0  \ \forall (p_i,a_k)\in \mathcal{B}_3
\end{align}
Feasibility 
   \begin{align}
   t^{(i)}_k&\geq 0\\
  C^{(i)}_{m(i),k}&\geq 1 \ \forall (p_i,a_k)\in \mathcal{B}_1\label{eq:feasb1}\\
C^{(m^{-1}(k))}_{k,i}&\geq 1 \ \forall (p_i,a_k)\in \mathcal{B}_2\label{eq:feasb2}\\
C^{(i,m^{-1}(k))}_{m(i),k}&\geq 1 \ \forall (p_i,a_k)\in \mathcal{B}_3\label{eq:feasb3}
\end{align}

For $(p_i,a_k)\in\mathcal{B}_1$, from Eq. \ref{eq:feasb1} we need $t^{(i)}_k>0$ and $t^{(i)}_{m(i)}>0$, this along with Eq.~\ref{eq:complslack} imply that $\gamma^{(i)}_k=0$\\
For $(p_i,a_k)\in\mathcal{B}_2$, from Eq. \ref{eq:feasb2} we need $t^{(i)}_k>0$ and $t^{(m^{-1}(k))}_{k}>0$, this along with Eq.~\ref{eq:complslack} imply that $\gamma^{(i)}_k=0$\\
For $(p_i,a_k)\in\mathcal{B}_3$, from Eq. \ref{eq:feasb3} we need $t^{(i)}_k>0$ and $t^{(i)}_{m(i)}>0$ OR $t^{(m^{-1}(k))}_{k}>0$, thus, this along with Eq.~\ref{eq:complslack} imply that $\gamma^{(i)}_k>0$ and $\min\{\gamma^{(i)}_{m(i)},\gamma^{(m^{-1}(k))}_{k}\}=0$.

Substituting $\gamma^{(i)}_k=0$ for $k\in \mathcal{B}_1^{(i)}\cup\mathcal{B}_2^{(i)}\cup\mathcal{B}_3^{(i)} \ \forall i$, we get $\alpha_{i,k},\beta_{i,k},\nu_{i,k}>0$ which from complementary slackness implies the equality of the indexes.
Defining
\begin{align*}
    &g^{(i)}_{m(i)}:=-1+\sum_{a_k\in\mathcal{B}_1^{(i)}}\dfrac{d(\mu^{(i)}_{m(i)} ,x^{(i)}_{m(i),k} )}{d(\mu^{(i)}_{k},x^{(i)}_{m(i),k} )}+\sum_{p_j:m(i)\in \mathcal{B}_2^{(j)}}\dfrac{d(\eta^{(i)}_{m(i)},y^{(j,i)}_{m(i)})}{d(\eta^{(j)}_{m(i)},y^{(j,i)}_{m(i)})}\\
   &+\sum_{k\in\mathcal{B}_3^{(i)}}\dfrac{d(\mu^{(i)}_{m(i)} ,x^{(i)}_{m(i),k} )}{d(\mu^{(i)}_{k},x^{(i)}_{m(i),k} )+d(\eta^{(i)}_{k},y^{(i,m^{-1}(k))}_{k})}+\sum_{p_j:m(i)\in \mathcal{B}_3^{(j)}}\dfrac{d(\eta^{(i)}_{m(i)},y^{(j,i)}_{m(i)})}{d(\mu^{(j)}_{m(i)},x^{(j)}_{m(j),m(i)} )+d(\eta^{(j)}_{m(i)},y^{(j,i)}_{m(i)})}.
\end{align*}
we have $\gamma^{(i)}_{m(i)}=-g^{(i)}_{m(i)}\ \forall i$ from the stationary condition Eq.~\ref{eq:statleader}. As we derived  $\min\{\gamma^{(i)}_{m(i)},\gamma^{(m^{-1}(k))}_{k}\}=0$, we get that $\max\{g^{(i)}_{m(i)},g^{(m^{-1}(k))}_{k}\}=0$ and $g^{(i)}_{m(i)}N^{(i)}_{m(i)}=0$ follows from the complementary slackness conditions.

%% file: paper/appendix/stopping-rule.tex
\section{Stopping Rule}
\label{app:stoppingrule}
Denote the set of all possible matching for $N$ player and $N$ arm as $\mathcal{M}$.\\
Define the event that the estimated matching is unique 
\begin{align*}
     \mathcal{E}_{\texttt{UM}} &:= \{ \exists\ N\in \mathbb{N}: \texttt{DA}_{\texttt{Player}}(\hat{\bm{\mu}}(N),\succ)=\texttt{DA}_{\texttt{Arm}}(\hat{\bm{\mu}}(N),\succ) )\}
\end{align*}

Our stopping time is defined as
\begin{align*}
    \tau_{\delta} &= \inf\left\{N\in\mathbb{N}:\max_{m\in\mathcal{M}} \inf_{\lambda\in\neg m}\sum_{i=1}^P\sum_{k=1}^{K}N^{(i)}_kD(\hat{\mu}^{(i)}_k,\lambda^{(i)}_k) >\beta(N,\delta)\ \text{and}\ \texttt{DA}_{\texttt{Player}}(\hat{\bm{\mu}}(N),\succ)=\texttt{DA}_{\texttt{Arm}}(\hat{\bm{\mu}}(N),\succ) \right\}\\
    &= \inf\Bigg\{N\in\mathbb{N}:\max_{m\in \mathcal{M}}\min_{p_i}\min_{k\in \mathcal{D}_m^{(i)}}N^{(i)}_{m(i)} D(\hat{\mu}_{m(i)}^{(i)},\hat{x}_{m(i),k}^{(i)})+ N^{(i)}_{k} D(\hat{\mu}_{k}^{(i)},\hat{x}_{m(i),k}^{(i)})>\beta(N,\delta)\\
    &\ \  \ \text{and}\ \texttt{DA}_{\texttt{Player}}(\hat{\bm{\mu}}(N),\succ)=\texttt{DA}_{\texttt{Arm}}(\hat{\bm{\mu}}(N),\succ) \Bigg\}
\end{align*}

Also
\begin{align*}
    \hat{m}_{\tau_{\delta}} &= \argmax_{m\in{\mathcal{M}}} \inf_{\lambda\in\neg m}\sum_{i=1}^P\sum_{k=1}^{K}N^{(i)}_kD(\hat{\mu}^{(i)}_k,\lambda^{(i)}_k)\\
   \overset{(1)}{\implies} \hat{m}_{\tau_{\delta}}  &=\argmax_{m\in\mathcal{M}_{\hat{\mu}}} \inf_{\lambda\in\neg m}\sum_{i=1}^P\sum_{k=1}^{K}N^{(i)}_kD(\hat{\mu}^{(i)}_k,\lambda^{(i)}_k)\tag{$\mathcal{M}_{\hat{\mu}}$ :set of all stable matching under $\hat{\mu},\succ$}\\
   & =\min_{p_i} \min_{k\in \mathcal{D}_m^{(i)}}N^{(i)}_{m(i)} D(\hat{\mu}_{m(i)}^{(i)},\hat{x}_{m(i),k}^{(i)})+ N^{(i)}_{k} D(\hat{\mu}_{k}^{(i)},\hat{x}_{m(i),k}^{(i)})\ \tag{$\because$ $\mathcal{M}_{\hat{\mu}}$ is singeton by definition of $\tau_{\delta}$}
\end{align*}
(1) follows since, suppose $\hat{m}_{\tau_{\delta}}\not\in\mathcal{M}_{\hat{\mu}}$ i.e. $\hat{m}_{\tau_{\delta}}$ is unstable under $\hat{\mu},\succ$. Since $\neg m$ is the set of all instances in which $m$ is an unstable match, thus we have $\hat{\mu}\in \neg m$, which evaluates the inner infimum to be 0.

Recall when $A\implies B$, then $\mathbb{P}(A)\leq \mathbb{P}(B)$. Denote $\overline{m}$ as the unique stable matching under $\mu,\succ$, we have
\begin{align*}
    &
    \hat{m}_{\tau_{\delta}}\neq \overline{m}\implies \ \exists m\neq \overline{m}\ \text{s.t.}\\
    &\inf_{\lambda\in\neg m}\sum_{i=1}^P\sum_{k=1}^{K}N^{(i)}_kD(\hat{\mu}^{(i)}_k,\lambda^{(i)}_k)>\beta(N,\delta)\ \text{and}\ \texttt{DA}_{\texttt{Player}}(\hat{\bm{\mu}}(N),\succ)=\texttt{DA}_{\texttt{Arm}}(\hat{\bm{\mu}}(N),\succ) 
\end{align*}
\begin{align*}
    \mathbb{P}_{\mu}\left(\tau<\infty, \hat{m}_{\tau_{\delta}}\neq \overline{m}\right)&\leq \mathbb{P}_{\mu}\left(\exists t\in \mathbb{N}^{\star}, \exists m\neq \overline{m}\ \text{s.t.}\ \inf_{\lambda\in\neg m}\sum_{i=1}^P\sum_{k=1}^{K}N^{(i)}_kD(\hat{\mu}^{(i)}_k(t),\lambda^{(i)}_k)>\beta(N,\delta)\cap \mathcal{E}_{\texttt{UM}}\right)\\
    &= \mathbb{P}_{\mu}\left(\exists t\in \mathbb{N}^{\star}, \exists m: \mu\in \neg m,\inf_{\lambda\in\neg m}  \sum_{i=1}^P\sum_{k=1}^KN^{(i)}_kD(\hat{\mu}^{(i)}_k(t), \lambda^{(i)}_k)>\beta(t,\delta)\right)\\
      &= \mathbb{P}_{\mu}\left(\exists t\in \mathbb{N}^{\star}, \exists m: \mu\in \neg m\sum_{i=1}^P\sum_{k=1}^KN^{(i)}_kD(\hat{\mu}^{(i)}_k(t), \mu^{(i)}_k)>\beta(t,\delta)\cap \mathcal{E}_{\texttt{UM}}\right)\\
    &\leq \sum_{m\neq \overline{m}}\mathbb{P}_{\mu}\left(\exists t\in \mathbb{N}^{\star},  \sum_{i=1}^P\sum_{k=1}^KN^{(i)}_kD(\hat{\mu}^{(i)}_k(t), \mu^{(i)}_k)>\beta(t,\delta)\cap \mathcal{E}_{\texttt{UM}}\right)\\
      &\leq  \sum_{m\neq \overline{m}}\mathbb{P}_{\mu}\left(\exists t\in \mathbb{N}^{\star},  \sum_{i=1}^P\sum_{k=1}^KN^{(i)}_kD(\hat{\mu}^{(i)}_k(t), \mu^{(i)}_k)>\beta(t,\delta)\cap \mathcal{E}_{\texttt{UM}}\right)
\end{align*}

\begin{theorem}
    Let $\mu$ be an exponential family market model. Under any matching rule $(M_t)$, for every subset of players $\mathcal{P}\subseteq [P]$ and subset of arms $\mathcal{A}\subseteq [K]$.
    Define
    \begin{align*}
        \mathcal{E}_1 &:= \left\{ \exists\ t\in \mathbb{N}:\ \sum_{i\in\mathcal{P}}\sum_{k\in\mathcal{A}} N_{k}^{(i)}(t)d\left(\hat{\mu}_{k}^{(i)}(t), \mu_{k}^{(i)}\right)\geq \sum_{i\in\mathcal{P}}\sum_{k\in\mathcal{A}} 3\log\left(1+\log\left(N^{(i)}_k(t)\right)\right)+|\mathcal{P}||\mathcal{A}|\mathcal{T}\left(\dfrac{x}{|\mathcal{P}||\mathcal{A}|}\right)\right\}
    \end{align*}
    we have
        \begin{align*}
        \mathbb{P}_{\mu}\left(\mathcal{E}_1\right)\leq e^{-x}
    \end{align*}
    where $\mathcal{T}(\cdot)$ is a non-explicit function defined in \cite{kaufmann2020contributions} (3.7) which scales as $T(x)\sim x$ for large $x$ values.
\end{theorem}

Thus taking the threshold as $\beta(t,\delta)= |\mathcal{P}||\mathcal{A}|\mathcal{T}\left(\dfrac{\log((\mathcal{M}-1)/\delta)}{|\mathcal{P}||\mathcal{A}|}\right)+3|\mathcal{P}||\mathcal{A}|\log(1+\log(t))$ where $\mathcal{M}$ is the number of all stable matching with $|\mathcal{P}|$ players and $|\mathcal{A}|$ arms. 

From above we have 
\begin{align*}
        &\mathbb{P}_{\mu}\left(\left\{\exists\ t\in \mathbb{N}:\ \sum_{i\in\mathcal{P}}\sum_{k\in\mathcal{A}} N_{k}^{(i)}(t)d\left(\hat{\mu}_{k}^{(i)}(t), \mu_{k}^{(i)}\right)\geq \beta(t,\delta)\right\}\cap  \mathcal{E}_{\texttt{UM}} \right)\\
        &\leq \mathbb{P}_{\mu}\left(\left\{\exists\ t\in \mathbb{N}:\ \sum_{i\in\mathcal{P}}\sum_{k\in\mathcal{A}} N_{k}^{(i)}(t)d\left(\hat{\mu}_{k}^{(i)}(t), \mu_{k}^{(i)}\right)\geq \beta(t,\delta)\right\} \right)\leq  \dfrac{\delta}{\mathcal{M}-1}
    \end{align*}

which from above gives
\begin{align*}
    \mathbb{P}_{\mu}\left(\tau<\infty, \hat{m}_{\tau_{\delta}}\neq \overline{m}\right) &\leq  \sum_{m\neq \overline{m}}\mathbb{P}_{\mu}\left(\exists t\in \mathbb{N}^{\star},  \sum_{i=1}^P\sum_{k=1}^KN^{(i)}_kD(\hat{\mu}^{(i)}_k(t), \mu^{(i)}_k)>\beta(t,\delta)\right) \leq  \sum_{m\neq \overline{m}}\dfrac{\delta}{\mathcal{M}-1}=\delta
\end{align*}

Before Proving Theorem 7.1. note that, we would like to derive a concentration inequality to upper bound the RHS such that it is
\begin{itemize}
    \item \textit{uniform over time} 
    \item deviations are measured simultaneously for all the players and arms 
\end{itemize}

The proof of the concentration inequality requires construction of a particular \textit{mixture martingale}.

\begin{lemma}[Ville's Inequality] 
    Let $S_t$ be a super-martingale (i.e. a sequence of random variables adapted to a filtration ($\mathcal{F}_t)_{t\in\mathbb{N}}$ such that $\E[S_{t+1}|\mathcal{F}_t]\leq S_t$) such that $S_t\geq 0$ and $\E[S_0]=1$. For all $\delta\in(0,1)$ 
    \begin{align*}
        \mathbb{P}(\exists\ t\in \mathbb{N}: S_t>1/\delta)\leq \delta
    \end{align*}
\end{lemma}
Upper-bounding the process that we want to control by a \textit{test (super) martingale}.

The process whose deviations should be controlled is 
\begin{align*}
   X_{i,k}(t):= N^{(i)}_k(t) d(\hat{\mu}^{(i)}_k(t),\mu^{(i)}_k)-3\log\left(1+\log\left(N^{(i)}_k(t)\right)\right)
\end{align*}

\begin{definition}
  Let $g:\Lambda \to \R$ be a function defined on a non-empty interval $\Lambda\subseteq \R$. A stochastic process $X(t) = \{X_{i,k}\}_{i\in[P]}\phantom{}_{k\in[K]}$ is called $g$-VCC if it satisfies the following properties 
  \begin{enumerate}
      \item For any player $p_i$, arm $a_k$ and $\lambda\in\Lambda$
 there exists a test martingale $M^{\lambda}_{i,k}(t)$ such that 
 \begin{align*}
     \forall\ t\in\mathbb{N}, M^{\lambda}_{i,k}(t)\geq \exp\left(\lambda X_{i,k}-g(\lambda)\right).
 \end{align*}
 \item For any subset of players $\mathcal{P}\subseteq \{1,\ldots, P\}$, subset of arms $\mathcal{A}\subseteq \{1,\ldots, K\}$ and for any $\lambda\in\Lambda$, the product $\prod_{i\in \mathcal{P}}\prod_{k\in\mathcal{A}}M^{\lambda}_{i,k}(t)$ is a martingale. 
 \end{enumerate}
\end{definition}

\begin{definition}
    For $g:\Lambda\to\R^{+}$, define for all $x>0$,
    \begin{align*}
        C^{g}(x) := \min_{\lambda\in\Lambda} \dfrac{g(\lambda)+x}{\lambda}.
    \end{align*}
\end{definition}

\begin{lemma}
    Fix subset of players $\mathcal{P}\subseteq \{1,\ldots, P\}$, subset of arms $\mathcal{A}\subseteq \{1,\ldots, K\}$. Assume that $X(t) = \{X_{i,k}\}_{i\in[P]}\phantom{}_{k\in[K]}$ is a $g$-VCC stochastic process. Then 
    \begin{align}
    \label{eq:gvcc}
        \forall x>0, \qquad \mathbb{P}\left(\exists\ t\in \mathbb{N}:\ \sum_{i\in\mathcal{P}}\sum_{k\in \mathcal{A}}X_{i,k}(t)>|\mathcal{P}||\mathcal{A}|C^g\left(\dfrac{x}{|\mathcal{P}||\mathcal{A}|}\right)\right) \leq e^{-x}
    \end{align}
\end{lemma}

\begin{proof}
    
Fix $\lambda\in\Lambda$. As $X(t)$ is $g$-VCC, for any $u\in\R$, we can write

\begin{align*}
    \mathbb{P}\left(\exists\ t\in \mathbb{N}:\ \sum_{i\in\mathcal{P}}\sum_{k\in \mathcal{A}}X_{i,k}(t)>u\right) &= \mathbb{P}\left(\exists\ t\in \mathbb{N}:\ \exp\left(\sum_{i\in\mathcal{P}}\sum_{k\in \mathcal{A}}X_{i,k}(t)\right)>\exp(\lambda u)\right)\\
    &\leq \mathbb{P}\left(\exists\ t\in \mathbb{N}:\ \prod_{i\in\mathcal{P}}\prod_{k\in\mathcal{A}}M^{\lambda}_{i,k}(t)>\exp\left(\lambda u-(|\mathcal{P}||\mathcal{A}|)g(\lambda)\right)\right)
\end{align*}

As $\prod_{i\in\mathcal{P}}\prod_{k\in\mathcal{A}}M^{\lambda}_{i,k}(t)$, we have the following from Ville's inequality
\begin{align*}
    \mathbb{P}\left(\exists\ t\in \mathbb{N}:\ \prod_{i\in\mathcal{P}}\prod_{k\in\mathcal{A}}M^{\lambda}_{i,k}(t)>\exp\left(\lambda u-(|\mathcal{P}||\mathcal{A}|)g(\lambda)\right)\right)\leq \exp\left(-\lambda u+(|\mathcal{P}||\mathcal{A}|)g(\lambda)\right)
\end{align*}
Thus overall we get
\begin{align*}
    \mathbb{P}\left(\exists\ t\in \mathbb{N}:\ \sum_{i\in\mathcal{P}}\sum_{k\in \mathcal{A}}X_{i,k}(t)>u\right) \leq \exp\left(-\lambda u+(|\mathcal{P}||\mathcal{A}|)g(\lambda)\right)
\end{align*}
Equivalently, for all $x>0$ and for all $\lambda\in\Lambda$, we get
\begin{align*}
    \mathbb{P}\left(\exists\ t\in \mathbb{N}:\ \sum_{i\in\mathcal{P}}\sum_{k\in \mathcal{A}}X_{i,k}(t)>\dfrac{(|\mathcal{P}||\mathcal{A}|)g(\lambda)+x}{\lambda}\right) \leq e^{-x}
\end{align*}

Optimizing over $\lambda\in\Lambda$, we get

\begin{align*}
    \mathbb{P}\left(\exists\ t\in \mathbb{N}:\ \sum_{i\in\mathcal{P}}\sum_{k\in \mathcal{A}}X_{i,k}(t)>|\mathcal{P}||\mathcal{A}|C^g\left(\dfrac{x}{|\mathcal{P}||\mathcal{A}|}\right)\right) \leq e^{-x}
\end{align*}

\end{proof}

\begin{proof} From Lemma 3.4 of \cite{kaufmann2020contributions} we have that
$\forall \ x>0$ Eq.~\ref{eq:gvcc} implies the following
    \begin{align*}
    \mathbb{P}\left(\exists\ t\in \mathbb{N}:\ \sum_{i\in\mathcal{P}}\sum_{k\in \mathcal{A}}X_{i,k}(t)>|\mathcal{P}||\mathcal{A}|\mathcal{T}\left(\dfrac{x}{|\mathcal{P}||\mathcal{A}|}\right)\right) \leq e^{-x}
\end{align*}
where $\mathcal{T}(\cdot)$ is a non-explicit function defined in (3.7) \cite{kaufmann2020contributions}.  Substituting $ X_{i,k}(t):= N^{(i)}_k(t) d(\hat{\mu}^{(i)}_k(t),\mu^{(i)}_k)-3\log\left(1+\log\left(N^{(i)}_k(t)\right)\right)$, we have our result. 
\end{proof}

\begin{center}
    \fbox{The threshold scales as  $
    \beta(t,\delta) \simeq \log((\mathcal{M}-1)/\delta)+3|\mathcal{P}||\mathcal{A}|\log(1+\log(t))$}
\end{center}

For two-sided learning,  substituting $ X_{i,k}(t):= N^{(i)}_k(t) [d(\hat{\mu}^{(i)}_k(t),\mu^{(i)}_k)+d(\hat{\eta}^{(i)}_k(t),\eta^{(i)}_k)]-3\log\left(1+\log\left(N^{(i)}_k(t)\right)\right)$, gives the same threshold.

%% file: paper/appendix/fluid-one-sided.tex
\section{Fluid One-Sided}
\begin{proposition}
For every positive $N\geq N_{\min}$, and for a player $p_i$, and a set of arms $\mathcal{B}^{(i)}\subseteq\mathcal{D}_m^{(i)}\cup\mathcal{A}_{\texttt{UM}}$ there is a unique set of variables $\bm{N}^{(i)}_{\mathcal{B}} = \left(N^{(i)}_k: a_k\in\mathcal{B}^{(i)}\cup\{m(i)\}\right)$, $I(N)$ satisfying the following conditions
\begin{align*}
    g^{(i)}_m &= 0\tag{$\forall p_i\in\mathcal{P}_{1}$}\\
    \sum_{i\in\mathcal{P}_1}\sum_{k\in\mathcal{A}}N^{(i)}_k &= N\\ 
    N^{(i)}_{k}D\left(\mu^{(i)}_{k},x^{(i)}_{m(i),k}\right)+N^{(i)}_{m(i)}D\left(\mu^{(i)}_{m(i)},x^{(i)}_{m(i),k}\right) &=I(N)\tag{$\forall a_k\in\mathcal{B}^{(i)},\forall p_i\in\mathcal{P}_{1}$}
\end{align*}

\end{proposition}

\begin{proof}

    Note that, the system of equations in the theorem statement has total $N\times K+1$ variables ($N_k^{(i)}$ for $i\in\{1,\dots,N\},~k\in\{1,\dots,K\}$ and $I$) and $N+1+N(K-1)=N\times K+1$ equations. We break our arguments into two separate. First we argue that, for a given no. of samples to a player $p_i$, there always exists an allocation which the player can make amongst the arms such that it's corresponding anchor function $g^{(i)}_{m(i)}(\cdot)=0$ and indexes of alternate arms are equal. We call this player specific allocation as the \emph{local allocation}. After this we argue that, the algorithm allocates samples to all the players in such a way that, if every player follow their local allocation, then their indexes will be equal. We call this the \emph{global allocation}. In the following argument, we separately show that a unique  \emph{local} and \emph{global} allocation exist for a given choice of $N$, and they are also smooth.


    \paragraph{Existence of local allocation} For every player $p_i$ we first fix  $N^{(i)}$. Then the $K+1$ variables $N^{(i)}_k$ and $I^{(i)}_{\mathcal{B}}$ specific to the samples allocated to player $p_i$ are decided by the following equations
    \begin{align*}
       g^{(i)}_m &= 0\\
       \sum_{k\in\mathcal{A}}N^{(i)}_k &= N^{(i)}\\
       N^{(i)}_{k}D\left(\mu^{(i)}_{k},x^{(i)}_{m(i),k}\right)+N^{(i)}_{m(i)}D\left(\mu^{(i)}_{m(i)},x^{(i)}_{m(i),k}\right) &= I^{(i)}_{\mathcal{B}}(N^{(i)})\tag{$\forall a_k\in\mathcal{B}^{(i)}$}.
    \end{align*}

    \cite[Proposition 2.1]{bandyopadhyay2024optimal} implies the existence of a unique solution to the above system of equations and the smoothness of these solutions as functions of $N^{(i)}$. Furthermore, \cite[Theorem 4.1]{bandyopadhyay2024optimal} provides the Anchored Top-2 fluid dynamics, which is a collection of ODE's pasted together via which the solution to the above system evolves as $N^{(i)}$ increases. We use $N_k^{(i)}(N^{(i)})$ for $k\in\{1,\dots,K\}$ and $I_{\mathcal{B}}^{(i)}(N^{(i)})$ to denote the allocation across arms and their minimum index after $N^{(i)}$ samples are given to player $p_i$.

    \paragraph{Existence of global allocation} We fix the total no. of samples $N$ allocated to all the players. For a given allocation $\left\{N^{(i)}\right\}_{i\in [N]}$ across the players, since every player follow the Anchored Top 2 fluid dynamics in their local time scale, index of every player $p_i$ will be $I^{(i)}_{\mathcal{B}}(N^{(i)})$ as defined in the argument for the local time-scale. We want this allocation to satisfy the following system of equations: 
    \begin{align*}
        \sum_{\mathcal{P}_1}N^{(i)} &= N\\
        I^{(i)}_{\mathcal{B}}(N^{(i)}) &= I(N)\tag{$\forall p_i\in\mathcal{P}_{1}$}.
    \end{align*}

    From \cite{bandyopadhyay2024optimal}, we know the uniqueness and smoothness of the local set of equations. We fix $N$ and divide the global analysis as follows 
    \begin{itemize}
        \item Fix $I$: ask if $\exists!\ N^{(i)} \ \text{s.t.}\ I^{(i)}_{\mathcal{B}}(N^{(i)})=I\ \forall p_i\in\mathcal{P}_1$. The answer is yes as
        \begin{itemize}
            \item $I^{(i)}_{\mathcal{B}}(\cdot )$ are strictly increasing, and thus $N^{(i)}(I):=\left(I^{(i)}_{\mathcal{B}}\right)^{-1}(I)$  are smooth functions of $I$, due to Inverse Function Theorem
            \item $I\to N^{(i)}(I)$ is strictly increasing, as , by the Inverse Function Theorem, derivative of $N^{(i)}(I)$ is:
            \begin{align*}
                \left(N^{(i)}\right)'(I) = \dfrac{1}{\left(I^{(i)}_{\mathcal{B}}\right)'\left(\left(I^{(i)}_{\mathcal{B}}\right)^{-1}(I)\right)}>0
            \end{align*}
        \end{itemize}
        \item We want to identify an $I$ such that $\sum_{p_i\in\mathcal{P}_1}N^{(i)}(I)=N$
        \begin{itemize}
            \item From above we know that $N^{(i)}(I)$ is strictly increasing. Thus $I(N) := \left(\sum N^{(i)}(\cdot)\right)^{-1}(N)$. Furthermore, since $I\to \sum_{p_i\in\mathcal{P}_1}N^{(i)}(I)$ is strictly increasing, Inverse function theorem implies that $N\to I(N)$ is a smooth function of $N$.  
        \end{itemize}
    \end{itemize}
    Thus we take $N^{(i)}(N):=N^{(i)}(I(N))$ following the definition of the function $N^{(i)}(I)$ in the above argument. By the above argument, the solution $N^{(i)}(N)$ for $p_i\in\mathcal{P}_1$ and $I(N)$ to the global system of equations is also unique. 
\end{proof}
\label{app:fluid-one-sided}
For each player $p_i$, define $f^{(i)}_k = -\dfrac{\partial}{\partial x} D\left(\mu^{(i)}_{m(i)},x\right)/D\left(\mu^{(i)}_{k},x\right)\Big|_{x=x^{(i)}_{m(i),k}}$, note that $f^{(i)}_k>0$.  Let $\Delta^{(i)}_{k}=\mu^{(i)}_{m(i)}-\mu^{(i)}_{k}$ and $h^{(i)}_{k} = f^{(i)}_k \dfrac{(t^{(i)}_{m(i)})^2\Delta^{(i)}_{k}}{(t^{(i)}_{m(i)}+t^{(i)}_k)^2}$. Further denote $D^{(i)}_{m(i),k}=D\left(\mu^{(i)}_{m(i)},x^{(i)}_{m(i),k}\right)$ where $k\in \mathcal{A}^{(i)}_{\min}$ and $D^{(i)}_{k}=D\left(\mu^{(i)}_{k},x^{(i)}_{m(i),k}\right)$. Let $h^{(i)}_{\min}= \sum_{k\in\mathcal{A}^{(i)}_{\min}(t)}h^{(i)}_{k}/D^{(i)}_{k,k}$, $h^{(i)}_{\texttt{w}} = \sum_{k\not\in\mathcal{A}^{(i)}_{\min}(t)\cup\{m(i)\}}h^{(i)}_{k}t^{(i)}_k$ and $D^{(i)}_{\min} = (\sum_{k\in\mathcal{A}^{(i)}_{\min}(t)}1/D^{(i)}_{k,k})^{-1}$

\begin{theorem}[One-Sided Fluid]
\label{thm:onesidedgeneral}
    If at some time $t_0$, suppose the minimum index player set $\mathcal{P}_{\min}(t_0)$ partitions into $\mathcal{P}_{+}(t_0)\cup\mathcal{P}_{-}(t_0)\cup \mathcal{P}_{0}(t_0)$. We consider the ODEs till $t_1$ defined as the smallest time after $t_0$ such that  one of the follows happens
    \begin{itemize}
    \item Minimum index at $t_1$ becomes equal to the index of player, arm pair for which the index at $t_0$ was strictly larger OR
    \item Anchor function of some player becomes equal to zero at $t_1$.
 \end{itemize}

        This gives  the following for all $t\in [t_0,t_1)$.

\begin{enumerate}
\item 
The minimum index evolves using the following ODE \vspace{-0.1in}\begin{align*}
    \left(\dot{C}_{\min}(t)\right)^{-1}&=\sum_{i\in\mathcal{P}_{+}(t_0)}\dfrac{1}{D^{(i)}_{m(i),k(t)}}+\sum_{\scriptstyle \substack{i \in \mathcal{P}_{-}(t_0) \\ k \in \mathcal{A}^{(i)}_{\min}(t_0)}}
\dfrac{1}{D^{(i)}_{k}}+\sum_{i\in\mathcal{P}_{0}(t_0)}\dfrac{h^{(i)}_{\min}\sum_{k\in\bar{\mathcal{A}}^{(i)}_{\min}}t^{(i)}_k+h^{(i)}_{\texttt{w}}/D^{(i)}_{\min}}{C_{\min}h^{(i)}_{\min}+h^{(i)}_{\texttt{w}}}
\end{align*}
where $\forall p_i\in \mathcal{P}_{+}(t_0)$, define $k(t) := \arg\min_{k\in\mathcal{A}^{(i)}_{\min}(t)}C^{(i)}_{m(i),k}(t)$.
    \item for $p_i\in \mathcal{P}_{-}(t_0), k\in\mathcal{A}^{(i)}_{\min}(t_0)$, we have $\dot{t}^{(i)}_k = (D^{(i)}_{k,k})^{-1}\dot{C}_{\min}(t)$ and for $p_i\in \mathcal{P}_{+}(t_0)$, we have $\dot{t}^{(i)}_{m(i)} = (D^{(i)}_{m(i),k(t)})^{-1}\dot{C}_{\min}(t)$. Further, for  $p_i\in \mathcal{P}_{0}(t_0), k\in\mathcal{A}^{(i)}_{\min}(t_0)$, we have
    \begin{align*}
     \dot{t}^{(i)}_{m(i)}(t) &= \dot{C}_{\min}\dfrac{t^{(i)}_{m(i)}h^{(i)}_{\min}}{C_{\min }h^{(i)}_{\min}+h^{(i)}_{\texttt{w}}}\ \ \text{and}\ \ 
      \dot{t}^{(i)}_{k}(t) = \dot{C}_{\min}\dfrac{h^{(i)}_{\texttt{w}}/D^{(i)}_{k,k}+t^{(i)}_k h^{(i)}_{\min}}{{C_{\min} h^{(i)}_{\min}}+h^{(i)}_{\texttt{w}}}
\end{align*}
 
        \item for $p_{i'}\in\mathcal{P}\backslash\mathcal{P}_{\min}(t)$ and $a_{k'}\in\mathcal{E}_{m}^{(i')}$, we have $\dot{C}^{(i')}_{m(i'),k'}(t)=0$

\item When $\mathcal{P}_{0}\neq \emptyset$, there exists a constant $\beta>0$ independent of $t$ such that $\dot{C}_{\min}>\beta$ and in addition, other larger indexes for players $p_i\in\mathcal{P}_{\min}$ are bounded from above. 
When $\mathcal{P}_{0}= \emptyset$, indexes for $p_i\in\mathcal{P}_{\min}$ increase with $N$. Indexes for player $p_i\not\in\mathcal{P}_{\min}$ are constant. Thus, eventually all indexes catch up.  
\end{enumerate}
\end{theorem}

\begin{proof}
    \begin{align}
        \sum_{i\in\mathcal{P}_{0}} \dot{t}^{(i)}_{m(i)}+ \sum_{i\in\mathcal{P}_{0}} \sum_{k\in\mathcal{A}^{(i)}_{\min}}\dot{t}^{(i)}_k +\sum_{i\in\mathcal{P}_{+}}\dot{t}^{(i)}_{m(i)}+\sum_{i\in \mathcal{P}_{-}}\sum_{k\in\mathcal{A}^{(i)}_{\min}}\dot{t}^{(i)}_k   = 1 \label{eq:sumtionesided}
    \end{align}
    \begin{itemize}
        \item    For $p_i\in \mathcal{P}_{-}(t_0), k\in\mathcal{A}^{(i)}_{\min}(t_0)$, we have $\dfrac{d}{dt}t^{(i)}_k = \dfrac{1}{D^{(i)}_{k,k}}\dfrac{d}{dt}C_{\min}(t)$
        \item For $p_i\in \mathcal{P}_{0}$, we use the fact that the minimum indexes for $p_i$ also stay together. If $|\mathcal{A}^{(i)}_{\min}|=1$, then for $a_{\ell}\in \mathcal{A}^{(i)}_{\min}$, we have
    \begin{align*}
        \dot{C}_{\min}=\dot{t}^{(i)}_{m(i)}D^{(i)}_{m(i),\ell} + \dot{t}^{(i)}_kD^{(i)}_{\ell,\ell}
    \end{align*}
      Using $\dot{g}_i=0$, we get
    \begin{align*}
        \dot{t}^{(i)}_{m(i)}\sum_{} h^{(i)}_k t^{(i)}_{k}=\dot{t}^{(i)}_{\ell}h^{(i)}_{\ell}t^{(i)}_{m(i)}
    \end{align*}
    \begin{align*}
        \dot{t}^{(i)}_{m(i)} =  \dfrac{\dot{C}_{\min}h^{(i)}_{\ell}t^{(i)}_{m(i)}}{D^{(i)}_{m(i),\ell}h^{(i)}_{\ell}t^{(i)}_{m(i)}+\sum_{} h^{(i)}_k t^{(i)}_{k}D^{(i)}_{\ell,\ell}},\ \ \dot{t}^{(i)}_{\ell} = \dfrac{\dot{C}_{\min}\sum_{} h^{(i)}_k t^{(i)}_{k}}{D^{(i)}_{m(i),\ell}h^{(i)}_{\ell}t^{(i)}_{m(i)}+\sum_{} h^{(i)}_k t^{(i)}_{k}D^{(i)}_{\ell,\ell}}
    \end{align*}
    Otherwise when $|\mathcal{A}^{(i)}_{\min}|>1$, for $a_k,a_l\in\mathcal{A}^{(i)}_{\min}$, we have
    \begin{align*}
    \dot{t}^{(i)}_{m(i)}D^{(i)}_{m(i),k} + \dot{t}^{(i)}_kD^{(i)}_{k,k}= \dot{t}^{(i)}_{m(i)}D^{(i)}_{m(i),l} + \dot{t}^{(i)}_lD^{(i)}_{l,l}
    \end{align*}
    Using $C_{i,k}=C_{i,l}$, we get
    \begin{align*}
        \dot{t}^{(i)}_k = \dfrac{1}{t^{(i)}_{m(i)}D^{(i)}_{k,k}}\left(t^{(i)}_{k}D^{(i)}_{k,k} -t^{(i)}_lD^{(i)}_{l,l})\right)\dot{t}^{(i)}_{m(i)} + \dfrac{D^{(i)}_{l,l}}{D^{(i)}_{k,k}} \dot{t}^{(i)}_k
    \end{align*}
        
    Using $\dot{g}_i=0$, we get
    \begin{align*}
        \dot{t}^{(i)}_{m(i)}\sum h^{(i)}_kt^{(i)}_k- \sum \dot{t}^{(i)}_{k}f^{(i)}_kh^{(i)}_kt^{(i)}_{m(i)}=0
    \end{align*}
Finally
    \begin{align*}
          \dot{t}^{(i)}_k =   \dot{t}^{(i)}_{m(i)}\dfrac{h^{(i)}_{\texttt{w}}/D^{(i)}_{k,k}+t^{(i)}_{k}h^{(i)}_{\min}}{t^{(i)}_{m(i)}h^{(i)}_{\min}}
    \end{align*}
    By definition, we have
    \begin{align*}
        \dot{C}_{\min}&= \dot{t}^{(i)}_{m(i)}D(\mu^{(i)}_{m(i)},x^{(i)}_{m(i),k}) + \dot{t}^{(i)}_kD(\mu^{(i)}_k,x^{(i)}_{m(i),k})\\
        &=\dot{t}^{(i)}_{m(i)}\left(  \dfrac{C_{\min}h^{(i)}_{\min}+h^{(i)}_{\texttt{w}}}{t^{(i)}_{m(i)}h^{(i)}_{\min}}\right)
    \end{align*}
      \item   For $p_i\in\mathcal{P}_{+}$, we note that the algorithm does not look at the individual indexes of each challenger arm for this $p_i$, thus, the indexes may separate out. However, the algorithm ensures that the minimum index of $p_i$, remains equal to minimum index of other players $p_j$. To solve the fluid dynamics, we however need $\dot{t}^{(i)}_{m(i)}$ in terms of $\dot{C}_{\min}$. Thus define the following quantity 
    \begin{align*}
        \dot{C}_{\min}(t) = \dot{t}^{(i)}_{m(i)}D(\mu^{(i)}_{m(i)},x^{(i)}_{m(i),k(t)})\quad \text{where}\  k(t)=\arg\min_{k\in\mathcal{A}^{(i)}_{\min}(t)}C_{i,k}(t)\tag{$ \forall t\in [t_0,t_1)$}
    \end{align*}
    \end{itemize}
    Using Eq.\ref{eq:sumtionesided} we get the dynamics of $C_{\min}$ and substituting it in the dynamics of allocation derived above we get the rest of the results.

    Proof of Statement 3. For $p_i\not\in\mathcal{P}_{\min}$, allocations to $p_i$ are constant thus $\dot{t}^{(i)}_k=0 \forall k$ which implies $\dot{C}^{(i)}_{m(i),k}=0$ by definition.

    Proof of Statement 4. When $\mathcal{P}_{0}\neq \emptyset$, for $p_i\in\mathcal{P}_{0}$ we have for $C^{(i)}_{\min}:=\min C^{(i)}_{m(i),k}$ that $\dot{C}^{(i)}_{\min}>\beta$ and indexes corresponding to $k\in\mathcal{E}^{(i)}_m\backslash\mathcal{A}_{\min}^{(i)}$ the indexes are bounded from above, which follows exactly as \cite[Lemma E.1., E.2., E.3.]{bandyopadhyay2024optimal}. As $\dot{C}_{\min}=C^{(i)}_{\min}$, this extends to other players $p_i\in \mathcal{P}_+\cup\mathcal{P}_-$.\\
    When $\mathcal{P}_{0}=\emptyset$, from $\dot{C}_{\min}$ we have $\dot{C}_{\min}>0$. Statement 3 implies indexes for players $p_i\not\in\mathcal{P}_{\min}$ are constant. 
\end{proof}

\subsection{Indexes stay together (Fluid arguments)}
Following we give an argument of why indexes stay together in the fluid setting. 
Let $\mathcal{PA}_{\min}(t_0) = \{(p_i,a_k) : a_k\in\mathcal{D}^{(i)}_m\cup\mathcal{A}_{\texttt{UM}}\ \ C_{i,m(i),k}(t_0)=C_{\min}(t_0)\}$ denote the set of player, arm pair with minimum indexes.  \\

We first show that minimum index for different players in $\mathcal{P}_{\min}(t_0)$ remain together.\\
We prove this by contradiction, suppose the algorithm allocates samples to players in $ \mathcal{P}_1\subseteq P_{\min}(t_0)$ but does not allocate to players in $\mathcal{P}_2= P_{\min}(t_0)\backslash \mathcal{P}_1\neq \emptyset$.
Then we have
\begin{align*}
    C^{'(j)}_{\min} &=0 \tag{$\forall p_j\in \mathcal{P}_2$}\\
    \dot{C}^{(i)}_{\min} = \dot{C}_{i,m(i),k} &=\dot{N}^{(i)}_{m(i)}D(\mu^{(i)}_{m(i)},x^{(i)}_{m(i),k})+\dot{N}^{(i)}_{k}D(\mu^{(i)}_{k},x^{(i)}_{m(i),k}) \tag{$\forall p_i\in \mathcal{P}_1\ a_k\in\mathcal{A}^{
(i)}_{\min}(t_0)$}
\end{align*}
Thus since $^{(i)}_{\min} >C^{'(j)}_{\min}$ as at least one of $\dot{N}^{(i)}_{m(i)}$ or $\dot{N}^{(i)}_{k}$ is strictly positive, player $p_j$ again becomes a minimum index player and samples get allocated to it, contradicting the assumption that $
\mathcal{P}_2\neq \emptyset$.\\

We now show that if for a player $p_i$,  $g^{(i)}_{m(i)}\leq 0$, then different indexes for this player also remains together. Fix a player $p_i$. Suppose $\mathcal{A}^{(i)}_1\subseteq \mathcal{A}^{(i)}_{\min}(t_0)$ are the set of arms for which $p_i$ is allocated samples with, but is not allocated with arms in $\mathcal{A}^{(i)}_2=\mathcal{A}^{(i)}_{\min}(t_0)\backslash \mathcal{A}^{(i)}_1$.  If $g^{(i)}_{m(i)}<0$, we have
\begin{align*}
     \dot{C}^{(i)}_{m(i),\ell} &=0 \tag{$\forall a_{\ell}\in  \mathcal{A}^{(i)}_2$}\\
   \dot{C}_{i,m(i),k} &=\dot{N}^{(i)}_{k}D(\mu^{(i)}_{k},x^{(i)}_{m(i),k})  \tag{$\forall a_k\in  \mathcal{A}^{(i)}_1$}
\end{align*}
Thus, $\dot{C}_{i,m(i),k} >\dot{C}^{(i)}_{m(i),\ell}$, and leads to a contradiction that $a_{\ell}\in  \mathcal{A}^{(i)}$.\\
If $g^{(i)}_{m(i)}=0$, however, we have
\begin{align*}
     \dot{C}^{(i)}_{m(i),\ell} &=\dot{N}^{(i)}_{m(i)}D(\mu^{(i)}_{m(i)},x^{(i)}_{m(i),\ell}) \tag{$\forall a_{\ell}\in  \mathcal{A}^{(i)}_2$}\\
   \dot{C}_{i,m(i),k} &=\dot{N}^{(i)}_{m(i)}D(\mu^{(i)}_{m(i)},x^{(i)}_{m(i),k})+\dot{N}^{(i)}_{k}D(\mu^{(i)}_{k},x^{(i)}_{m(i),k})  \tag{$\forall a_k\in  \mathcal{A}^{(i)}_1$}
\end{align*}

If $D(\mu^{(i)}_{m(i)},x^{(i)}_{m(i),k})-D(\mu^{(i)}_{m(i)},x^{(i)}_{m(i),\ell})\geq 0$ it trivially follows as then $\dot{C}^{(i)}_{m(i),\ell}<\dot{C}_{i,m(i),k}$. \\
For $D(\mu^{(i)}_{m(i)},x^{(i)}_{m(i),k})-D(\mu^{(i)}_{m(i)},x^{(i)}_{m(i),\ell})<0$, we use the argument similar to similar to \cite{bandyopadhyay2024optimal} forward utilizing the fact that $g^{'(i)}_{m(i)}=0$.

If $g^{(i)}_{m(i)}>0$, then we have from above that
\begin{align*}
    C_{\min}^{'(i)} = \dot{N}^{(i)}_{m(i)}\min_{k\in\mathcal{A}^{(i)}_{\min}}D(\mu^{(i)}_{m(i)},x^{(i)}_{m(i),k}) = C_{\min}^{'(j)}
\end{align*}

\subsection{Relating Global and Local timescales}

\begin{theorem}[Relating the local time scale with the global time scale]
\label{thm:onesidedlocalglobal}
    If $p_i$ correspond to the minimum index player at time $t_0$, then its minimum index defined as $C^{(i)}_{\min}(t_0)$ follows the following ODE 
    \begin{align}
 \dot{C}^{(i)}_{\min}(t_0)  = \dot{C}_{\min}(t_0)= \left({\sum_{p_j\in\mathcal{P}_{\min}(t_0)}\left(\dfrac{d}{dt^{(j)}}C^{(j)}_{\min}(t_0^{(j)})\right)^{-1}}\right)^{-1}\ \ \forall p_i\in\mathcal{P}_{\min}(t_0) \label{eq:cminlocal}
\end{align}
where the player $p_i$'s allocation defined as $t^{(i)}(t_0)=\sum_{k\in\mathcal{D}^{(i)}\cup \UMA_m}t^{(i)}_k(t_0)+t^{(i)}_{m(i)}$ evolves as follows
\begin{align*}
    \dot{t}^{(i)}(t_0) &= 
    \dot{C}_{\min}(t_0)\left(\dfrac{d}{dt^{(i)}}C^{(i)}_{\min}(t_0^{(i)})\right)^{-1}.
\end{align*}
\begin{itemize}
    \item for $p_i\in \mathcal{P}_{-}(t_0), k\in\mathcal{A}^{(i)}_{\min}(t_0)$
    \begin{align*}
        \dfrac{d}{dt^{(i)}}t^{(i)}_k(t_0)  &= \dfrac{1}{D^{(i)}_{k,k}} \dfrac{d}{dt^{(i)}}C^{(i)}_{\min}(t_0^{(i)})\ \ \  \text{and}\ \ \ 
        \dfrac{d}{dt^{(i)}}C^{(i)}_{\min}(t_0^{(i)})= \dfrac{1}{\sum_{k\in\mathcal{A}^{(i)}_{\min}(t_0)}\dfrac{1}{D^{(i)}_{k,k}}}
    \end{align*}
    \item for $p_i\in \mathcal{P}_{+}(t_0), k\in\mathcal{A}^{(i)}_{\min}(t_0)$, we have $\dfrac{d}{dt^{(i)}}t^{(i)}_{m(i)}(t_0) =1, \dfrac{d}{dt^{(i)}}C^{(i)}_{\min}(t_0^{(i)})=D^{(i)}_{m(i)}$
    \item for  $p_i\in \mathcal{P}_{0}(t_0), k\in\mathcal{A}^{(i)}_{\min}(t_0)$
\begin{align*}
     \dfrac{d}{dt^{(i)}}t^{(i)}_{m(i)}(t) &= \dfrac{d}{dt^{(i)}}C^{(i)}_{\min}(t_0^{(i)})\dfrac{t^{(i)}_{m(i)}h^{(i)}_{\min}}{C_{\min h^{(i)}_{\min}}+h^{(i)}_{\texttt{w}}}\quad\text{and}\\ 
     \dfrac{d}{dt^{(i)}}t^{(i)}_{k}(t)&= \dfrac{d}{dt^{(i)}}C^{(i)}_{\min}(t_0^{(i)})\dfrac{h^{(i)}_{\texttt{w}}/D^{(i)}_{k,k}+t^{(i)}_k h^{(i)}_{\min}}{{C_{\min} h^{(i)}_{\min}}+h^{(i)}_{\texttt{w}}}
\end{align*}
where 
\begin{align*}
     \dfrac{d}{dt^{(i)}}C^{(i)}_{\min}(t_0^{(i)})&= \dfrac{{C_{\min} h^{(i)}_{\min}}+h^{(i)}_{\texttt{w}}}{h^{(i)}_{\texttt{w}}/D^{(i)}_{k,k}+ h^{(i)}_{\min}\sum_{k\in\mathcal{A}^{(i)}_{\min}\cup \{m(i)\}}t^{(i)}_k}
\end{align*}
and for players $p_i,p_j$ we have
\begin{align*}
    \dfrac{d}{dt}t^{(i)}(t_0)\dfrac{d}{dt^{(i)}}C^{(i)}_{\min}(t_0^{(i)}) &= \dfrac{d}{dt}t^{(j)}(t_0)\dfrac{d}{dt^{(j)}}C^{(j)}_{\min}(t_0^{(j)})
\end{align*}

\end{itemize}
\end{theorem}

\begin{proof}
    Since $\dot{C}_{\min}=\dot{C}^{(i)}_{\min}$ we have
    \begin{align*}
   \dot{C}_{\min}(t_0) &= \dot{t}^{(i)}\dfrac{d}{dt^{(i)}}C^{(i)}_{\min}(t_0^{(i)})
\end{align*}
Using $\sum_i\dot{t}^{(i)}=1$, we get the Eq.~\ref{eq:cminlocal}. The dynamics for $t^{(i)}_k$ w.r.t. $t^{(i)}$ follows as in \cite{bandyopadhyay2024optimal} as in local time scale, defined by $t^{(i)}$ for each player $p_i$, the algorithm replicates the best-arm idenfitication \texttt{BAI} problem.
\end{proof}

%% file: paper/appendix/nonfluid-one-sided.tex
\section{Non-Fluid One-Sided learning}
\label{app:nonfluid-one-sided}
From \cite[Proposition G.5, Statement 2]{bandyopadhyay2024optimal}, we have for $\texttt{ATT1}$, Algo.~\ref{algo:mainATT1} that $N^{(i)}(N)\geq N^{\gamma}-C_{\texttt{exp}}$ where $C_{\texttt{exp}}$ is a constant dependent on $M$ and $\gamma$. \\
The key in the analysis is to analyze after $N\geq T_{\texttt{conv}}: \forall\ p_i$ we have $N^{(i)}\geq T^{(i)}_{\stable}$. Choosing $T_{\texttt{conv}}:=\max_i\left(T_{\texttt{stable}}^{(i)}+C_{\texttt{exp}}\right)^{1/\gamma}$ implies $N^{(i)}(N) \geq T^{(i)}_{\stable} \forall i$. Introduce the constant $\epsilon(\mu^{(i)})=\frac{1}{4}\min_{k\in\mathcal{E}^{(i)}_m}|\mu^{(i)}_{m(i)}-\mu^{(i)}_k|$.
\begin{lemma}
    For all player $p_i$, after $N^{(i)}(N) \geq T^{(i)}_{\stable}$, we have the following, 
        \[\max_{k}|\hat{\mu}^{(i)}_k-\mu^{(i)}_k|\leq \epsilon(\mu^{(i)})(N^{(i)}(N))^{-\xi}\] and $N^{(i)}_k=\Theta(N^{(i)}(N))$.
        \label{lemma:localprop}
\end{lemma}
Notation: Let $\mathcal{I}^{(i)}_{m(i),k} := N^{(i)}_kD(\mu^{(i)}_k,x^{(i)}_{m(i),k})+N^{(i)}_{m(i)}D(\mu^{(i)}_{m(i)},x^{(i)}_{m(i),k})$ and $g^{(i)}_m:=g^{(i)}_{m}(\mu,\mathcal{D}^{(i)}_m,\mathcal{A}^{\texttt{UM}}_m,N^{(i)})$, further let $\mathcal{I}^{(i)}_{\min}=\min_k \mathcal{I}^{(i)}_{m(i),k}$.
\begin{lemma}
  For all $N\geq T_{\texttt{conv}}$,  if a player $p_i$ got samples at iteration $N$ then for all other players $p_j$ we have $N^{(j)}(N)\geq \kappa N^{(i)}(N)$ for some constant $\kappa>0$.
    \label{lemma:diffplayer}
\end{lemma}
 \begin{proof}
\textit{Case I: $N$ is an exploration}: 
\begin{align*}
    N^{(j)}(N) \geq  \sqrt{N}-C_{\texttt{exp}} \overset{(1)}{\geq} N^{(i)}(N-1)-C_{\texttt{exp}} = N^{(i)}(N)-1-C_{\texttt{exp}}
\end{align*}
where (1) follows as $N$ is exploration and $p_i$ got pulled. This gives $N^{(j)}/N^{(i)}\geq 1- (1+C_{\texttt{exp}})/N^{(i)}$.
Further using $N^{(i)}\geq  \sqrt{N}-C_{\texttt{exp}}$, we have $1- (1+C_{\texttt{exp}})/N^{(i)}\geq 1-(1+C_{\texttt{exp}})/(  \sqrt{N}-C_{\texttt{exp}})$. Introduce $M: N\geq M\implies  1-(1+C_{\texttt{exp}})/( \sqrt{N}-C_{\texttt{exp}})\geq \kappa$ equivalently, $ \sqrt{N} \geq C_{\texttt{exp}}+(1+C_{\texttt{exp}})/(1-\kappa)$. \\
\textit{Case II: $N$ is not an exploration}: As $N^{(i)}\geq  T_{\texttt{stable}}^{(i)}$ and from Lemma~\ref{lemma:localprop}, we have $N^{(i)}_k=\Theta(N^{(i)})\implies w^{(i)}_{k,\texttt{L}}:=\frac{N^{(i)}_k}{N^{(i)}}=\Theta(1)$. In addition, this holds for any other player $p_j$, as $N^{(j)}\geq T_{\texttt{stable}}^{(j)}$.

Further, we have $\mathcal{I}_{\min}^{(i)}(N-1)<\mathcal{I}_{\min}^{(j)}(N-1)$. Suppose the arm which got sample with $p_i$ be $a_k$. Then we have
    \begin{align*}
        N^{(i)}C_1\dfrac{w^{(i)}_{m(i),\texttt{L}}w^{(i)}_{k,\texttt{L}}}{w^{(i)}_{m(i),\texttt{L}}+w^{(i)}_{k,\texttt{L}}} \leq \mathcal{I}_{\min}^{(i)}(N-1)\leq \mathcal{I}_{\min}^{(j)}(N-1)\leq   N^{(j)}C_2\dfrac{w^{(j)}_{m(j),\texttt{L}}w^{(j)}_{l,\texttt{L}}}{w^{(j)}_{m(j),\texttt{L}}+w^{(j)}_{l,\texttt{L}}} 
    \end{align*}
    This gives $\kappa N^{(i)}\leq  N^{(j)}$.
\end{proof}

\begin{lemma}
    For any $p_i$ and $N\geq T_{\texttt{conv}}$ we have $N^{(i)}(N)=\Theta(N)$.
    \label{lemma:globaltime}
\end{lemma}
\begin{proof}
    Consider a player $p_i$ such that $N^{(i)}\geq N/M$. Then $N^{(i)}(N)\geq \beta N\geq T_3+1$, where $\beta:=1/M$ and $T_3:=\lfloor T_{\texttt{conv}}/\beta\rfloor-1$. This implies that $p_i$ got samples between $T_3$ and $N$. Let $N'$ be the last time before $N$ when $p_i$ got a sample, then by definition, $N'>T_3$.
    \begin{align*}
        N^{(j)}(N')\overset{(1)}{\geq} \gamma  N^{(i)}(N') \overset{(2)}{=} \gamma N^{(i)}(N) = \gamma' N
    \end{align*}
    where (1) follows from Lemma~\ref{lemma:diffplayer}, (2) follows from the definition of $N'$ and $\gamma':=\gamma\beta$.\\
    Since $N\geq N'$, we have $N^{(j)}(N)\geq N^{(j)}(N')$, which implies $ N^{(j)}(N)\geq \gamma'N$. Thus $ N^{(j)}(N)=\Theta(N)$. 
\end{proof}

\begin{lemma}
    For any player $p_i$, arm $a_k$ and $N\geq T_{\texttt{conv}}$, we have $N^{(i)}_k(N) = \Theta(N)$. 
    \label{lemma:linearN}
\end{lemma}
\begin{proof}
    This follows from Lemma~\ref{lemma:localprop} and Lemma~\ref{lemma:globaltime}.
\end{proof}
\begin{proposition}[Proposition 5.1\cite{bandyopadhyay2024optimal} ]\label{thm:local_convergence}
    For each player $p_i$,  for all $N^{(i)}(N)>T_{\texttt{stable}}^{(i)}$ we have for some instance dependent constant $D>0$ the following
    \begin{align*}
        |g^{(i)}_{m}|&<(DN^{(i)})^{-\xi}\ \quad\ 
        \max_{a_k,a_l}|\mathcal{I}^{(i)}_{m(i),k}-\mathcal{I}^{(i)}_{m(i),l}|<(DN^{(i)})^{1-\xi}
    \end{align*}
\end{proposition}

\begin{lemma}
   For all $N>T_{\texttt{conv}}$ we have for some instance dependent constant $D'>0$ that $\max_{p_i} |g^{(i)}_{m}|<(D'N)^{-\xi}$ and $ \max_{p_i}\max_{a_k,a_l}|\mathcal{I}^{(i)}_{m(i),k}-\mathcal{I}^{(i)}_{m(i),l}|<(D'N)^{-\xi}$
   \label{lemma:firstorderpi}
\end{lemma}
\begin{proof}
For all players $p_i\in \mathcal{P}$, we have $|g^{(i)}_{m}|\leq (DN^{(i)})^{-\xi} \overset{(1)}{\leq} (D'N)^{-\xi}$ 
    where (1) follows from Lemma~\ref{lemma:linearN}. 
    Thus, taking maximum we get the first result. Further from 
    \begin{align*}
         \dfrac{1}{N}\max_{p_i}\max_{a_k,a_l}|\mathcal{I}^{(i)}_{m(i),k}-\mathcal{I}^{(i)}_{m(i),l}|\leq \dfrac{1}{N^{(i)}}(DN^{(i)})^{1-\xi}\leq (DN^{(i)})^{-\xi}<(D''N)^{-\xi}
    \end{align*}
    we get the second result. Re-defining $D'=\max\{D',D''\}$ we have the result.
\end{proof}
\begin{lemma}[Noise in the anchor function and indexes]
    For all $N\geq T_{\texttt{conv}}$ we have
      \begin{align*}
      |\tilde{g}^{(i)}-g^{(i)}|&\leq  \mathcal{O}(N^{-\xi})\ \ \forall \ p_i\\
        |\mathcal{I}^{(i)}_{m(i),k}(N)-\tilde{C}^{(i)}_{m(i),k}(N)|&\leq  \mathcal{O}(N^{1-\xi})
    \end{align*}
    \label{lemma:noise}
\end{lemma}
\begin{proof}
    The first and second result is a direct consequence of \cite[Lemma G.1.]{bandyopadhyay2024optimal} and \cite[Lemma G.3.]{bandyopadhyay2024optimal} respectively in addition to Lemma~\ref{lemma:globaltime}.
\end{proof}

\begin{lemma}
For a fixed player $p_i$, if it is allocated sample at iteration $N\geq T_{\texttt{conv}}$, then it gets samples again in the next $\mathcal{O}(N^{1-\xi})$ iterations.
\label{lemma:playerfreq}
\end{lemma}

\begin{proof}

We define the local proportion followed by every player $p_i$ for arm $a_k$ as: 
\[\widetilde{w}^{(i)}_{k}(N^{(i)})~=~\frac{N_k^{(i)}}{N^{(i)}}.\]
By \cite[Proposition 2.2]{bandyopadhyay2024optimal}, every player $p_i$ has a unique local proportion $w^{\star (i)}_{\texttt{L}}~=~(w^{\star (i)}_{k,\texttt{L}})_{k\in[K]}$ such that $\sum_{k}w^{\star (i)}_{k,\texttt{L}}=1$ where the anchor function $g^{(i)}_m(\cdot)$ becomes zero and index of all the arms become equal. By Theorem \ref{thm:local_convergence}, for every player $p_i$, the local proportion $\widetilde{w}^{(i)}(N)$ satisfies the conditions in the earlier statement upto a perturbation of $O((N^{(i)}(N))^{-\xi})$. Following the argument in the proof of \cite[Proposition 3.1]{bandyopadhyay2024optimal}, using the implicit function theorem, we get: 
\[\max_{k}|\widetilde{w}^{(i)}_{k,\texttt{L}}(N^{(i)})-w^{\star (i)}_{k,\texttt{L}}|~=~O((N^{(i)})^{-\xi}).\]
As a result, minimum index over all the arms for player $p_i$ behaves like: 
\[\tilde{C}_{\min}^{(i)}(N)~=~c_i N^{(i)}(N) \pm O((N^{(i)}(N))^{1-\xi}),\]
where $c_i$ is the index of all arms for player $p_i$ under the local proportion $\bm{w}^{\star (i)}_\texttt{L}$. 
We now consider the situation where some player $p_i$ is sampled at time $N$ and then not sampled between iterations $N+1$ and $N+R$ where $R\leq N$. On the other hand, let $p_j$ be the player sampled for maximum no. of times between $N+1$ and $N+R$. Then, 
\[\tilde{C}_{\min}^{(i)}(N+R)~=~\tilde{C}_{\min}^{(i)}(N)\quad \text{and}\quad \tilde{C}_{\min}^{(j)}(N+R)-\tilde{C}_{\min}^{(j)}(N)~=~c_j(N^{(j)}(N+R)-N^{(j)}(N))\pm O(N^{1-\xi}).\]
As a result, 
\begin{align}\label{eq:1sided_periodicity_1}
    \tilde{C}_{\min}^{(i)}(N+R)-\tilde{C}_{\min}^{(j)}(N+R)~&=~\tilde{C}_{\min}^{(i)}(N)-\tilde{C}_{\min}^{(j)}(N)+(\tilde{C}_{\min}^{(i)}(N+R)-\tilde{C}_{\min}^{(i)}(N))\nonumber\\ 
    &~-(\tilde{C}_{\min}^{(j)}(N+R)-\tilde{C}_{\min}^{(j)}(N))\nonumber\\
    &\leq~-(\tilde{C}_{\min}^{(j)}(N+R)-\tilde{C}_{\min}^{(j)}(N))\quad \text{since $\tilde{C}_{\min}^{(i)}(N)\leq \tilde{C}_{\min}^{(j)}(N)$}\nonumber\\
    &\leq~-c_j(N^{(j)}(N+R)-N^{(j)}(N))+O(N^{1-\xi})\nonumber\\
    &\leq~-c_j \Delta N^{(j)}(N,N+R)+O(N^{1-\xi}).
\end{align}
Let $R_j$ be the last time between $N+1$ and $N+R$ at which arm $p_j$ was pulled. By definition of $R_j$, 
\[\Delta N^{(j)}(N,N+R_j)~=~\Delta N^{(j)}(N,N+R).\]
Using \ref{eq:1sided_periodicity_1}, which is true for all choices of $R$, we have: 
\begin{align*}
    \tilde{C}_{\min}^{(i)}(N+R_j)-\tilde{C}_{\min}^{(j)}(N+R_j)~&\leq~-c_j \Delta N^{(j)}(N,N+R_j)+O(N^{1-\xi})\\
    ~&\leq ~-c_j \Delta N^{(j)}(N,N+R)+c^\prime N^{1-\xi}.
\end{align*}
By the definition of $p_j$, we have $\Delta N^{(j)}(N,N+R)\geq\frac{R}{M}$. As a result, 
\[\tilde{C}_{\min}^{(i)}(N+R_j)-\tilde{C}_{\min}^{(j)}(N+R_j)~\leq~-c_j \frac{R}{M}+c^\prime N^{1-\xi}.\]
Upon taking $R\geq \frac{2c' M}{c_j}N^{1-\xi}$, we have
\[\tilde{C}_{\min}^{(i)}(N+R_j)-\tilde{C}_{\min}^{(j)}(N+R_j)~\leq~-c^\prime N^{1-\xi} \quad \text{implying}\quad \tilde{C}_{\min}^{(i)}(N+R_j)<\tilde{C}_{\min}^{(j)}(N+R_j).\]
Thus the algorithm samples from player $p_j$, even though player $p_i$ has index lesser than that of $p_j$. Therefore, our assumption must be wrong and arm $p_i$ must have been pulled again before iteration $N+\frac{2c^\prime M}{c_j}N^{1-\xi}$. 
\end{proof}

\begin{lemma}
    For any players $p_i,p_j$ and for all $N\geq T_{\texttt{stable}}$, we have
\begin{align*}
    \max_{(p_i,a_k),(p_j,a_l)}|\mathcal{I}^{(i)}_{m(i),k}-\mathcal{I}^{(j)}_{m(j),l}|\leq \mathcal{O}(N^{1-\xi})
\end{align*}
\end{lemma}

\begin{proof}
    Lemmas~\ref{lemma:firstorderpi} and \ref{lemma:noise} tells us that for $N\geq T_{\texttt{conv}}$, we have
\begin{align*}
    \tilde{C}^{(i)}_{m(i),k}(N) &=\tilde{C}^{(i)}_{\min}(N) \pm \mathcal{O}\left(N^{1-\xi}\right)\\
    \tilde{C}^{(j)}_{m(j),l}(N) &= \tilde{C}^{(j)}_{\min}(N) \pm \mathcal{O}\left(N^{1-\xi}\right).
\end{align*}
where recall $\tilde{C}^{(i)}_{\min}:=\min_{p_i}\tilde{C}^{(i)}_{m(i),k}$.

Define for any two players $p_i,p_j$ the following
\begin{align*}
    \tau_{i,j} = \{t-N: t\geq N\ \text{and}\ \tilde{C}^{(i)}_{\min}(N)-\tilde{C}^{(j)}_{\min}(N), \tilde{C}^{(i)}_{\min}(N+R)-\tilde{C}^{(j)}_{\min}(N+R)\ \text{have different signs}\}
\end{align*}

Define $T_{\texttt{stable}}$ after $T_{\texttt{conv}}$ such that all players $p_i$ receives a sample atleast once in the interval $[T_{\texttt{conv}},T_{\texttt{stable}}]$. From \ref{lemma:playerfreq}, and using $\E[T_{\texttt{conv}}]<\infty$, we get $\E[T_{\texttt{stable}}]<\infty$. By the definition of $ \tau_{i,j}$, we have $N+ \tau_{i,j}$ to be less than the iteration at which $p_i,p_j$ both got samples after $N$. As all players got samples atleast once between $[T_{\texttt{conv}},N]$, we have that $\tau_{i,j}(N)=\mathcal{O}(N^{1-\xi})$. Further, this holds for all players $p_i,p_j$. From

\begin{align*}
    |\tilde{C}^{(i)}_{\min}(N)-\tilde{C}^{(j)}_{\min}(N)|&\leq |\tilde{C}^{(i)}_{\min}(N)-\tilde{C}^{(j)}_{\min}(N)- (\tilde{C}^{(i)}_{\min}(N+\tau_{i,j})-\tilde{C}^{(j)}_{\min}(N+\tau_{i,j}))|\\
    &\leq |\tilde{C}^{(i)}_{\min}(N+\tau_{i,j})-\tilde{C}^{(i)}_{\min}(N)| + |\tilde{C}^{(j)}_{\min}(N+\tau_{i,j})-\tilde{C}^{(j)}_{\min}(N)|\\
\end{align*}
Suppose $\tilde{C}^{(i)}_{\min}(N+\tau_{i,j})=\tilde{C}^{(i)}_{m(i),k}(N+\tau_{i,j})$ and $\tilde{C}^{(i)}_{\min}(N)=\tilde{C}^{(i)}_{m(i),l}(N)$, we then have
\begin{align*}
     |\tilde{C}^{(i)}_{\min}(N+\tau_{i,j})-\tilde{C}^{(i)}_{\min}(N)| &= |\tilde{C}^{(i)}_{m(i),k}(N+\tau_{i,j})-\tilde{C}^{(i)}_{m(i),l}(N)| \\
     &\leq |\mathcal{I}^{(i)}_{m(i),k}(N+\tau_{i,j})-\mathcal{I}^{(i)}_{m(i),l}(N)|  +\mathcal{O}\left((N+\tau_{i,j})^{1-\xi}\right)
\end{align*}
further, using closeness of indexes for each player we have
\begin{align*}
    |\mathcal{I}^{(i)}_{m(i),k}(N+\tau_{i,j})-\mathcal{I}^{(i)}_{m(i),l}(N)| &=  |\mathcal{I}^{(i)}_{m(i),k}(N+\tau_{i,j})-\mathcal{I}^{(i)}_{m(i),k}(N)+\mathcal{I}^{(i)}_{m(i),k}(N)-\mathcal{I}^{(i)}_{m(i),l}(N)| \\
    &= |\mathcal{I}^{(i)}_{m(i),k}(N+\tau_{i,j})-\mathcal{I}^{(i)}_{m(i),k}(N)|+\mathcal{O}(N^{1-\xi})\\
\end{align*}
Using mean value theorem, we get $|\mathcal{I}^{(i)}_{m(i),k}(N+\tau_{i,j})-\mathcal{I}^{(i)}_{m(i),k}(N)| \leq \mathcal{O}(1)\cdot \tau_{i,j}\leq \mathcal{O}(N^{1-\xi})$, we get
\begin{align*}
     |\mathcal{I}^{(i)}_{m(i),k}(N+\tau_{i,j})-\mathcal{I}^{(i)}_{m(i),l}(N)|\leq \mathcal{O}(N^{1-\xi})
\end{align*}
 Further using $\mathcal{O}\left((N+\tau_{i,j})^{1-\xi}\right)= \mathcal{O}(N^{1-\xi})$, we have
 \begin{align*}
     |\tilde{C}^{(i)}_{\min}(N)-\tilde{C}^{(j)}_{\min}(N)| \leq \mathcal{O}(N^{1-\xi})
 \end{align*}
 Finally, letting $\tilde{C}^{(i)}_{\min}(N)=C^{(i)}_{m(i),k}$ and $\tilde{C}^{(j)}_{\min}(N)=\tilde{C}^{(j)}_{m(j),l}(N)$, we have for any $a_{k'}$ and $a_{l'}$ the following
 \begin{align*}
     |\tilde{C}^{(i)}_{m(i),k'}(N)-\tilde{C}^{(j)}_{m(j),l'}(N)|&\leq |\tilde{C}^{(i)}_{m(i),k'}(N)-\tilde{C}^{(i)}_{\min}(N)|+|\tilde{C}^{(j)}_{m(j),l'}(N)-\tilde{C}^{(j)}_{\min}(N)| \\
     &+|\tilde{C}^{(i)}_{\min}(N)-\tilde{C}^{(j)}_{\min}(N)|\\
     & \leq \mathcal{O}(N^{1-\xi})
 \end{align*}
 Using
\begin{align*}
    |\mathcal{I}^{(i)}_{m(i),k'}(N)-\mathcal{I}^{(j)}_{m(j),l'}(N)|\leq  |\tilde{C}^{(i)}_{m(i),k'}(N)-\tilde{C}^{(j)}_{m(j),l'}(N)| + \mathcal{O}(N^{1-\xi})
\end{align*}
we get
\begin{align*}
     |\mathcal{I}^{(i)}_{m(i),k'}(N)-\mathcal{I}^{(j)}_{m(j),l'}(N)| \leq  \mathcal{O}(N^{1-\xi})
\end{align*}

\end{proof}

\begin{theorem}[Convergence of proportions] For all $N\geq T_{\texttt{stable}}$, we have $ |N^{(i)}_k - N w^{\star (i)}_k | \leq \mathcal{O}(N^{1-\xi})$.
   \label{thm:convprop}
\end{theorem}
\begin{proof}
       For $w^{\star (i)}_{k,\texttt{L}}:=\frac{w^{\star (i)}_{k}}{w^{\star (i)}}$  where $w^{\star (i)}:=\sum_kw^{\star (i)}_{k}$ we have from \cite[Proposition 3.1.]{bandyopadhyay2024optimal} where $w^{\star(i)}_k$ is the optimal proportion,\[\max_{k}|\widetilde{w}^{(i)}_{k,\texttt{L}}(N^{(i)})-w^{\star (i)}_{k,\texttt{L}}|~=~O((N^{(i)})^{-\xi}).\] where $\tilde{w}^{(i)}_{k,\texttt{L}}:=\frac{\tilde{w}^{ (i)}_{k}}{\tilde{w}^{(i)}}$, $\tilde{w}^{(i)}_k=\frac{N^{(i)}_k(N)}{N^{(i)}(N)},\tilde{w}^{(i)}=\sum_k\tilde{w}^{(i)}_k$  and using $N^{(i)} = \Theta(N)$ we have $N^{(i)}_k = N^{(i)}\cdot w^{\star (i)}_{k,\texttt{L}}\pm O(N^{1-\xi})$.

      Writing $\bm{w}^{\star (i)}$ in terms of indexes: $\mathcal{I}^{\star(i)}:= w^{\star (i)}_{k,\texttt{L}}D(\mu^{(i)}_k,x^{\star(i)}_{m(i),k,\texttt{L}})+w^{\star (i)}_{m(i),\texttt{L}}D(\mu^{(i)}_{m(i)},x^{\star(i)}_{m(i),k,\texttt{L}})\ \forall k$ where $x^{\star(i)}_{m(i),k,\texttt{L}} = \frac{\mu^{(i)}_kw^{\star(i)}_{k,\texttt{L}}+\mu^{(i)}_{m(i)}w^{\star(i)}_{m(i),\texttt{L}}}{w^{\star(i)}_{k,\texttt{L}}+w^{\star(i)}_{m(i),\texttt{L}}}$, we have by the index equality the following
       \begin{align*}
           w^{\star (i)} \mathcal{I}^{\star(i)} =  w^{\star (j)} \mathcal{I}^{\star(j)}.
       \end{align*}
       Additionally, using $\sum_{i} w^{\star (i)}=1$, we have from above
       \begin{align}
           w^{\star (i)}\left(  +\sum_{j\neq i} \frac{ \mathcal{I}^{\star(i)}}{\mathcal{I}^{\star(j)}}\right)=1 \label{eq:wstar}
       \end{align}
       letting $\mathcal{I}_k^{(i)}(\tilde{\bm{w}}^{(i)})=\tilde{w}^{(i)}_{k,\texttt{L}}D(\mu^{(i)}_k,\tilde{x}^{(i)}_{m(i),k})+\tilde{w}^{(i)}_{m(i),\texttt{L}}D(\mu^{(i)}_{m(i)},\tilde{x}^{(i)}_{m(i),k})$ where $\tilde{x}^{(i)}_{m(i),k,\texttt{L}} = \frac{\mu^{(i)}_k\tilde{w}^{(i)}_{k,\texttt{L}}+\mu^{(i)}_{m(i)}\tilde{w}^{(i)}_{m(i),\texttt{L}}}{\tilde{w}^{(i)}_{k,\texttt{L}}+\tilde{w}^{(i)}_{m(i),\texttt{L}}}$ we have $\mathcal{I}^{\star(i)} = \mathcal{I}^{(i)}_k(\bm{w}^{\star(i)})\ \forall k$, thus
       \begin{align*}
           |\mathcal{I}_k^{(i)}(\tilde{\bm{w}}^{(i)})-\mathcal{I}_k^{(i)}(\bm{w}^{\star(i)})|\leq  \left(\dfrac{\partial}{\partial w^{(i)}_{m(i)}}\mathcal{I}_k^{(i)}(\Bar{\bm{w}}^{(i)})+\dfrac{\partial}{\partial w^{(i)}_{k}}\mathcal{I}_k^{(i)}(\Bar{\bm{w}}^{(i)})\right)\|\tilde{\bm{w}}^{(i)}-{\bm{w}}^{\star(i)}\|_{\infty}
       \end{align*}
       where $\Bar{\bm{w}}^{(i)}:\Bar{w}^{(i)}_{m(i)}\in[w^{(i)}_{m(i)},w^{\star (i)}_{m(i)}]$ and $\Bar{w}^{(i)}_{k}\in[w^{(i)}_{k},w^{\star (i)}_{k}]$. Recall
       \begin{align*}
           \dfrac{\partial}{\partial w^{(i)}_{m(i)}}\mathcal{I}_k^{(i)}(\Bar{\bm{w}}^{(i)}) = D(\mu^{(i)}_{m(i)},\Bar{x}^{(i)}_{m(i),k})\quad \text{and}\quad  \dfrac{\partial}{\partial w^{(i)}_{k}}\mathcal{I}_k^{(i)}(\Bar{\bm{w}}^{(i)}) = D(\mu^{(i)}_{k},\Bar{x}^{(i)}_{m(i),k})
       \end{align*}
       where $\Bar{x}^{(i)}_{m(i),k}=\frac{\bar{w}^{(i)}_{m(i)}\mu^{(i)}_{m(i)}+\bar{w}^{(i)}_{k}\mu^{(i)}_{k}}{\bar{w}^{(i)}_{m(i)}+\bar{w}^{(i)}_{k}}$. Observing that both partial derivatives are upper bounded by $\max\{D(\mu^{(i)}_{m(i)},\mu^{(i)}_{k}),D(\mu^{(i)}_{k},\mu^{(i)}_{m(i)})\}$, we have $ \mathcal{I}^{(i)}(\tilde{\bm{w}}^{(i)}) =\mathcal{I}^{(i)}(\bm{w}^{\star(i)}) \pm\mathcal{O}(N^{-\xi})$. Recall
       \begin{align*}
           \tilde{w}^{(i)}  \mathcal{I}^{(i)}(\tilde{\bm{w}}^{(i)})- \tilde{w}^{(j)}  \mathcal{I}^{(j)}(\tilde{\bm{w}}^{(j)})&=\pm \mathcal{O}(N^{-\xi})\\
           \implies  \tilde{w}^{(j)}\mathcal{I}^{(j)}({\bm{w}}^{\star(j)}) &= \tilde{w}^{(i)}\mathcal{I}^{(i)}({\bm{w}}^{\star(i)}) \pm\mathcal{O}(N^{-\xi})
       \end{align*}
Similarly, observing $\sum_j\tilde{w}^{(j)}=1$, we get
\begin{align}
   \tilde{w}^{(i)} +  \tilde{w}^{(i)} \sum_{j\neq i}\dfrac{\mathcal{I}^{(i)}({\bm{w}}^{\star(i)}) }{\mathcal{I}^{(j)}({\bm{w}}^{\star(j)}) } \pm\mathcal{O}(N^{-\xi}) = 1 \label{eq:wtilde}
\end{align}
Subtracting Eq~\ref{eq:wstar} and Eq~\ref{eq:wtilde}, we have
\begin{align*}
    (w^{\star (i)} - \tilde{w}^{(i)}) \left( \sum_{j\neq i}\dfrac{\mathcal{I}^{(i)}({\bm{w}}^{\star(i)}) }{\mathcal{I}^{(j)}({\bm{w}}^{\star(j)}) }+1\right)  = \pm\mathcal{O}(N^{-\xi})\\
    \implies  \tilde{w}^{(i)} = w^{\star (i)} \pm\mathcal{O}(N^{-\xi})
\end{align*}
Now, 
\begin{align*}
    \tilde{w}^{(i)}_{k} - w^{\star (i)}_{k}&= \tilde{w}^{ (i)}\cdot\tilde{w}^{(i)}_{k,\texttt{L}}-w^{\star (i)}\cdot w^{\star (i)}_{k,\texttt{L}}\\
    &= [w^{\star (i)} \pm\mathcal{O}(N^{-\xi})][w^{\star (i)}_{k,\texttt{L}}\pm\mathcal{O}(N^{-\xi})]-w^{\star (i)}\cdot w^{\star (i)}_{k,\texttt{L}}\\
    &= \pm\mathcal{O}(N^{-\xi})
\end{align*}
which concludes the proof.
\end{proof}

\begin{theorem}
 \texttt{ATT1} is asymptotically optimal for one-sided learning over market instances $\mu\in \mathcal{S}_{1}$, the corresponding stopping time satisfy $\limsup_{\delta\to 0}\frac{\E_{\mu}[\tau_{\delta}]}{\log(1/\delta)}\leq T^*(\mu)$ and $\limsup_{\delta\to 0}\frac{\tau_{\delta}}{\log(1/\delta)}\leq T^*(\mu)$ a.s. in $\mathbb{P}_\mu$. 
\end{theorem}
\begin{proof}
The proof exactly follows as 
\cite[Appendix H]{bandyopadhyay2024optimal}, we write it here for completeness.\\
Define $\mathcal{I}^{(i)}_k(\tilde{\bm{w}})=\tilde{w}^{(i)}_{k}D(\mu^{(i)}_k,\tilde{x}^{(i)}_{m(i),k})+\tilde{w}^{(i)}_{m(i)}D(\mu^{(i)}_{m(i)},\tilde{x}^{(i)}_{m(i),k})$ where recall $\tilde{w}^{(i)}_k=\frac{N^{(i)}_k(N)}{N^{(i)}(N)},\tilde{w}^{(i)}=\sum_k\tilde{w}^{(i)}_k$ and $\tilde{x}^{(i)}_{m(i),k} = \frac{\mu^{(i)}_k\tilde{w}^{(i)}_{k}+\mu^{(i)}_{m(i)}\tilde{w}^{(i)}_{m(i),}}{\tilde{w}^{(i)}_{k}+\tilde{w}^{(i)}_{m(i)}}$. Further notice that $\mathcal{I}^{(i)}_k({\bm{w}^{\star}})=(T^{\star}(\mu))^{-1}\ \forall i,k$. Following similar argument in the proof of Theorem~\ref{thm:convprop} using mean value theorem we have
\begin{align*}
    |\mathcal{I}^{(i)}_k(\tilde{\bm{w}}(N))-(T^{\star}(\mu))^{-1}|\leq \mathcal{O}(N^{-\xi})
\end{align*}
this implies that there exists a constant $C$ such that
\begin{align}
    \tilde{C}^{(i)}_{m(i),k}\geq \frac{N}{T^{\star}(\mu)}-CN^{1-\xi}\ \forall\ N\geq T_{\texttt{stable}},\ \forall i,k\in\mathcal{E}^{(i)}_m.\label{eq:stopplB}
\end{align}
Recall the stopping rule as $ \tilde{C}^{(i)}_{m(i),k}=\beta(N,\delta)=\mathcal{O}(N,\delta)$.

    Define for every $\delta>0$, the following deterministic quantity
    \begin{align*}
        t_{\max,\delta}:= \min\left\{N\geq T_{\texttt{stable}}\Big| \dfrac{N}{T^{\star}(u)}-CN^{1-\xi}>\beta(N,\delta)\right\}. 
    \end{align*}
    This implies that for stopping time $\tau_{\delta}$ when $\tau_{\delta}\geq T_{\texttt{stable}}$, we have $\tau_{\delta}\leq t_{\max,\delta}$. This is so because $ \tilde{C}^{(i)}_{m(i),k}$ is increasing in $N$ and having the R.H.S. of Eq.~\ref{eq:stopplB} exceed $\beta(N,\delta)$ implies $ \tilde{C}^{(i)}_{m(i),k}$ exceeding the threshold $\beta(N,\delta)$. Thus we have $\tau_{\delta}\leq \max\{ t_{\max,\delta},T_{\stable}\}$\\ 
    \begin{align*}
        \lim_{\delta\to 0}\dfrac{\E[\tau_{\delta}]}{\log(1/\delta)}\leq  \lim_{\delta\to 0}\dfrac{\E[T_{\stable}]+ t_{\max,\delta}}{\log(1/\delta)}\overset{(1)}{=} \lim_{\delta\to 0}\dfrac{ t_{\max,\delta}}{\log(1/\delta)}
    \end{align*}
    where (1) is follows as $\E[T_{\stable}]< \infty$.\\
    Defining $s_{\max,\delta}=t_{\max,\delta}-1$, and note that $\lim_{\delta\to 0}\frac{t_{\max,\delta}}{\log(1/\delta)}=\lim_{\delta\to 0}\frac{s_{\max,\delta}}{\log(1/\delta)}$. Buy definition we have
    \begin{align*}
        \dfrac{s_{\max,\delta}}{T^{\star}(u)}-Cs_{\max,\delta}^{1-\xi}<\beta(s_{\max,\delta},\delta)\ = \mathcal{O}(\log(1/\delta)) 
    \end{align*}

    As $\delta\to 0$, we have $\tau_{\max,\delta}\to\infty$, which implies $s_{\max,\delta}\to \infty$. Rearranging terms above we get
    \begin{align*}
        \frac{1}{T^{\star}(\mu)}\lim\sup_{\delta\to 0 }\dfrac{s_{\max,\delta}}{\log(1/\delta)}\leq 1
    \end{align*}
    this concludes the proof. The statement $\limsup_{\delta\to 0}\frac{\tau_{\delta}}{\log(1/\delta)}\leq T^*(\mu)$ a.s. in $\mathbb{P}_\mu$ can similarly be proved as in \cite[Appendix H.]{bandyopadhyay2024optimal}.
\end{proof}


%% file: paper/appendix/fluid-two-sided.tex
\section{Fluid Two-sided}
\label{app:fluid-two-sided}
\begin{lemma}
    For each player $p_i, g^{(i)}_{m(i)}$ is monotonically increasing in $N^{(i)}_k\ \forall k\in \mathcal{B}_1^{(i)}\cup \mathcal{B}_3^{(i)}$, $N^{(j)}_{k}\ \forall j:k\in\mathcal{B}_2^{(j)}\cup\mathcal{B}_3^{(j)}$ and monotonically decreasing in $N^{(k)}_{m(k)}\forall k\in \mathcal{B}_3^{(i)}$ and $N^{(j)}_{m(j)}\ \forall j:k\in\mathcal{B}_3^{(j)}$
\end{lemma}

\subsection{Serial Dictatorship on both sides}
 We now analyze a case when the preference induced follows a serial dictatorship condition on both sides i.e. each player $p_i$ has same preference ordering $a_1\succ \cdots \succ a_{|\mathcal{A}|}$ and each arm $a_k$ has same preference ordering $p_1\succ \cdots \succ p_{|\mathcal{P}|}. $. The stable matching is thus $m(p_i)=a_i$. Due to the preference profile we have that $\mathcal{B}_3=\emptyset$ as $\not\exists (p_i,a_k)$ s.t. $a_i\succ_{p_i}a_k$ and $p_i\succ_{a_i}p_k$.\\
For all players $p_i\in\mathcal{P}_{\min}$, we partition the set of arms $\mathcal{A}^{(i)}_{\min}$ into two disjoint subsets, $\mathcal{B}^{(i)}_{1,\min}$ and $\mathcal{B}^{(i)}_{2,\min}$, based on whether the pair $(p_i,a_k)$ belongs to $\mathcal{B}_1$ or $\mathcal{B}_2$ for each $a_k\in\mathcal{A}^{(i)}_{\min}$.

To facilitate the analysis, we introduce the following notation.

\paragraph{Difference and Derivative Terms.} We define the mean differences and their associated derivative terms as follows:
\begin{align*}
    \Delta^{(i)}_{k} &:= \mu^{(i)}_{m(i)}-\mu^{(i)}_{k} & \bar{\Delta}^{(j)}_{m(i)} &:= \eta^{(i)}_{m(i)}-\eta^{(j)}_{m(i)} \\
    f^{(i)}_k &:= -\frac{\partial}{\partial x}\frac{d(\mu^{(i)}_{m(i)} ,x)}{d(\mu^{(i)}_{k},x)} & \bar{f}^{(j)}_{m(i)} &:= -\frac{\partial}{\partial y}\frac{d(\eta^{(i)}_{m(i)},y)}{d(\eta^{(j)}_{m(i)},y)}
\end{align*}

\paragraph{Composite Fluid Terms.} These are combined to form the composite terms $h$ and $\bar{h}$:
\begin{align*}
    h^{(i)}_k &:= f^{(i)}_k \frac{(t^{(i)}_{m(i)})^2\Delta^{(i)}_{k}}{(t^{(i)}_{m(i)}+t^{(i)}_k)^2} & \bar{h}^{(j)}_{m(i)} &:= \bar{f}^{(j)}_{m(i)} \frac{(t^{(i)}_{m(i)})^2\bar{\Delta}^{(j)}_{m(i)}}{(t^{(i)}_{m(i)}+t^{(j)}_{m(i)})^2}
\end{align*}

\paragraph{KL-Divergence Terms.} For pairs in $\mathcal{B}^{(i)}_{1,\min}$, we denote the KL-divergence components as:
\begin{align*}
    D^{(i)}_{m(i),k} &:= D(\mu^{(i)}_{m(i)},x^{(i)}_{m(i),k}) & D^{(i)}_{k,k} &:= D(\mu^{(i)}_{k},x^{(i)}_{m(i),k})
\end{align*}
And for arms $m(i)$ in $\mathcal{B}^{(j)}_{2,\min}$, the components are:
\begin{align*}
    D^{(i,j)}_{m(i)} &:= D(\eta^{(i)}_{m(i)},y^{(i,j)}_{m(i)}) & D^{(j,j)}_{m(i)} &:= D(\eta^{(j)}_{m(i)},y^{(i,j)}_{m(i)})
\end{align*}

\paragraph{Aggregated Terms.} Finally, we define several aggregated quantities used in our analysis:

\begin{align*}
     h^{(i)}_{\min} := \sum_{k\in\mathcal{B}^{(i)}_{1,\min}}\frac{h^{(i)}_{k}}{D^{(i)}_{k,k}},\ \   h_{m(i),\min} := \sum_{j:m(i)\in\mathcal{B}^{(j)}_{2,\min}}\frac{\bar{h}^{(j)}_{m(i)}}{D^{(i,j)}_{m(i)}} \ \text{and}\    h^{(i)}_{m(i),\min} := h^{(i)}_{\min}+h_{m(i),\min}
\end{align*}

\begin{align*}
h^{(i)}_{m(i),\texttt{w}} &:= \sum_{k\not\in\mathcal{B}^{(i)}_{1,\min}\cup\{m(i)\}}h^{(i)}_{k}t^{(i)}_k + \sum_{j\not\in\{j:m(i)\in\mathcal{B}^{(j)}_{2,\min}\}\cup\{i\}}h^{(j)}_{m(i)}t^{(j)}_{m(i)} \\
    D^{(i)}_{m(i),\min} &:= \left(\sum_{k\in\mathcal{B}^{(i)}_{1,\min}}\left(D^{(i)}_{k,k}\right)^{-1}+\sum_{j:m(i)\in\mathcal{B}^{(j)}_{2,\min}}\left(D^{(i,j)}_{m(i)}\right)^{-1}\right)^{-1} \\ \texttt{T}^{(i)}_{\min} &:= \sum_{k\in \mathcal{B}_{1,\min}^{(i)}}t^{(i)}_k+\sum_{p_j:m(i)\in \mathcal{B}_{2,\min}^{(j)}}t^{(j)}_{m(i)}
\end{align*}
\begin{theorem}
Consider an instant $t_0$ such that $\forall p_i\in\mathcal{P}_{\min}(t_0)$,$g^{(i)}_{m(i)}=0$.
 We consider the ODEs till $t_1$ defined as the smallest time after $t_0$ such that either Minimum index at $t_1$ becomes equal to the index of player, arm pair for which the index at $t_0$ was strictly larger
     or anchor function of some player becomes equal to zero at $t_1$. This gives  the following for all $t\in [t_0,t_1)$.
        \begin{align*}
       (\dot{C}_{\min})^{-1
        }&=  \sum_{i\in\mathcal{P}_{0}} \dfrac{ t^{(i)}_{m(i)}h^{(i)}_{m(i),\min}+h^{(i)}_{m(i),\texttt{w}}/D^{(i)}_{m(i),\min}+h^{(i)}_{m(i),\min}\texttt{T}^{(i)}_{\min}}{ C_{\min}h^{(i)}_{m(i),\min}+h^{(i)}_{m(i),\texttt{w}}} 
\end{align*}
and 
 \begin{align*}
     \dot{t}^{(i)}_{m(i)} &= \dot{C}_{\min}\dfrac{ t^{(i)}_{m(i)}(h^{(i)}_{\min}+h_{m(i),\min})}{ C_{\min}(h^{(i)}_{\min}+h_{m(i),\min})+(h^{(i)}_{\texttt{w}}+h_{m(i),\texttt{w}})}\quad \forall p_i\in\mathcal{P}_{\min}\\
       \dot{t}^{(i)}_{l} &= \dot{C}_{\min}\dfrac{(h^{(i)}_{\texttt{w}}+h_{m(i),\texttt{w}})/D^{(i)}_{l,l}+t^{(i)}_l(h^{(i)}_{\min}+h_{m(i),\min})}{ C_{\min}(h^{(i)}_{\min}+h_{m(i),\min})+(h^{(i)}_{\texttt{w}}+h_{m(i),\texttt{w}})}\quad \forall p_i\in\mathcal{P}_{\min}, a_l\in\mathcal{B}_{1,\min}^{(i)}\\
       \dot{t}^{(j)}_{m(i)} &=   \dot{C}_{\min}\dfrac{(h^{(i)}_{\texttt{w}}+h_{m(i),\texttt{w}})/D^{(j,j)}_{m(i)}+t^{(j)}_{m(i)}(h^{(i)}_{\min}+h_{m(i),\min})}{ C_{\min}(h^{(i)}_{\min}+h_{m(i),\min})+(h^{(i)}_{\texttt{w}}+h_{m(i),\texttt{w}})}\quad \forall p_j\in\mathcal{P}_{\min}, p_j:m(i)\in\mathcal{B}_{2,\min}^{(j)}
   \end{align*}
\end{theorem}

\begin{proof}

 For simplicity we consider $P_+\cup P_-=\emptyset$
    \begin{align}
        \sum_{i\in\mathcal{P}_{0}\cup\mathcal{P}_{+}} \dot{t}^{(i)}_{m(i)}  + \sum_{i\in\mathcal{P}_{0}} \sum_{k\in\mathcal{A}^{(i)}_{\min}}\dot{t}^{(i)}_k +\sum_{i\in \mathcal{P}_{-}}\sum_{k\in\mathcal{A}^{(i)}_{\min}}\dot{t}^{(i)}_k   = 1 \label{eq:sumalltitwo}
    \end{align}

 For $p_i\in \mathcal{P}_{0}$, we use the fact that the minimum indexes for $p_i$ also stay together.

          Using $\dot{g}^{(i)}_{m(i)}=0$, we get
    \begin{align*}
        &\sum_{a_k\in\mathcal{B}_1^{(i)}}f^{(i)}_k\dfrac{d}{dt^{(i)}_{m(i)}}x^{(i)}_{m(i),k}\dot{t}^{(i)}_{m(i)}+ \sum_{a_k\in\mathcal{B}_{1,\min}^{(i)}}f^{(i)}_k\dfrac{d}{dt^{(i)}_{k}}x^{(i)}_{m(i),k}\dot{t}^{(i)}_{k}  \\
        &+ \sum_{p_j:m(i)\in\mathcal{B}_2^{(j)}}\bar{f}^{(j)}_{m(i)}\dfrac{d}{dt^{(i)}_{m(i)}}y^{(j,i)}_{m(i)}\dot{t}^{(i)}_{m(i)} +\sum_{p_j:m(i)\in\mathcal{B}_{2,\min}^{(j)}}\bar{f}^{(j)}_{m(i)}\dfrac{d}{dt^{(j)}_{m(i)}}y^{(j,i)}_{m(i)}\dot{t}^{(j)}_{m(i)} =0
    \end{align*}
    Observe that $\frac{d}{dt^{(i)}_{k}}x^{(i)}_{m(i),k}=-\frac{t^{(i)}_{m(i)}}{t^{(i)}_k}\frac{d}{dt^{(i)}_{m(i)}}x^{(i)}_{m(i),k}$ and $\frac{d}{dt^{(j)}_{m(i)}}y^{(j,i)}_{m(i)} = -\frac{t^{(i)}_{m(i)}}{t^{(j)}_{m(i)}}\frac{d}{dt^{(i)}_{m(i)}}y^{(j,i)}_{m(i)}$, which gives
     \begin{align*}
        &\sum_{a_k\in\mathcal{B}_1^{(i)}}f^{(i)}_k\dfrac{d}{dt^{(i)}_{m(i)}}x^{(i)}_{m(i),k}\dot{t}^{(i)}_{m(i)}- \sum_{a_k\in\mathcal{B}_{1,\min}^{(i)}}f^{(i)}_k\dfrac{t^{(i)}_{m(i)}}{t^{(i)}_k}\dfrac{d}{dt^{(i)}_{m(i)}}x^{(i)}_{m(i),k}\dot{t}^{(i)}_{k} \\
        &+ \sum_{p_j\in:m(i)\in\mathcal{B}_{2,\min}^{(j)}}\bar{f}^{(j)}_{m(i)}\dfrac{d}{dt^{(i)}_{m(i)}}y^{(j,i)}_{m(i)}\dot{t}^{(i)}_{m(i)} -\sum_{p_j:m(i)\in\mathcal{B}_{2,\min}^{(j)}}\bar{f}^{(j)}_{m(i)}\dfrac{t^{(i)}_{m(i)}}{t^{(j)}_{m(i)}}\dfrac{d}{dt^{(i)}_{m(i)}}y^{(j,i)}_{m(i)}\dot{t}^{(j)}_{m(i)} &=0
    \end{align*}
   First we represent all $\dot{t}^{(i)}_k$ in terms of $\dot{t}^{(i)}_l$ for $k,l\in\mathcal{B}_{1,\min}^{(i)}$. Thus for $a_k\in\mathcal{B}^{(i)}_{1,\min}$, we have using $\dot{C}^{(i)}_{m(i),k}=\dot{C}^{(i)}_{m(i),l}$, the following
    \begin{align*}
    \dot{t}^{(i)}_{m(i)}D^{(i)}_{m(i),k} + \dot{t}^{(i)}_kD^{(i)}_{k,k}= \dot{t}^{(i)}_{m(i)}D^{(i)}_{m(i),l} + \dot{t}^{(i)}_lD^{(i)}_{l,l}.
    \end{align*}
    Further using ${C}^{(i)}_{m(i),k}={C}^{(i)}_{m(i),l}$, we get
    \begin{align*}
        \dot{t}^{(i)}_k = \dot{t}^{(i)}_{m(i)}\dfrac{1}{D^{(i)}_{k,k}t^{(i)}_{m(i)}}(t^{(i)}_kD^{(i)}_{k,k}-t^{(i)}_lD^{(i)}_{l,l}) + \dot{t}^{(i)}_l\dfrac{D^{(i)}_{l,l}}{D^{(i)}_{k,k}}
    \end{align*}

 For $p_j\in\mathcal{P}_{0}:m(i)\in\mathcal{B}_{2,\min}^{(j)}$, we similarly represent all $\dot{t}^{(j)}_{m(i)}$ in terms of $\dot{t}^{(i)}_l$. Using $\dot{C}^{(i)}_{m(i),l}=\dot{C}^{(i,j)}_{m(i)}$ and ${C}^{(i)}_{m(i),l}=C^{(i,j)}_{m(i)}$ we get
    \begin{align*}
        \dot{t}^{(j)}_{m(i)} =  \dot{t}^{(i)}_{m(i)}\dfrac{1}{D^{(j,j)}_{m(i)}t^{(i)}_{m(i)}}(t^{(j)}_{m(i)}D^{(j,j)}_{m(i)}-t^{(i)}_lD^{(i)}_{l,l})+ \dot{t}^{(i)}_l\dfrac{D^{(i)}_{l,l}}{D^{(j,j)}_{m(i)}}
    \end{align*}
  Substituting this in $\dot{g}^{(i)}_{m(i)}=0$, we have
  \begin{align*}
      \dot{t}^{(i)}_{m(i)} = \dot{t}^{(i)}_l \dfrac{ t^{(i)}_{m(i)}(h^{(i)}_{\min}+h_{m(i),\min})}{(h^{(i)}_{\texttt{w}}+h_{m(i),\texttt{w}})/D^{(i)}_{l,l}+t^{(i)}_l(h^{(i)}_{\min}+h_{m(i),\min})}
  \end{align*}
  Using the definition of $C_{\min}$ we get
  \begin{align*}
      \dot{C}_{\min} &= \dot{t}^{(i)}_{m(i)}D^{(i)}_{m(i),l}+\dot{t}^{(i)}_{l} D^{(i)}_{l,l}= \dot{t}^{(i)}_{l}[ \dfrac{ C_{\min}(h^{(i)}_{\min}+h_{m(i),\min})+(h^{(i)}_{\texttt{w}}+h_{m(i),\texttt{w}})}{(h^{(i)}_{\texttt{w}}+h_{m(i),\texttt{w}})/D^{(i)}_{l,l}+t^{(i)}_l(h^{(i)}_{\min}+h_{m(i),\min})}]
   \end{align*}
   which gives us the following
   \begin{align*}
       \dot{t}^{(i)}_{l} &= \dot{C}_{\min}\dfrac{(h^{(i)}_{\texttt{w}}+h_{m(i),\texttt{w}})/D^{(i)}_{l,l}+t^{(i)}_l(h^{(i)}_{\min}+h_{m(i),\min})}{ C_{\min}(h^{(i)}_{\min}+h_{m(i),\min})+(h^{(i)}_{\texttt{w}}+h_{m(i),\texttt{w}})}\\
       \dot{t}^{(i)}_{m(i)} &= \dot{C}_{\min}\dfrac{ t^{(i)}_{m(i)}(h^{(i)}_{\min}+h_{m(i),\min})}{ C_{\min}(h^{(i)}_{\min}+h_{m(i),\min})+(h^{(i)}_{\texttt{w}}+h_{m(i),\texttt{w}})}\\
       \dot{t}^{(j)}_{m(i)} &=   \dot{C}_{\min}\dfrac{(h^{(i)}_{\texttt{w}}+h_{m(i),\texttt{w}})/D^{(j,j)}_{m(i)}+t^{(j)}_{m(i)}(h^{(i)}_{\min}+h_{m(i),\min})}{ C_{\min}(h^{(i)}_{\min}+h_{m(i),\min})+(h^{(i)}_{\texttt{w}}+h_{m(i),\texttt{w}})}
   \end{align*}
   Substituting this in Eq. \ref{eq:sumalltitwo}, gives us

    \begin{align*}
       (\dot{C}_{\min})^{-1
        }&=  \sum_{i\in\mathcal{P}_{0}} \dfrac{ t^{(i)}_{m(i)}(h^{(i)}_{\min}+h_{m(i),\min})}{ C_{\min}(h^{(i)}_{\min}+h_{m(i),\min})+(h^{(i)}_{\texttt{w}}+h_{m(i),\texttt{w}})} \\
        &+\sum_{i\in\mathcal{P}_{0}} \sum_{k\in \mathcal{B}_{1,\min}^{(i)}}\dfrac{(h^{(i)}_{\texttt{w}}+h_{m(i),\texttt{w}})/D^{(i)}_{k,k}+t^{(i)}_k(h^{(i)}_{\min}+h_{m(i),\min})}{ C_{\min}(h^{(i)}_{\min}+h_{m(i),\min})+(h^{(i)}_{\texttt{w}}+h_{m(i),\texttt{w}})}\\
        &+\sum_{m(i):p_i\in \mathcal{P}_{0}}\sum_{p_j:m(i)\in \mathcal{B}_2^{(j)}} \dot{C}_{\min}\dfrac{(h^{(i)}_{\texttt{w}}+h_{m(i),\texttt{w}})/D^{(j,j)}_{m(i)}+t^{(j)}_{m(i)}(h^{(i)}_{\min}+h_{m(i),\min})}{ C_{\min}(h^{(i)}_{\min}+h_{m(i),\min})+(h^{(i)}_{\texttt{w}}+h_{m(i),\texttt{w}})}
\end{align*}
\end{proof}

\subsection{2x2 Example with Distinct Preference Profile}

Consider two players and two arms, with the preference profile:
\begin{align*}
    p_1 & : a_1 \succ a_2, & a_1 & : p_1 \succ p_2 \\
    p_2 & : a_2 \succ a_1, & a_2 & : p_2 \succ p_1
\end{align*}

Let $t_{ik} := t^{(i)}_k$ denote the continuous sampling quantity for pair $(p_i, a_k)$. Define:
\[
    \dot{g}_\ell = a_\ell \dot{t}_{12} + b_\ell \dot{t}_{21} - c_\ell \dot{t}_{11} - d_\ell \dot{t}_{22}, \quad \text{for } \ell = 1, 2,
\]
where
\[
    a_\ell = \frac{\partial g_\ell}{\partial t_{12}}, \quad
    b_\ell = \frac{\partial g_\ell}{\partial t_{21}}, \quad
    c_\ell = -\frac{\partial g_\ell}{\partial t_{11}}, \quad
    d_\ell = -\frac{\partial g_\ell}{\partial t_{22}}.
\]

Let $h := \max\{g_1, g_2\}$, and define the index values:
\[
    I_{12} := C^{(1,2)}_{1,2}, \quad I_{21} := C^{(2,1)}_{2,1}, \quad \text{with } m(i) = i.
\]
Define divergence terms:
\begin{table}[H]
    \centering
    \begin{tabular}{cccc}
       $d_{11} := D^{(1)}_{1,2}$  & $d_{12} := D^{(1)}_{2,2}$ & $d_{21} := D^{(2)}_{1,1}$ & $d_{22} := D^{(2)}_{2,1}$  \\
       \ \ \ $b_{11} := D^{(1,2)}_{1}$ &\ \ \  $b_{12} := D^{(1,1)}_{2}$  &\ \ \  $b_{22} := D^{(2,1)}_{2} $&\ \ \   $b_{21} := D^{(2,2)}_{1}$
    \end{tabular}
\end{table}

We now describe fluid ODEs for different initial configurations:

\begin{itemize}[leftmargin=*, nosep]
  \item \textbf{Case 1: $h \neq 0$ and $I_{12} < I_{21}$}
  \begin{itemize}[leftmargin=*]
    \item If $h < 0$:
      \[
        \dot{t}_{12} = 1, \quad \dot{g}_\ell = a_\ell, \quad \dot{h} = \max\{a_1, a_2\}.
      \]
    \item If $h > 0$:
      \begin{itemize}
        \item If $g_1 > g_2$: $\dot{h} = -c_1$
        \item If $g_2 > g_1$: $\dot{h} = -d_2$
        \item If $g_1 = g_2$:
          \begin{enumerate}[leftmargin=*]
            \item \textit{dynamics such that $g_1=g_2$ is maintained:} 
            \[
              \begin{aligned}
              \dot{t}_{11} &= \frac{d_1 - d_2}{\kappa}, \quad
              \dot{t}_{22} = \frac{c_2 - c_1}{\kappa}, \quad
              \dot{h} = \dot{g}_1 = \dot{g}_2 = \frac{c_1 d_2 - d_1 c_2}{\kappa}, \\
              \text{where } \kappa &= c_2 - c_1 + d_1 - d_2
              \end{aligned}
            \]
             \[
            \text{Valid iff } \frac{d_1 - d_2}{c_2 - c_1} > 0 \iff \frac{\partial (g_1 - g_2)}{\partial t_{11}} \cdot \frac{\partial (g_1 - g_2)}{\partial t_{22}} < 0.
          \]
            \item \textit{$g_1$ starts dominating:} $\dot{t}_{11} = 1$, $\dot{g}_1 = -c_1$, $\dot{g}_2 = -c_2$, $\dot{h} = \dot{g}_1 = -c_1$
            \item \textit{$g_2$ starts dominating:} $\dot{t}_{22} = 1$, $\dot{g}_1 = -d_1$, $\dot{g}_2 = -d_2$, $\dot{h} = \dot{g}_2 = -d_2$
          \end{enumerate}
      \end{itemize}
  \end{itemize}

  \item \textbf{Case 2: $h \neq 0$ and $I_{12} = I_{21}$}
  \begin{itemize}[leftmargin=*]
    \item If $h < 0$:
      \[
        \begin{aligned}
        \dot{t}_{12} &= \frac{b_{21}+d_{21}}{\Sigma}, \quad \dot{t}_{21} = \frac{b_{12}+d_{12}}{\Sigma}, \\
        \dot{I}_{12} &= \dot{I}_{21} = \frac{2(b_{21}+d_{21})(b_{12}+d_{12})}{\Sigma}, \\
        \dot{g}_\ell &= a_\ell \dot{t}_{12} + b_\ell \dot{t}_{21} > 0 \text{for } \ell = 1, 2
        \end{aligned}
      \]
      where $\Sigma = b_{21} + d_{21} + b_{12} + d_{12}$
    \item If $h > 0$: indexes may separate as the algorithms does not look at the indexes
  \end{itemize}

  \item \textbf{Case 3: $h = 0$ and  $I_{12} < I_{21}$}: fluid dynamics is derived such that $h=0$ is maintained
  \begin{itemize}[leftmargin=*]
    \item $g_1 = 0, g_2 < 0$:
      \[
        \begin{aligned}
        \dot{t}_{12} &= \frac{c_1}{a_1 + c_1}, \quad \dot{t}_{11} = \frac{a_1}{a_1 + c_1}, \\
        \dot{g}_2 &= \frac{c_1 a_2 - a_1 c_2}{a_1 + c_1}, \\
        \dot{I}_{12} &= d_{11} \dot{t}_{11} + \dot{t}_{12}(d_{12} + b_{12}), \quad \dot{I}_{21} = b_{11} \dot{t}_{11}
        \end{aligned}
      \]
    \item $g_2 = 0, g_1 < 0$: analogous
    \item $g_1 = g_2 = 0$:
      \begin{enumerate}[label=(\roman*),leftmargin=*]
        \item \textit{dynamics such that $g_1=g_2$ is maintained:}  $\dot{g}_1=0=\dot{g}_2$
          \[
            \begin{aligned}
            \dot{t}_{11} &= \frac{a_1 d_2 - d_1 a_2}{\texttt{Den}}, \quad
            \dot{t}_{22} = \frac{a_2 c_1 - a_1 c_2}{\texttt{Den}}, \quad
            \dot{t}_{12} = \frac{c_1 d_2 - d_1 c_2}{\texttt{Den}}, \\
            \texttt{Den} &= a_1 d_2 - d_1 a_2 + a_2 c_1 - a_1 c_2 + c_1 d_2 - d_1 c_2
            \end{aligned}
          \]
          Valid when ordering ratios satisfy: $\frac{c_1}{c_2} > \frac{a_1}{a_2} > \frac{d_1}{d_2} \text{ or } \frac{c_1}{c_2} < \frac{a_1}{a_2} < \frac{d_1}{d_2}$ i.e. invalid when $\frac{a_1}{a_2}>\max\left\{\frac{c_1}{c_2},\frac{d_1}{d_2}\right\},\quad \frac{a_1}{a_2}<\min\left\{\frac{c_1}{c_2},\frac{d_1}{d_2}\right\}$. \\
          Indexes evolves as
\[\dot{I}_{12}=d_{11}\dot{t}_{11}+\dot{t}_{12}(d_{12}+b_{12})+b_{22}\dot{t}_{22},\quad \dot{I}_{21}=b_{11}\dot{t}_{11}+d_{22}\dot{t}_{22}\]
 $\frac{a_1}{a_2}>\max\{\frac{c_1}{c_2},\frac{d_1}{d_2}\}$ implies that $a_2c_1-a_1c_2<0$ and $a_1d_2-a_2d_1>0$.\\
  $\frac{a_1}{a_2}<\min\{\frac{c_1}{c_2},\frac{d_1}{d_2}\}$ implies that $a_2c_1-a_1c_2>0$ and $a_1d_2-a_2d_1<0$.\\
  Combining $(a_2c_1-a_1c_2)(a_1d_2-a_2d_1)<0$
    \item  \textit{$g_1$ starts dominating:} $\dot{h}=\dot{g}_1=0$, $\dot{t}_{12}=\frac{c_1}{a_1+c_1},\dot{t}_{11}=\frac{a_1}{a_1+c_1}$, $\dot{g}_2= \frac{c_1a_2-a_1c_2}{a_1+c_1}$.\\
    \item  \textit{$g_2$ starts dominating:} $\dot{h}=\dot{g}_2=0$, $\dot{t}_{12}=\frac{d_2}{a_1+d_2},\dot{t}_{11}=\frac{a_1}{a_1+d_2}$, $\dot{g}_1= \frac{d_2a_1-a_2d_1}{a_1+d_2}$\\
      \end{enumerate}
  \end{itemize}

  \item \textbf{Case 4: $h = 0$ and $I_{12} = I_{21}$}: $\dot{t}_{ik} = w^{\star}_{ik}, \text{ i.e., follow optimal ratio}$
\end{itemize}

\paragraph{Poisson Bracket Intuition} for Case (3) $h=0$ and $I_{12}<I_{21}$ with $g_1=g_2=0$. Define:
\[
J =
\begin{pmatrix}
\partial_{t_{12}} g_1 & \partial_{t_{11}} g_1 & \partial_{t_{22}} g_1 \\
\partial_{t_{12}} g_2 & \partial_{t_{11}} g_2 & \partial_{t_{22}} g_2
\end{pmatrix} =
\begin{pmatrix}
a_1 & -c_1 & -d_1 \\
a_2 & -c_2 & -d_2
\end{pmatrix}
\]
Define the following Poisson brackets corresponding to the pair $t_{12},t_{11}$ and $t_{12},t_{22}$ as
\begin{align*}
    \{g_1,g_2\}_1&=\dfrac{\partial g_1}{\partial t_{12}}\dfrac{\partial g_2}{\partial t_{11}}-\dfrac{\partial g_1}{\partial t_{11}}\dfrac{\partial g_2}{\partial t_{12}} = -a_1c_2+c_1a_2\ \text{and}\\
        \{g_1,g_2\}_2&=\dfrac{\partial g_1}{\partial t_{12}}\dfrac{\partial g_2}{\partial t_{22}}-\dfrac{\partial g_1}{\partial t_{22}}\dfrac{\partial g_2}{\partial t_{12}} = -a_1d_2+d_1a_2
\end{align*}
respectively. 
\textit{Intuition}: $\{f,g\}>0$ it means that following $f$'s flow makes $g$ increase whereas $\{f,g\}<0$   means that following $f$'s flow makes $g$ decrease. \\

Thus $\frac{a_1}{a_1}>\max\{\frac{c_1}{c_2},\frac{d_1}{d_2}\}$ means that $  \{g_1,g_2\}_1, \{g_1,g_2\}_2<0$, implying that if $g_1$'s flow is followed, $g_2$ decreases irrespective of increasing $t_{11},t_{22}$. Further, since $\{f,g\}=-\{g,f\}$. This implies that if $g_2$'s flow is followed, $g_1$ increases irrespective of increasing $t_{11},t_{22}$.\\
Hence, $g_1$ starts to dominate and we enter sub-case (ii)\\

Thus $\frac{a_1}{a_1}<\min\{\frac{c_1}{c_2},\frac{d_1}{d_2}\}$ means that $  \{g_1,g_2\}_1, \{g_1,g_2\}_2>0$, implying that if $g_1$'s flow is followed, $g_2$ increases irrespective of increasing $t_{11},t_{22}$. Further, since $\{f,g\}=-\{g,f\}$. This implies that if $g_2$'s flow is followed, $g_1$ decreases irrespective of increasing $t_{11},t_{22}$.\\
Hence, $g_2$ starts to dominate and we enter sub-case (iii)

%% file: paper/appendix/beta-algos.tex
\section{$\beta$-Top-Two algorithms}
\label{app:beta-algos}
\subsection{One-Sided Learning}
$\bar{\beta}$ Top Two algorithm for one-sided learning: 

     

\begin{algorithm}[H]
\caption{$\bar{\beta}$-Top-Two with $M\leq K$}
Input: preferences of arms over players $\pi$\\
    \For{$t=1,2,\ldots$}{
    $\mathcal{E}_p=\{p_i:\sum_kN^{(i)}_k\leq \sqrt{N}\}$\\
    $\hat{m}\leftarrow \texttt{DA}_{\texttt{Arm}}(\hat{\mu},\succ)$\\
    Construct $\mathcal{D}_{\hat{m}}^{(i)}=\{k\neq \hat{m}(i): p_i\succ_{a_k}\hat{m}^{-1}(k)\} \ \forall i$ and $\UMA_{\hat{m}}=\{k: \hat{m}^{-1}(k)=\emptyset\}$\\
      \lIf{$\mathcal{E}\neq \emptyset$ }{
    Match any pair $(p_i,a_k)$ from $\mathcal{E}$
    }
    \Else{
     
    Let $\Tilde{C}^{(i_t)}_{\hat{m}(i_t),k_t}=\min_{i}\min_{k\in \mathcal{D}^{(i)}_{\hat{m}}}\Tilde{C}^{(i)}_{\hat{m}(i),k}(\hat{\mu},N)$ \\
    \uIf{$\beta^{(i)}<\mathcal{U}[0,1]$}{
    Match player $p_{i_{t}}$ with arm $a_{{\hat{m}}(i_t)}$
    }
    \Else{
     Match player $p_{i_{t}}$ with arm $a_{k_t}$
    }}
    \lIf{$\texttt{DA}_{\texttt{Arm}}(\hat{\mu},\succ)=\texttt{DA}_{\texttt{Player}}(\hat{\mu},\succ)$ and $\Tilde{C}>\beta(t,\delta)$ }{Recommend $\texttt{DA}_{\texttt{Arm}}(\hat{\mu},\succ)$}}
\end{algorithm}

\subsubsection{Competitive Ratio}
Every $\bar{\beta}$-top-two one-sided algorithm can be indexed by a tuple $\bar{\beta}~=~(\beta_1,\beta_2,\hdots,\beta_{|\mathcal{P}|})\in[0,1]^{|\mathcal{P}|}$. The sample complexity lower bound of $\bar{\beta}$-top-two policy is $D_{\bar{\beta}}(\mu)^{-1}\cdot \log(1/\delta)$ where
\begin{align*}
    D_{\bar{\beta}}(\mu)~&=~\max_{w\in \Sigma^{|\mathcal{P}|\cdot |\mathcal{A}|}_{\bar{\beta}}}\min_{p_i\in\mathcal{P}}\min_{a_k\neq m(i)}\left\{w^{(i)}_{m(i)} d(\mu^{(i)}_{m(i)},x^{(i)}_{m(i),k})+w^{(i)}_k d(\mu^{(i)}_k,x^{(i)}_{m(i),k})\right\},
\end{align*}
where $\Sigma^{|\mathcal{P}|\cdot |\mathcal{A}|}_{\bar{\beta}}$ is the set of allocations satisfying $\dfrac{w^{(i)}_{m(i)}}{\sum_{k\in \mathcal{A}} w_{k}^{(i)}}=\beta_i$ for every player $p_i\in\mathcal{P}$. We can rewrite the characteristic time $D_{\bar{\beta}}(\mu)$ as: 
\[D_{\bar{\beta}}(\mu)~=~\max_{w\in\Sigma^{|\mathcal{P}|}}\min_{p_i\in\mathcal{P}}w^{(i)}\cdot D_{p_i,\beta_i}(\mu^{(i)}),\]
where $w^{(i)}=\sum_{k} w^{(i)}_k$ and
\[D_{p_i,\beta_i}(\mu^i)~=~\max_{w_i\in \Sigma^{|\mathcal{A}|}_{\beta_i}}\min_{a_k\neq m(i)}\left\{\beta_i d(\mu^i_{m(i)},x^i_{m(i),k})+w_{k|i} d(\mu^i_k,x^i_{m(i),k})\right\}\]
with the transformation of variables $w^{(i)}_k=w^{(i)}\cdot w_{k|i}$ and the constraint $w_{m(i)|i}=\beta_i$. For every player $p_i$, $w_i$ represent the proportion of the global samples allocated to player $p_i$ and $w_{k|i}$ is the proportion of samples allocated by player $p_i$ to arm $a_k$ out of the total no. of samples allocated to player $p_i$. 

With this reformulation of the lower bound problem we prove the following theorem 
\begin{theorem}\label{thm:cr_of_beta_top_two}
    $\bar{\beta}$-Top-Two policy is asymptotically $\max_{p_i}\max\left\{\frac{\beta_i^\star}{\beta_i},\frac{1-\beta_i^\star}{1-\beta_i}\right\}$-competitive with respect to the multi-player Anchored-Top-Two policy, where $\beta_i^\star$ is the optimal value of $w_{m(i)|i}$.   Therefore, upon taking $\beta_i=\frac{1}{2}$ for every player $p_i$, $1/2$-Top-Two algorithm is 2-competitive. 
\end{theorem}

\begin{proof}
    Using \cite[Theorem 1, Statement 3]{russo2020simple}, 
\[D_{p_i,\beta_i}(\mu^i)~\geq~\max\left\{\frac{\beta_i^\star}{\beta_i},\frac{1-\beta_i^\star}{1-\beta_i}\right\}\cdot D_{p_i,\beta_i^\star}(\mu^i)\quad \text{for every}~p_i\in \mathcal{P}.\]
Therefore, 
\begin{align*}
    D_{\bar{\beta}}(\mu)~&=~\max_{w\in \Sigma^{|\mathcal{P}|}}\min_{p_i} w_i \cdot D_{p_i,\beta_i}(\mu^i)\\
    &\geq~\max_{w\in \Sigma^{|\mathcal{P}|}}\min_{p_i} w_i \cdot \max\left\{\frac{\beta_i^\star}{\beta_i},\frac{1-\beta_i^\star}{1-\beta_i}\right\}\cdot D_{p_i,\beta_i^\star}(\mu^i)\\
    ~&\geq~\left(\max_{p_i\in \mathcal{P}}\max\left\{\frac{\beta_i^\star}{\beta_i},\frac{1-\beta_i^\star}{1-\beta_i}\right\}\right)\cdot\max_{w\in \Sigma^{|\mathcal{P}|}}\min_{p_i} w_i \cdot D_{p_i,\beta_i^\star}(\mu^i)\\
    ~&=~\left(\max_{p_i\in \mathcal{P}}\max\left\{\frac{\beta_i^\star}{\beta_i},\frac{1-\beta_i^\star}{1-\beta_i}\right\}\right)\cdot D_{\bar{\beta}^\star}(\mu).
\end{align*}    
\end{proof}
\clearpage
\subsection{Two-Sided Learning}

$(\bar{\alpha},\bar{\beta},\bar{\delta})$ are defined such that $\bar{\alpha}:=\{\alpha_1,\ldots,\alpha_{|\mathcal{P}|}\}, \bar{\beta}:=\{\beta^{(1)},\ldots,\beta^{(|\mathcal{P}|)}\}$ and $\bar{\delta}:=\{\delta^{(i)}_k\}_{(p_i,a_k)\in\mathcal{P}\times \mathcal{A}}$ interpreted as
\begin{align*}
    \beta^{(i)} = \dfrac{w^{(i)}_{m(i)}}{\sum_{k\in\mathcal{A}}w^{(i)}_k}, \quad \gamma_k = \dfrac{w^{(m^{-1}(k)}_k}{\sum_{i\in\mathcal{P}}w^{(i)}_k},\quad \delta^{(i)}_k = \dfrac{\sum_{k\in\mathcal{A}}w^{(i)}_k}{\sum_{k\in\mathcal{A}}w^{(i)}_k+\sum_{i\in\mathcal{P}}w^{(i)}_k}
\end{align*}
where $w^{(i)}_k$ represent proportion of samples to pair $p_i,a_k$. Define the following terms, corresponding to the different indexes



\begin{algorithm}[H]
\caption{$(\bar{\alpha},\bar{\beta},\bar{\delta})$ (Two-Sided Learning model with $M\leq K)$}
Input: preferences of arms over players $\pi$\\
    \For{$t=1,2,\ldots$}{
    $\mathcal{E}=\{(i,k):N^{(i)}_k\leq \sqrt{N}\}$\\
    \lIf{$\mathcal{E}\neq \emptyset$}{Match any $p_i$, $a_k$ from $\mathcal{E}$}
    \Else{
    $\hat{m}\leftarrow \texttt{DA}_{\texttt{Arm}}(\hat{\mu},\hat{\eta})$\\
    Construct $\mathcal{B}_1^{(i)},\mathcal{B}_2^{(i)},\mathcal{B}_3^{(i)}$ for all players $p_i$ and $\UMA_{\hat{m}}=\{k: \hat{m}^{-1}(k)=\emptyset\}$\\
    Let $\Tilde{C}^{(i)}_{\hat{m}(i),k}=\min_{i\in[M],k\in \mathcal{B}_1^{(i)}\cup \UMA_{\hat{m}}}\Tilde{C}^{(i)}_{\hat{m}(i),k}(\hat{\mu})$,\\
    \phantom{Let }$\Tilde{C}^{(i,\hat{m}^{-1}(k))}_{k}=\min_{i\in[M],k\in \mathcal{B}_2^{(i)}}\Tilde{C}^{(i,\hat{m}^{-1}(k))}_{k}(\hat{\eta})$, \\
    \phantom{Let }$ \Tilde{C}^{(i,\hat{m}^{-1}(k))}_{m(i),k}=\min_{i\in[M],m(j)\in \mathcal{B}_3^{(i)}}  \Tilde{C}^{(i,\hat{m}^{-1}(k))}_{m(i),k}(\hat{\mu},\hat{\eta})$\\

    Let $\Tilde{C}=\min\left\{\Tilde{C}^{(i)}_{\hat{m}(i),k},\Tilde{C}'_{i,\hat{m}^{-1}(k),k} \Tilde{C}''_{i,\hat{m}(i),\hat{m}^{-1}(k),k}\right\}$
     
     \lIf{$\Tilde{C}=\Tilde{C}^{(i)}_{\hat{m}(i),k}$}{Call Subroutine 1}
\lIf{$\Tilde{C}=\Tilde{C}^{(i,\hat{m}^{-1}(k))}_{k}$}{Call Subroutine 2}

\lIf{$\Tilde{C}= \Tilde{C}^{(i,\hat{m}^{-1}(k))}_{m(i),k}$}{
Call Subroutine 3
}}

    \uIf{$\texttt{DA}_{\texttt{Arm}}(\hat{\mu},\hat{\eta})=\texttt{DA}_{\texttt{Player}}(\hat{\mu},\hat{\eta})$ and $\Tilde{C}>\beta(t,\delta)$ }{Recommend the matching $\texttt{DA}_{\texttt{Arm}}(\hat{\mu},\hat{\eta})$}}
\end{algorithm}

 \setlength{\BoxH}{2.4cm}  
 \vspace{-10pt}
\begin{figure}[H]
    \centering
    \begin{minipage}[t]{0.6\textwidth}
        \centering
       {\setcounter{algocf}{1}
\SetAlgorithmName{Subroutine}{subroutine}{List of Subroutines}
\begin{minipage}[t]{\textwidth}
            \begin{minipage}[t]{0.45\textwidth}
                \begin{algorithm}[H]
                \small
                \uIf{$ \beta^{(i)}<\mathcal{U}[0,1]$}{
                \textcolor{red}{Match $p_{i}$ with $\hat{m}(i)$}
                }
                \Else{
                \textcolor{orange}{Match $p_{i}$ with $a_{k}$}
                }
                \caption*{}
                \end{algorithm}
            \end{minipage}
            \hfill
            \begin{minipage}[t]{0.5\textwidth}
               {\setcounter{algocf}{2}
\SetAlgorithmName{Subroutine}{subroutine}{List of Subroutines}
\begin{algorithm}[H]
                \small
                \uIf{$  \alpha_k<\mathcal{U}[0,1]$}{
                \textcolor{blue}{Match $\hat{m}^{-1}(k)$ with $a_k$}
                }
                \Else{
                \textcolor{orange}{Match $p_{i}$ with $a_{k}$}
                }
                \caption*{}
                \end{algorithm}}
            \end{minipage}
        \end{minipage}}
    \end{minipage}
    \hfill
    \begin{minipage}[t]{0.38\textwidth}
     {\setcounter{algocf}{3}
\SetAlgorithmName{Subroutine}{subroutine}{List of Subroutines}
        \begin{algorithm}[H]
        \small
        \uIf{$\delta^{(i)}_k<\mathcal{U}[0,1]$}
        { Follow Subroutine 1}
        \Else{
        Follow Subroutine 2
        }
        \caption*{}
        \end{algorithm}}
    \end{minipage}
\end{figure}



%% file: paper/appendix/add-exps.tex
\section{Experiments}
\label{app:add-exps}

Preference profiles used are as follows, they were randomly found such that they satisfy the specified property of preference profile.
\begin{itemize}[leftmargin=*, nosep,noitemsep, topsep=0pt]
    \item \texttt{Distinct}: this is a preference profile in which every player and every arm prefers distinct agent the most. 
    \begin{align*}
    p_1&: a_1 \succ a_2 \succ a_3 \succ a_4 \succ a_5 \quad & a_1&: p_1 \succ p_5 \succ p_3 \succ p_4 \succ p_2 \\
    p_2&: a_2 \succ a_3 \succ a_4 \succ a_5 \succ a_1 \quad & a_2&: p_2 \succ p_3 \succ p_1 \succ p_5 \succ p_4 \\
    p_3&: a_3 \succ a_4 \succ a_5 \succ a_1 \succ a_2 \quad & a_3&: p_3 \succ p_4 \succ p_2 \succ p_1 \succ p_5 \\
    p_4&: a_4 \succ a_5 \succ a_1 \succ a_2 \succ a_3 \quad & a_4&: p_4 \succ p_5 \succ p_3 \succ p_2 \succ p_1 \\
    p_5&: a_5 \succ a_1 \succ a_2 \succ a_3 \succ a_4 \quad & a_5&: p_5 \succ p_1 \succ p_4 \succ p_3 \succ p_2
\end{align*}
\item \texttt{Serial}: this is a preference profile in which every arm has same preference profile i.e. $p_1\succ \ldots p_{|\mathcal{P}|}$ and for players, the preference profile can be anything.
\begin{align*}
    p_1&: a_3 \succ a_2 \succ a_4 \succ a_1 \succ a_5 \quad & a_1&: p_1 \succ p_2 \succ p_3 \succ p_4 \succ p_5 \\
    p_2&: a_1 \succ a_3 \succ a_2 \succ a_4 \succ a_5 \quad & a_2&: p_1 \succ p_2 \succ p_3 \succ p_4 \succ p_5 \\
    p_3&: a_3 \succ a_4 \succ a_2 \succ a_5 \succ a_1 \quad & a_3&: p_1 \succ p_2 \succ p_3 \succ p_4 \succ p_5 \\
    p_4&: a_2 \succ a_5 \succ a_3 \succ a_1 \succ a_4 \quad & a_4&: p_1 \succ p_2 \succ p_3 \succ p_4 \succ p_5 \\
    p_5&: a_1 \succ a_5 \succ a_2 \succ a_3 \succ a_4 \quad & a_5&: p_1 \succ p_2 \succ p_3 \succ p_4 \succ p_5
\end{align*}
\item \texttt{SPC}: this is a preference profile such that there exists a ordering of players and arms in which $p_i$ prefers $a_i$ over $a_{i+1},\ldots,a_K$ and $a_j$ prefers $p_j$ over $p_{j+1},\ldots,p_N$. In the following example the ordering is the natural ordering i.e. $\{1,2,3,4,5\}$. 
\begin{align*}
    p_1&: a_1 \succ a_3 \succ a_2 \succ a_5 \succ a_4 \quad & a_1&: p_1 \succ p_4 \succ p_2 \succ p_3 \succ p_5 \\
    p_2&: a_1 \succ a_2 \succ a_4 \succ a_3 \succ a_5 \quad & a_2&: p_1 \succ p_2 \succ p_5 \succ p_3 \succ p_4 \\
    p_3&: a_1 \succ a_3 \succ a_5 \succ a_2 \succ a_4 \quad & a_3&: p_1 \succ p_3 \succ p_2 \succ p_4 \succ p_5 \\
    p_4&: a_2 \succ a_4 \succ a_5 \succ a_1 \succ a_3 \quad & a_4&: p_2 \succ p_4 \succ p_5 \succ p_1 \succ p_3 \\
    p_5&: a_1 \succ a_5 \succ a_4 \succ a_2 \succ a_3 \quad & a_5&: p_3 \succ p_1 \succ p_2 \succ p_4 \succ p_5
\end{align*}
\end{itemize}
As discussed in the main paper, the $\mu,\eta$ matrices are generated such that each $\mu^{(i)}:=\{\mu^{(i)}_1,\ldots,\mu^{(i)}_{|\mathcal{A}|}\}$ and $\eta_k:=\{\eta^{(1)}_k,\ldots,\eta^{(|\mathcal{P}|)}_{k}\}$ belong to $\{2,2.5,3.5,5,7\}$  and are ordered according to the above preference profiles.

\begin{figure}[H]
    \centering
    \begin{subfigure}[t]{0.3\textwidth}
        \includegraphics[width=\linewidth]{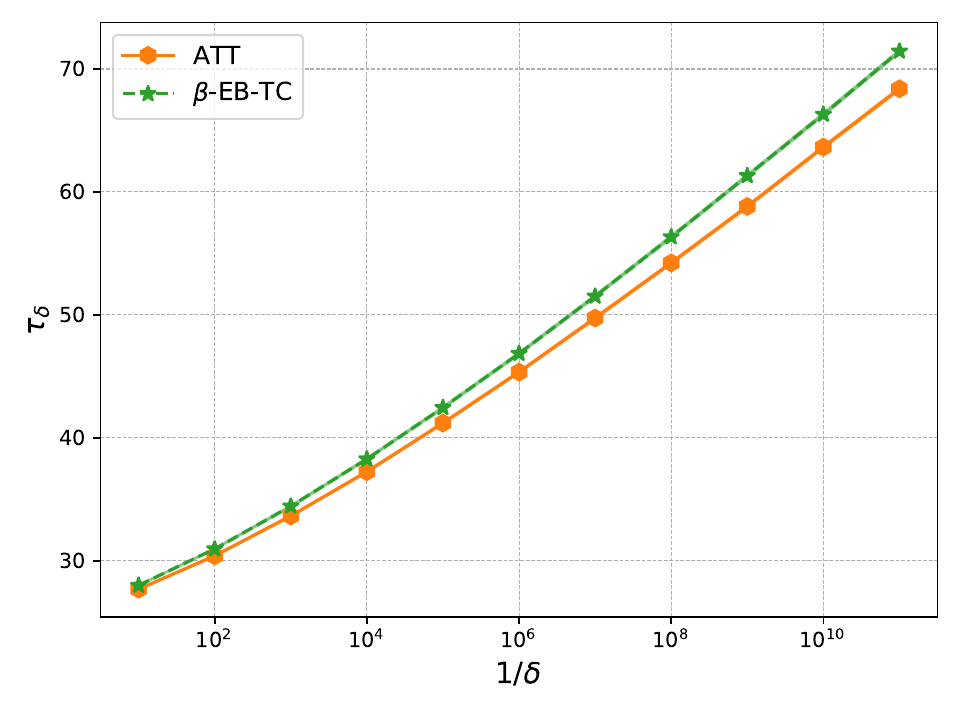}
        \caption{\texttt{Distinct}}
    \end{subfigure}%
    \hfill
    \begin{subfigure}[t]{0.3\textwidth}
        \includegraphics[width=\linewidth]{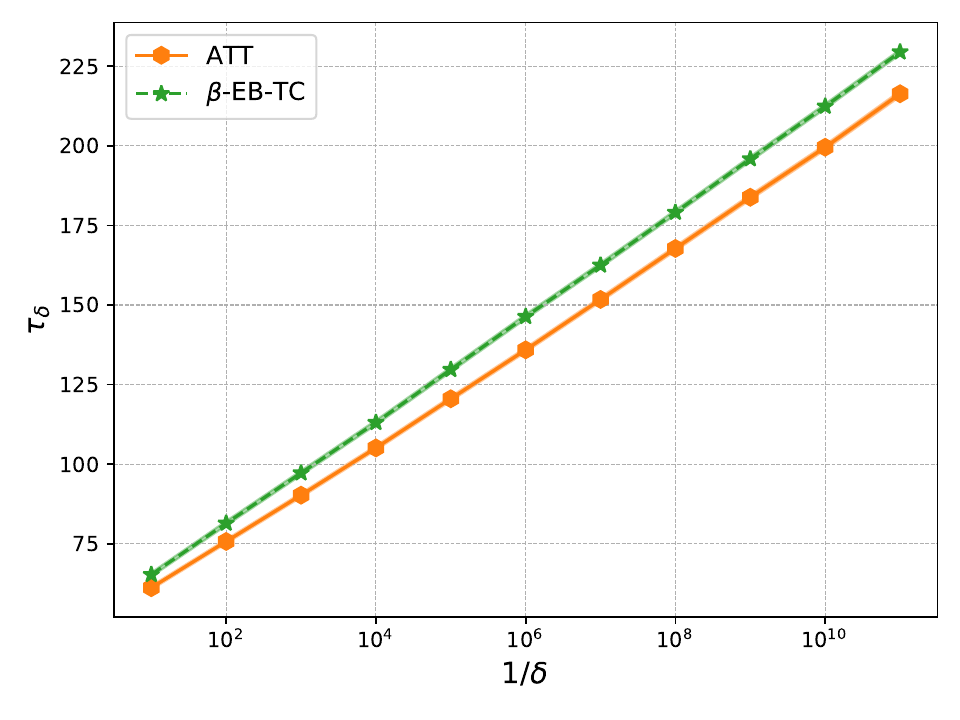}
        \caption{\texttt{Serial Dictorship}}
    \end{subfigure}%
    \hfill
    \begin{subfigure}[t]{0.3\textwidth}
        \includegraphics[width=\linewidth]{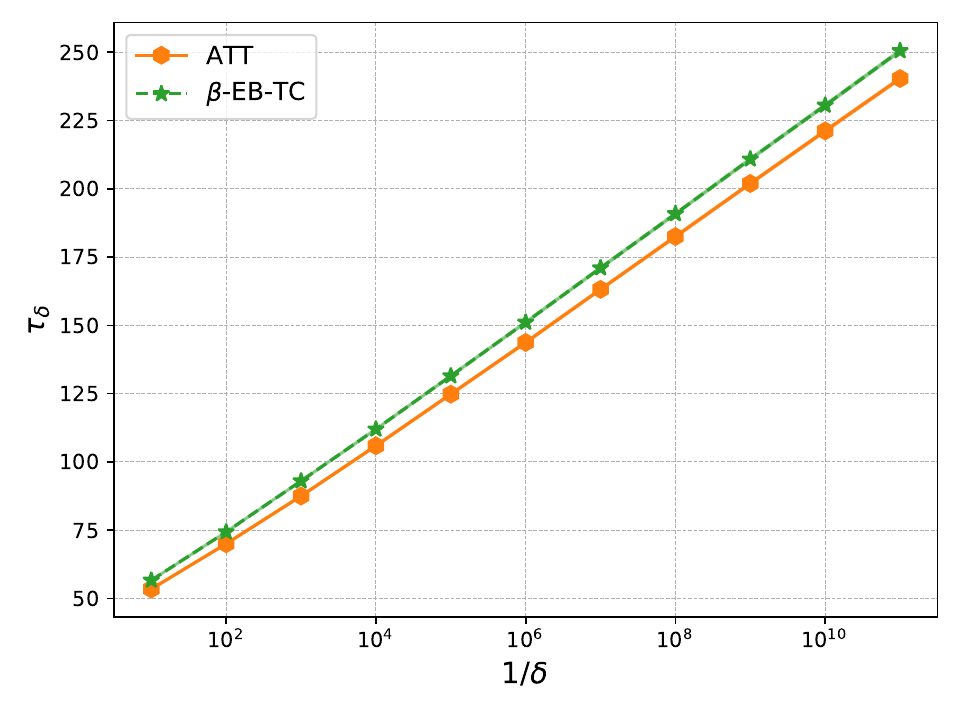}
        \caption{\texttt{SPC}}
    \end{subfigure}
    \caption{Sample complexity $\tau_{\delta}$ v/s $1/\delta$ for Two-sided Learning}
    \label{fig:sample/delta}
\end{figure}

In Fig.~\ref{fig:sample/delta} we aim to verify the asymptotic optimality of \texttt{ATT2} for two-sided learning by plotting the stopping time averaged over 5000 runs vs $1/\delta$ on $\log$-scale. We consider three preference profiles as described above and keep $\gamma=0.0$ with a smaller threshold $\beta=\log((1+\log t)/\delta)$. This verifies that the constant i.e. $\E[\tau_{\delta}]/\log(1/\delta)$ for $\beta$-algorithms are larger.

%% file: paper/appendix/multistable.tex
\section{Multiple Stable Matching}
\label{app:multistable}

For multiple-stable matching, there can be the following two main objectives:
\begin{enumerate}[leftmargin=*, nosep,noitemsep, topsep=0pt]
    \item Announced matching should be stable with probability at least $1-\delta$ i.e. $\mathbb{P}(\hat{m}_{\tau_{\delta}}\not\in\mathcal{M}_{\zeta})\leq \delta$
    \item Announced matching should be equal to a specific stable matching $m_{\zeta}$ w.p. at least $1-\delta$ i.e $\mathbb{P}(\hat{m}_{\tau_{\delta}}\neq m_{\zeta})\leq \delta$, where $m_{\zeta}$ can be the player-optimal, arm-optimal or some fair matching.
\end{enumerate}
We highlight the difficulty in both the objectives, which lies mainly in the difficulty in simplifying the lower bound. \\
The possible alternative instance for objective 1. is $\texttt{Alt}_1:=\{\lambda:\mathcal{M}_{\zeta}\subset \mathcal{M}_{\lambda}\}$, equivalently, this says that $\exists$ is a matching which is stable in $\lambda$ but not in $\zeta$. Recall that the lower bound requires finding the instance $\lambda$ ``closest'' to the original instance $\zeta$. This can be closely related to the objective of multiple correct answers \cite{degenne2019pure}, which requires finding any of the correct answers with high confidence. For example, the objective of the \textit{Any-low} arm requires finding the index of any arm with mean less than the pre-specified threshold, if it exists, and answering no, it does not. The lower bound is now given as follows.
\begin{theorem}
    Any $\delta$-correct algorithm satisfies 
    \begin{align*}
        \lim\inf_{\delta\to 0}\dfrac{\E_{\zeta}[\tau_{\delta}]}{\log(1/\delta)}\geq T^*(\zeta):=D(\zeta)^{-1}\ \ \ \text{where}\ D(\zeta) = \max_{m\in \mathcal{M}_{\zeta}}\max_{w\in\Delta_{M\times K}}\inf_{\lambda\in\neg m}\sum_{i=1}^M\sum_{k=1}^K w^{(i)}_k D\left(\zeta^{(i)}_k,\lambda^{(i)}_k\right)
    \end{align*}
    where $\neg m:=\{\lambda:m\not\in\mathcal{M}_{\lambda}\}$, and $\Delta_{M\times K}=\{w:\sum_{i=1}^M\sum_{k=1}^Kw^{(i)}_k=1\}$ where $w^{(i)}_k\geq 0$ is the proportion of samples given to the pair $(p_i,a_k)$.
\end{theorem}
Although, simplifying the objective for specific $m\in\mathcal{M}_{\zeta}$ is easy as it requires optimizing over instances which induce blocking pair in the matching $m$. However, finding such a matching $m$, which maximizes the objective, is combinatorially difficult as it requires traversing the entire stable matching lattice of $\zeta$.

For Objective 2, the alternate instance is complex. If $m_{\zeta}$ in the objective is player-optimal, then the alternate instance can be designed such that $m_{\zeta}$ is not the player-optimal stable match. This can be further decomposed into two alternate sets $\texttt{Alt}_{\stable}$ which are instances in which $m_{\zeta}$ is not stable and $\texttt{Alt}_{\texttt{unstable}}$ instances in which $m_{\zeta}$ is stable but is not the player-optimal stable matching. Optimizing the lower bound over $\texttt{Alt}_{\stable}$ is objective 1. However, the optimizing instance over $\texttt{Alt}_{\texttt{unstable}}$ is harder to find, as it also requires optimizing over the entire stable matching lattice of $\zeta$.

Note that the Track-and-Stop algorithms from \cite{kaufmann2020contributions} can still be applied; however, the computation requirement can be exponential. 
Recall that our primary objective is to devise \textit{computationally efficient} algorithms, however, unless we restrict $\zeta$ to have specific property with knowledge to the decision maker, the problem seems impossible due to the combinatorial nature in the lower bound optimization problem. Thus, one may want an approximate solution, possibly with a different objective than $\delta$-correctness. 